\documentclass[11pt]{article}
\usepackage{graphicx}
\usepackage{pgfplots}
\usepackage[dvipsnames]{xcolor}
\usepackage[colorlinks=true,citecolor=blue]{hyperref}
\usepackage{natbib}
\usepackage{graphicx}
\usepackage{amsfonts}
\usepackage{amsmath}
\usepackage{amssymb}
\usepackage{url}
\usepackage{fancyhdr}
\usepackage{indentfirst}
\usepackage{enumerate}
\usepackage{titlesec}
\usepackage{amsthm}
\usepackage{dsfont}
\usepackage[misc]{ifsym}
\usepackage{enumitem}

\theoremstyle{definition}

\newtheorem{example}{Example}

\theoremstyle{plain}
\newtheorem{theorem}{Theorem}
\newtheorem{lemma}{Lemma}
\newtheorem{proposition}{Proposition}

\theoremstyle{remark}

\usepackage{cases}

\def\laweq{\buildrel \mathrm{d} \over =}
\def\lawis{\buildrel \mathrm{d} \over \sim}

\theoremstyle{definition}

\def\d{\mathrm{d}}
\DeclareMathOperator*{\esssup}{ess\text{-}sup}
\DeclareMathOperator*{\essinf}{ess\text{-}inf}

\newcommand{\R}{\mathbb{R}}
\newcommand{\p}{\mathbb{P}}

\usepackage[onehalfspacing]{setspace}

\usepackage{bm}
\usepackage{tikz-qtree}
\usepackage{tikz}
\usepackage{float}

\floatstyle{plaintop}
\newfloat{algofloat}{tbp}{loa}
\floatname{algofloat}{Algorithm}

\usepackage{caption}
\usepackage[ruled,vlined]{algorithm2e}

\def\id{\mathds{1}}

\pgfplotsset{compat=1.18}

\title{
Probability of worthwhile effect of monotone-response treatments}
\author{Benjamin C\^ot\'e\thanks{Department of Statistics and Actuarial Science, University of Waterloo, Canada.   \texttt{b3cote@uwaterloo.ca}}\and  Ruodu Wang\thanks{Department of Statistics and Actuarial Science, University of Waterloo, Canada.   \texttt{wang@uwaterloo.ca}} }
\date{\today}

\begin{document}

	\maketitle

\begin{abstract}  
Experiments may, by design, prevent one from observing on a single subject both the response to a treatment and to its absence. Because of this, marginal distributions for both cases may be observable but not their joint distribution, thus obscuring the distribution of the treatment effect. We examine the case where we impose that the treatment effect is nonnegative, also called  monotone treatment response, a common assumption relevant to many practical applications. We solve the problems of best- and worst-case probabilities that the treatment effect exceeds a given value, using an explicit construction for the dependence scheme in each case. %We also discuss best- and worst-case risk measurements of the treatment effect. 
Such problems can equivalently be described, in different contexts,  as risk aggregation under dependence uncertainty and an order constraint, and as optimal transport  with a particular cost function. 


\end{abstract}

\section{Introduction}
\label{sect:intro}

In a simple experiment, we study the effect of a treatment. The responses of subjects to the treatment, $Y$, and to the absence of treatment, $X$,  are monitored. Responses $X$ and $Y$ are random variables; their distributions, respectively $\mu$ and $\nu$, may either be interpreted as that of the (deterministic) treatment response within a population or of the aleatoric response of any single subject. The treatment effect is $Y-X$.

Most considerations on the treatment effect are answered by elementary statistics if we are able to observe both responses for every subject.
However, this is often not possible for experimental studies outside of controlled environments or conducted over a long time frame, or sometimes due to ethical considerations.  
Consider, for instance, a study on the effects of malnutrition on children's health: a researcher cannot require subjects to be malnourished for the length of the study. 
The researcher must rather rely on observational data, meaning that data on responses to treatment and non-treatment are collected separately. 
This leads to two major impediments for statistical purposes: 
\begin{enumerate}
    \item[(S1)] The observed marginal distributions, from each group's data, may not be the true marginal distributions of $X$ and $Y$, as some biases may arise, e.g.,~from selection.  
    \item[(S2)] The separate collection of data obscures the coupling between $X$ and $Y$, and hence also the distribution of the treatment effect $Y-X$.
\end{enumerate}  Let us provide a simple example to illustrate this second impediment. 

% Marginal distributions for $X$ and $Y$ may be estimated from each group's data.\footnote{In practice, some biases may arise, see \cite{WM99}, but we suppose this is not the case here.} However, the experiment's design leaves us blind to the joint distribution of $(X,Y)$, hence befogging the distribution of the treatment effect $Y-X$.  
% One may also think of it as an experience in which the names of the subjects have been redacted. Let us provide a simple example to illustrate how such designs obstruct knowledge of the joint distribution of $(X,Y)$.  
\begin{example}
	\label{ex:design}
	For an experiment, we have the following samples respectively without and with treatment: $\mathbf{x}=\{1,2,3,4\}$ and $\mathbf{y}=\{2,3,4,5\}$. One may think that the treatment has had an effect of 1 for every subject, thus supposing the coupling $ \{(1,2), (2,3), (3,4), (4,5)\}$; another may rather think that the treatment was only effective on one subject, with an effect of 4, thus supposing the coupling  $ \{(1,5), (2,2), (3,3), (4,4)\}$. Without additional information on the relation between $\mathbf{x}$ and $\mathbf{y}$, we cannot distinguish either coupling from the data. 
\end{example}

Experimental designs based on observational data are prominent in many fields, such as sociology, statistics and econometrics; and their study led to the \textit{counterfactual} conception of causal effects.  
 One may consult 
 \cite{BP94},
 \cite{WM99}, \cite{D00}, \cite{S09}, and \cite{D15} for an introduction
 %, an overview of the literature, 
 and some (philosophical, statistical) considerations on counterfactual  causality.

As underscored in all of these works, studies involving counterfactual causality will almost always examine the average treatment effect, $\mathbb{E}[Y-X]$. This has the marked advantage of dispelling impediment (S2) because %the value does not depend on the joint distribution of $(X,Y)$, only on their marginal distributions. 
$\mathbb{E}[Y-X]$ is unaffected by the dependence between $Y$ and $X$.
Most seminal statistical works on counterfactual causality thus exclusively focused on mitigating (S1); see notably the series of works \citep{M94,M95, M97, M07, MP00, MP09}. 

Average treatment effect,
however, may be ill-suited to some situations. 
In Example~\ref{ex:design}, both couplings yield the same average effect, but may lead to different decisions. If, for example, treatments are antidepressants and responses are their side effect on blood pressure, the first coupling, with a constant~1 effect, would probably lead to the drug being prescribed by a physician, as they may judge that the side effect is tolerable. The second coupling, however, would probably incite physicians to avoid prescribing the drug, as they deem too risky to potentially see a dramatic increase of blood pressure. 
See \cite{F99} and \cite{YG24} for further discussions on the limitations of average treatment effect in medical trials. Let us present other examples illustrating this rhetoric.

\begin{example}
\label{ex:tutor}
A tutoring service believes that parents will renew their subscription if their children perform at least 10 points better on their tests. The tutoring service wants to obtain bounds on the number of renewals, given the tutored students' grades compared to the non-tutored students.
\end{example}

\begin{example}
\label{ex:smoking}
Ex-smokers will continue to abstain from smoking only if they see an increase in their overall health. Public health services want to assess for the worst case of smoking relapse given health data.
\end{example}

We remark that the analytics under study in these examples do not depend on the actual responses $X$ or $Y$ but solely on the treatment effect $Y-X$. However, the average treatment effect is unable to answer the threshold-related considerations. 

For such considerations, one may rather consider the probability of a worthwhile effect; see, e.g., \cite{F99}. Let $k\geq 0$ be the level above which a treatment effect is considered clinically relevant. %; this level $k$ is called the \textit{minimum worthwhile effect}. 
 The probability of worthwhile effect 
for a given $k$ is the value $\mathbb{P}(Y-X > k)$.  
\cite{F99} discussed probabilities of worthwhile effect as an alternative to p-values in the context of medical treatments. An issue for p-values is that statistical significance of the treatment effect is unrelated to whether this effect is clinically worthwhile. \cite{F99} argued that a medical researcher should not heuristically determine what constitutes a worthwhile effect based on the collected data; probabilities of worthwhile effect effectively require the researcher to specify beforehand the threshold for worthwhile effect $k$.  
For instance, in our earlier example, the physician would establish which threshold $k$ of increase in blood pressure is dangerous and then inquire about the probability of worthwhile effect for $k$. In Examples~\ref{ex:tutor} and~\ref{ex:smoking}, thresholds for worthwhile effects are specified as $k=10$ and $k=0$ respectively.

When examining probabilities of worthwhile effect instead of the average treatment effect, impediment (S2) comes back: the value $\mathbb{P}(Y-X > k)$ is impacted by the dependence relation between $X$ and $Y$. Since there is already a long string of literature on how to mitigate (S1), as we specified above, we will focus on (S2) and assume that the treatment has fully identified marginals, to remove (S1). %This assumption is sometimes called exogenous treatment selection, or randomization of trials.

Because of (S2), probabilities of worthwhile effect cannot be computed unless we specify the dependence relation  between $X$ and $Y$. One common assumption is to suppose that $X$ and $Y$ are comonotonic.\footnote{Two random variables $Z,W$ are comonotonic if $(Z(\omega) - Z(\omega^{\prime}))(W(\omega)- W(\omega^{\prime}))\geq 0$ for all $\omega, \omega^{\prime}\in\Omega$.} If $\nu$ is a translation of $\mu$, then this is called the Treatment-Unit Additivity assumption \citep[Section 4.2]{D15}, which effectively means that $Y-X$ is degenerate. Another common assumption is to suppose that $X$ and $Y$ are independent \citep{D00}. Letting $k=0$, the probability of worthwhile effect in that case corresponds, up to a small adjustment for ties, to the Mann-Whitney parameter, used in the Wilcoxon-Mann-Whitney test for causal effects \citep{W45,MW47} and later introduced by \cite{D16} under the name \textit{D-value}. Not only are independence and comonotonicity very different dependence schemes, they also mean very strong assumptions on the treatment's effect. A researcher may well not want to commit to either. In particular, the independence assumption is heavily criticized by \cite{GFBSFGR20}, and Hand's paradox \citep{H92} is an example of how misleading it can be.  
As \cite{GFBSFGR20} note, to rely on such assumptions for causal inferences marks an obliviousness to impediment (S2), perhaps forgotten because the average treatment effect bypasses it. Consistent with such criticism, our approach will rather consist of deriving worst- and best-case probabilities of worthwhile effect under the full uncertainty stemming from (S2).

In this paper, we will study probabilities of worthwhile effect of treatments with the additional assumption of a monotone response. Monotone treatment response, first studied by \cite{M97}, stipulates that responses weakly increase with the strength of treatment. 
For a binary treatment, this means $Y\geq X$.\footnote{Such relations for random variables are to be interpreted as holding almost surely.} Put equivalently: the treatment effect 
$Y-X$ is nonnegative almost surely, meaning that the treatment cannot decrease the subject's response.%\footnote{Words such as  \textit{increasing} or \textit{convex} are understood in their non-strict sense throughout. }

This assumption may seem strong, but is far from unreasonable in many contexts, for instance in a study of
the effects of private tutoring on academic success \citep{H14}.
Other examples include 
the effects of schooling on elders' cognition \citep{ABFFFKKS25}, 
 of a smaller health-insurance deductible on the number of doctor visits \citep{GS06},
 of moral suasion on tax evasion \citep{BCST20}, 
 of food assistance programs on food security \citep{KPR16, GKP17, C22},
and of food security on children's health \citep{GK09, GKP12}.
%Under the assumption that there are no selection bias so that the marginal distribution of the treated group truly represents that of the untreated group if they had been treated. 
 In Example~\ref{ex:design}, both couplings satisfy the assumption of monotone treatment response $ y_i \geq x_i$ for every element $i$. Thus, it is clear that, while limiting the number of possible couplings,  the monotone-treatment-response assumption alone in general does not suffice to fully identify the joint distribution of responses.

%To our knowledge, no research involving monotone treatment response has examined probabilities of worthwhile effect to assess treatment effect (for example, none of the aforementioned papers). We attribute this lack in part to the  absence of proper methods to mitigate (S2), which this article aims to remedy. 
%Many researchers would assess the presence of a causal effect using the average treatment effect; 

%A first consideration is that of effectiveness of the treatment, the probability that a patient experiences sensible effects from the treatment.

We examine worst- and best-case probabilities of worthwhile effect respectively defined as
\begin{align}
	Q^{\inf}_k(\mu,\nu) &= \inf\{\mathbb{P}(Y-X > k): X\lawis \mu,\, Y\lawis  \nu, \, Y\geq X  \};
	\label{eq:monotoneeffectiveness-inf}\\
	Q^{\sup}_k(\mu, \nu) &= \sup\{\mathbb{P}(Y-X > k): X\lawis \mu,\, Y \lawis  \nu, \, Y\geq X  \},
	\label{eq:monotoneeffectiveness-sup}
\end{align}
where $X\lawis \mu$ means that the distribution of $X$ is $\mu$.

We do not make further assumptions than full identification of marginals and monotone treatment response; in particular, we do not impose any shape or specific features for the marginal distributions, as long as they allow for monotone treatment response to be possible. The set of possible joint distributions given this constraint is studied in \cite{AMZ20}. %A non-sharp bound for \eqref{eq:monotoneeffectiveness-inf}  is given in \citet[Corollary 1]{K14}.

In the remainder of this section, we examine problems \eqref{eq:monotoneeffectiveness-inf}--\eqref{eq:monotoneeffectiveness-sup} from different perspectives:  of risk aggregation and of optimal transport; we also introduce some notation. In Section~\ref{sect:EFFECTIVENESS}, we solve the problems for binary treatments. %In Section~\ref{sect:EXOGENOUS}, we discuss how to address partial identification of the joint distribution coming from the observation of exogenous variables. 
In Section~\ref{sect:NONBINARY}, we solve the problems for non-binary treatments, where intermediate degrees of treatment constitute a source of partial identification of the joint distribution for the full treatment, from the monotone-treatment-response assumption.
Section~\ref{sect:CONCLUSION} concludes.

\subsection{Formulation in optimal transport}

In optimal transport theory, the primal Kantorovich problem consists of minimizing an expected cost over all joint distributions with given marginals. See \cite{RR06}, \cite{V09} and \cite{F24} for an introduction to optimal transport. We rewrite problems \eqref{eq:monotoneeffectiveness-inf}--\eqref{eq:monotoneeffectiveness-sup} in terms of optimal transport:
\begin{align}
   Q^{\inf}_k(\mu,\nu) &= \inf\{\mathbb{E}[c(X,Y)]: X\lawis \mu, Y\lawis \nu \}, ~~~\text{where }c(x,y) = \id_{\{y-x > k\}} + \infty\id_{\{y<x\}};\label{eq:nutz-inf}
\\   Q^{\sup}_k(\mu,\nu) &=1- \inf\{\mathbb{E}[c^*(X,Y)]: X\lawis \mu, Y\lawis \nu \}, ~~~\text{where }c^*(x,y) = \id_{\{y-x \leq k\}} + \infty\id_{\{y<x\}}. \label{eq:nutz-sup}
\end{align}

Note how the second terms of the cost functions $c$ and $c^*$ enforce the ordering constraint from the monotone-treatment-response assumption.
For atomic distributions (see Section~\ref{sect:atomic}), one may replace $\infty$ by a finite number that is large enough. 
Optimal transport with this ordering constraint is studied by \cite{NW22} under the name \textit{directional optimal transport}. They solve problems  similar to \eqref{eq:nutz-inf}--\eqref{eq:nutz-sup}, but with the first term of $c$ or $c^*$ being submodular or supermodular.\footnote{A function $(x,y)\mapsto f(x,y)$ is submodular if $f(x_1,y_1) + f(x_2,y_2) \leq f(x_1\wedge x_2, y_1 \wedge y_2) + f(x_1\vee x_2, y_1 \vee y_2)$ for all $(x_1,y_1),(x_2,y_2)$ in its domain. It is supermodular if the reverse inequality holds for all $(x_1,y_1),(x_2,y_2)$ in its domain.} Here, the functions $(x,y)\mapsto \id_{\{y-x > k\}}$ and $(x,y)\mapsto \id_{\{y-x \leq k\}}$ are neither submodular nor supermodular. 
Although it will not be our approach to solving them, strong duality holds for problems \eqref{eq:nutz-inf}--\eqref{eq:nutz-sup}; see Appendix~\ref{sect:duality}. %Duality of \eqref{eq:nutz-inf}--\eqref{eq:nutz-sup} is discussed in \cite{K14}, and \cite{JLS23}'s method is applicable. Although, both approaches provide no information on the coupling solving these. This information on the couplings is a main motivation in solving the primal problems directly, and many of our results are aimed toward this objective in particular. 
%The choice of a strict inequality in \eqref{eq:monotoneeffectiveness-inf}, and of a weak inequality in \eqref{eq:monotoneeffectiveness-sup}, is to let $c$ and $c^*$ be lower semicontinuous.

\subsection{Formulation in risk aggregation under dependence uncertainty}

In the risk-management literature, the framework obtained by dispelling (S1) but retaining (S2) is called \textit{dependence uncertainty}: we suppose that the marginal distributions of $(X,Y)$ are known, but the dependence relation between its components is unknown. See \cite{EP10} for considerations and  motivations for such a framework in the context of risk management. 
\cite{CLW22} proposed and solved the following problems:
\begin{align}
	M^{\inf}_k(\mu,\nu) &= \inf\{\mathbb{P}(X+Y > k):X\lawis \mu , \,  Y\lawis \nu,\,  Y\geq  X  \};
	\label{eq:chenminus-inf}\\
	M^{\sup}_k(\mu,\nu) &= \sup\{\mathbb{P}(X+Y \geq k): X\lawis \mu , \, Y\lawis \nu,  \, Y\geq  X  \},
	\label{eq:chenminus-sup}
\end{align}
Note that the only difference between their problems
and  ours, beyond the weak inequality in \eqref{eq:chenminus-sup}, is that $X+Y$ in \eqref{eq:chenminus-inf}--\eqref{eq:chenminus-sup} takes the place of $Y-X$ in \eqref{eq:monotoneeffectiveness-inf}--\eqref{eq:monotoneeffectiveness-sup}. 
   This difference  has a significant impact on the optimal couplings, making the problems \eqref{eq:monotoneeffectiveness-inf}--\eqref{eq:monotoneeffectiveness-sup} 
   completely different from \eqref{eq:chenminus-inf}--\eqref{eq:chenminus-sup}; in particular, the results of \cite{CLW22} cannot be applied to our problems. 
  We discuss their connections in Appendix~\ref{sect:CHEN}. 
%By letting $Z = -X$, 
%our problems \eqref{eq:monotoneeffectiveness-inf}--\eqref{eq:monotoneeffectiveness-sup}
%can be formulated with  
%$\p(Y+Z\ge k)$
%with given marginals of $Y$ and $Z$, 
%under the constraint $Y\ge -Z$.

Without the constraint of monotone treatment response, the literature on risk aggregation under dependence uncertainty is rich and extends beyond the case of two variables: most notably, \cite{R82} gives the solution to the worst-case $\mathbb{P}(X+Y>k)$ and best-case $\mathbb{P}(X+Y\geq k)$. 
\cite{BJW14} and \cite{MW15} examine the set of all possible distributions of the sum of several random variables with given marginal distributions; bounds on best- and worst-case values of the quantile of the sum are given in \cite{EWW15} and  algorithms to numerically compute these values are developed in \cite{EPR13} and \cite{BLLW24}.

% \subsection{The problem in terms of causality in statistical inference}

% The average treatment effect is $\beta$ in regression. 

\section{Preliminaries and notation}
\label{sect:prelim}

All random variables live on an atomless probability space $(\Omega, \mathcal{F}, \mathbb{P})$. All order terms like ``increasing" are in the weak sense.  
Write $\R_+=[0,\infty)$. 
Let $\mathcal{M}$ be the set of all finite Borel measures on $\R$, equipped with the usual order $\le$. Note that $(\mathcal{M},\leq)$ is a lattice; $\wedge$ and $\vee$ respectively mean the infimum and supremum operators on that lattice. 
The probability measure $\delta_x$ is the point-mass at $x\in \R$. 
%\textcolor{magenta}{For two measures $\mu,\nu \in \mathcal M$, their common part is $\mu \wedge \nu  = \sup\{\theta\in \mathcal M:\theta\le \mu~\mbox{and}~\theta\le \nu\}$, which is well defined because $\mathcal M$ is a lattice.}  
The survival function of a measure $\mu\in\mathcal{M}$ is the function $S_{\mu} : \mathbb{R} \to \R_+ $ given by $S_\mu(x)= \mu((x,\infty))$.  
We write $\mu \le_{\rm st} \nu$
if $S_{\mu}\le S_{\nu}$. An equivalent property to $\mu \le_{\rm st} \nu$  is 
\begin{equation*}
\int f \d \mu \le \int f \d \nu
 \mbox{ for all increasing functions $f:\R\to \R_+$.}
\end{equation*}
Note that $\mu\le \nu$ implies $\mu\le_{\rm st } \nu$, and both relations are transitive. 
The quantile function of a measure $\mu\in\mathcal{M}$ is the function $F^{-1}_\mu : (0,\mu(\mathbb{R})) \to \mathbb{R}$ defined as
$
    F^{-1}_{\mu}(u) = \inf\{x\in\mathbb{R}: \mu(\R)-S_{\mu}(x) \geq u\},$ $u\in(0,1).$
For a random variable $Z$, its distribution  is the probability measure $\mu\in \mathcal M$ given by $\mu(A)=\p(Z\in A)$ for Borel $A\subseteq \R$, and we denote this by $Z\lawis \mu$.
While the exposition on treatment effects in the Introduction considered only probability measures (instead of general measures), we do not necessarily make this assumption henceforth and allow for non-probability measures as well in most of our results.

Let $\mathbb H=\{(x,y)\in \R^2: y \ge x\}$. 
% For two measures $\mu,\nu \in \mathcal M$,
% we denote by  
% $ \mathcal{H}(\mu,\nu)$  
%  the set  of all Borel measures $\pi$ on $\R^2$ supported on $\mathbb H$ with marginals $\mu$ and $\nu$.
Let $\Pi$ be the set of all Borel measures on $\R^2$.
For $\pi\in \Pi$, we let $P_1 (\pi)$ denote
the first marginal of $\pi$
and 
 $P_2( \pi)$ the second marginal of $\pi$. 
For $\mu,\nu \in \mathcal M$, 
define the set 
\begin{equation}
     \mathcal{H}(\mu,\nu) = \{ \pi \in \Pi : \mathrm{supp}(\pi)\subseteq \mathbb{H},~ P_1 (\pi) = \mu ,~ P_2 (\pi) \leq \nu  \} ,
     \label{eq:semicoupling}
\end{equation}
which is a set of \emph{semi-couplings} between $\mu$ and $\nu$ in the sense of \cite{HS13}. 
For simplicity, we refer to any element in $     \mathcal{H}(\mu,\nu)$ as a \emph{coupling} between $\mu$ and $\nu$. 
For $ \mathcal{H}(\mu,\nu)$   to be nonempty, it is necessary and sufficient that $\mu\le_{\rm st}\nu$;  see   \citet[Theorem 2.6.3]{MS02}. 
Note that $\pi(\R^2) = \mu(\R)\le \nu(\R)$ for all $\pi\in \mathcal H(\mu,\nu)$ when $\mu \le_{\rm st} \nu$.  \cite{AMZ20} and \cite{NW22} analyzed $ \mathcal{H}(\mu,\nu)$   when $\mu$ and $\nu$ are probability measures.

\section{Probabilities of worthwhile effect for binary treatments}
\label{sect:EFFECTIVENESS}

For measures $\mu$ and $\nu$ with $\mu\le_{\rm st} \nu$  that are not necessarily probability measures, we generalize problems \eqref{eq:monotoneeffectiveness-inf}--\eqref{eq:monotoneeffectiveness-sup} in the Introduction  by rewriting them as 
\begin{align}
    Q_k^{\inf}(\mu,\nu) &= \inf\left\{\int \id_{\{(x,y)\in\mathbb{R}^2: y-x>k\}}\mathrm{d}\pi : \pi\in \mathcal{H}(\mu,\nu) \right\}; \label{eq:submeasures-inf} \\
    Q_k^{\sup}(\mu,\nu) &= \sup\left\{\int \id_{\{(x,y)\in\mathbb{R}^2: y-x>k\}}\mathrm{d}\pi : \pi\in \mathcal{H}(\mu,\nu) \right\}. \label{eq:submeasures-sup}
\end{align}

 Our solutions to problems \eqref{eq:submeasures-inf}--\eqref{eq:submeasures-sup} are predicated on coupling algorithms. In this section, we present these algorithms, solving the problems for specific discrete distributions, and discuss the adequacy of approximating the solutions for the general discrete case and the continuous case this way. Although problems \eqref{eq:submeasures-inf} and \eqref{eq:submeasures-sup} look very similar, the solution to one does not follow directly from the other's by symmetry. %The difference between the two problems will become increasingly apparent throughout this section. 

For $\mu$ and $\nu$ that are probability measures with finite means and $\mu\leq_{\rm st}\nu$,
crude bounds for $Q^{\inf}_k(\mu,\nu)$ and $Q^{\sup}_k(\mu,\nu)$ are,  
\begin{equation*}
    0 \leq Q^{\inf}_k(\mu,\nu) \leq Q^{\sup}_k(\mu,\nu) \leq \frac{1}{k}\left(\int x \, \nu (\mathrm{d} x)   - \int x  \, \mu (\mathrm{d} x)  \right), \quad \quad k> 0,
\end{equation*}
where the last inequality follows from Markov's inequality. The algorithms presented below leverage the identification of $\mu$ and $\nu$ to produce improved bounds, tailored to the specified marginal distributions.

\subsection{Solutions for atomic measures with specific atom size }
\label{sect:atomic}

% In Algorithm~\ref{algo:A}, we explicitly construct the optimal coupling solving the problem, and then compute the infimal probability of clinically relevant benefit; the coupling is $\textbf{(x,y)} = \{(x_i, y_{d_i}) : i\in[n]\}$.

We first study a special case of discrete measures $\mu$ and $\nu.$
A measure $\mu$ is \textit{atomic with size $a>0$} if 
\begin{equation}
\label{eq:mu-a-m}
\mu= a \sum_{i\in \{1,\ldots,m\}}\delta_{x_i},
\end{equation}
for some $x_1, \ldots, x_m \in\mathbb{R}$.
When $\mu  $ has total mass $1$ (that is, $a=1/m$),  $\mu$ is simply an empirical distribution.
We call $x_1, \ldots, x_m$ the \textit{locations} of atoms, $m$ the number of atoms, and $a$  the size of an atom; note that multiple atoms may be at identical locations. 
For atomic measures $\mu$ and $\nu$ with the same size (meaning that they share the same $a>0$), 
if $\mu\le_{\rm st}\nu$ holds, then $m\le n$, where $m$ is the number of atoms in $\mu$ and $n$ is the number of atoms in $\nu$.

%Let us present Algorithm~\ref{algo:A}.

%\begin{proposition}
%	In the problem \eqref{eq:submeasures-sup}, the supremum is attained, meaning \eqref{eq:submeasures-sup} can be rewritten as 
%	$\max\{\mathbb{P}(X-Y\geq k): (X,Y)\in\mathcal{H}\}$.
%\end{proposition}
%\begin{proof}
%	co
%\end{proof}

As our first main result, the minimization in \eqref{eq:submeasures-inf} is attained by
Algorithm~\ref{algo:A},
and the  maximization  in \eqref{eq:submeasures-sup} is attained by
Algorithm~\ref{algo:B}.
We explain the intuition behind  these algorithms first.  Throughout, for ease of exposition, let $x$ denote a generic atom location in $\mu$ and $y$ denote a generic atom location in $\nu$.

For the construction of the pairs $(x,y)$ in the optimal $\pi$ for \eqref{eq:submeasures-inf}, we couple atoms in $\mu$ iteratively, starting from largest location to smallest location. For each atom, when its turn to be coupled comes, we check whether there is an uncoupled atom in $\nu$ allowing to satisfy $y-x\leq k$. If so, it couples to the one at the largest location among them.  If it is unable to satisfy $y-x\leq k$, it sacrifices itself and takes down the atom in $\nu$ at the largest location; this effectively helps the next-to-be-coupled atom in $\mu$, by not being a nuisance, since that probability atom in $\nu$ would not allow any atom at smaller locations in $\mu$ to satisfy $y-x\leq k$ anyway. By always transporting to atoms in $\nu$ at the largest allowed locations, either when helping toward the   
objective $y-x\leq k$ or not, we make sure that probability mass at smaller locations in $\mu$ has the best chance of satisfying $y-x\leq k$ when its turn to be coupled comes. What could occur, however, would be an atom in $\mu$ satisfying $y-x\leq k$ preventing an atom at a smaller location to do so, but this situation results in effectively the same value for $\int \id_{\{(x,y)\in\mathbb{R}^2: y-x>k\}}\mathrm{d}\pi$, $\pi \in \mathcal{H}(\mu,\nu)$, as switching would necessarily cause the first atom to not satisfy it anymore and both atoms have the same size. 
This construction is made precise in Algorithm~\ref{algo:A}.

\renewcommand{\thealgofloat}{A}

\begin{algofloat}
\centering

   \caption{}
    \fbox{
    \begin{minipage}{0.9\textwidth}
    \onehalfspacing
       % \textbf{Algorithm~\ref{algo:A} %solving $\inf\{\mathbb{P}(Y-X > k) : X\lawis \mu, Y\lawis \nu, Y\geq X\}$ }\\
       For $\mu, \nu \in \mathcal{M}$ that are atomic with the same size $a>0$ and $\mu\le_{\rm st}\nu$, 
        let $x_1 \leq \cdots \leq x_m$ and $y_1 \leq \cdots \leq y_n$ be the ordered locations of the atoms in $\mu$ and $\nu$. 
        \begin{itemize}
             \item[$\blacktriangleright$] Set $\mathcal{I}_m: = \{1,\ldots,n \}$. \item[$\blacktriangleright$] For $i=m, m-1, \ldots, 1$, in  descending order, do:
            \begin{itemize}
                \item[$\triangleright$] Check whether $\mathcal{J}_i := \{j\in \mathcal{I}_i:  x_i\leq y_j\leq x_i+k\}$ is empty.
                \item[$\triangleright$] If $\mathcal{J}_i$ is non-empty, set $d_i := \sup \mathcal{J}_i$;\\
                if $\mathcal{J}_i$ is empty, set $d_i: = \sup \mathcal{I}_i$.
                \item[$\triangleright$]  Define $\mathcal{I}_{i-1}: = \mathcal{I}_i\backslash \{ d_i\}$.
            \end{itemize}
            \item[$\blacktriangleright$] Return $Q^{\rm A}_{k}(\mu,\nu): = a (\sum_{i=1}^m \id_{\{y_{d_i} - x_i > k\}}) $. 
        \end{itemize}
    \end{minipage}
    }
\label{algo:A}
\end{algofloat}

The idea for constructing the optimal  $\pi$ for \eqref{eq:submeasures-sup} is similar to the one for  \eqref{eq:submeasures-inf}. Iteratively for each $x$ in descending order, we couple it to the yet-uncoupled atom in $\nu$ at the smallest location $y$ such that $y-x > k$, and, if not possible, to the one at the smallest location such that $y \geq x$. 
This construction  maximizes the total portion of   mass in $\nu$ transported such that $y- x> k$. The rest of the transport instructions is predicated on the principle of least nuisance: If the mass portion is capable of helping toward the objective $y-x> k$, it should do so by employing the smallest location $y$. If it cannot help, it should be the least troublesome and take with itself the smallest possible values of $y$ under the constraint $y\geq x$: this leaves the greatest flexibility for subsequent to-be-coupled atoms to satisfy $y-x> k$. By going through atoms in $\mu$ from the largest to the smallest location, we make sure that unhelpful atoms are so by obligation. It also ensures that at least some uncoupled atom always permits $y\geq x$, given stochastic dominance. This construction is made precise in Algorithm~\ref{algo:B}.

\renewcommand{\thealgofloat}{B}

\begin{algofloat}
\centering

   \caption{}
    \fbox{
    \begin{minipage}{0.9\textwidth}
        %\textbf{Algorithm~\ref{algo:B}}\\
            \onehalfspacing
           For $\mu, \nu \in \mathcal{M}$ that are atomic with the same size $a>0$,
        let $x_1 \leq \cdots \leq x_m$ and $y_1 \leq \cdots \leq y_n$ be the ordered locations of the atoms in $\mu$ and $\nu$. 
        \begin{itemize}
        \item[$\blacktriangleright$] Set $\mathcal{I}_m := \{1,\ldots,n \}$. 
            \item[$\blacktriangleright$] For $i=m, m-1, \ldots, 1$, in  descending order, do:
            \begin{itemize}
                \item[$\triangleright$] Check whether $\mathcal{J}_i := \{j\in \mathcal{I}_i:  y_j > x_i+k\}$ is empty.
                \item[$\triangleright$] If $\mathcal{J}_i$ is non-empty, set $d_i := \inf \mathcal{J}_i$;\\
                if $\mathcal{J}_i$ is empty, set $d_i := \inf \{j\in \mathcal{I}_i : y_j\geq x_i\}$.
                \item[$\triangleright$]  Define $\mathcal{I}_{i-1} := \mathcal{I}_i\backslash \{ d_i\}$.
            \end{itemize}
            \item[$\blacktriangleright$] Return $Q^{\rm B}_k(\mu,\nu) := a (\sum_{i=1}^m \id_{\{y_{d_i} - x_i > k\}}) $. 
        \end{itemize}
    \end{minipage}
    }
    \label{algo:B}
\end{algofloat}

  The optimality claimed for Algorithms \ref{algo:A}--\ref{algo:B}
  is formally stated in 
 Theorem~\ref{th:Q-inf}. 

\begin{theorem}
\label{th:Q-inf}
    Let $\mu, \nu \in \mathcal{M}$ be atomic with the same size $a>0$,  $\mu\leq_{\rm st}\nu$, and  $k\geq 0$. It holds that
    \begin{enumerate}[label={\rm (\roman*)}, ref={\rm (\roman*)}]
        \item \label{item:A}  $ Q^{\inf}_k(\mu,\nu)=Q^{\rm A}_k(\mu,\nu) $ with $Q^{\rm A}_k(\mu,\nu)$ specified in 
    Algorithm~\ref{algo:A};  
        \item \label{item:B} $ Q^{\sup}_k(\mu,\nu)=Q^{\rm B}_k(\mu,\nu) $ with  $Q^{\rm B}_k(\mu,\nu)$ specified in 
    Algorithm~\ref{algo:B}. 
    \end{enumerate} 
    Moreover, $ Q^{\inf}_k(\mu,\nu)$ 
    and 
     $ Q^{\sup}_k(\mu,\nu)  $
     take values in the set $\{0,a,2a,\dots, m a\}$,
     where   $m$ is in \eqref{eq:mu-a-m}.
     
\end{theorem}
The proof of Theorem \ref{th:Q-inf} is presented in Appendix~\ref{sect:proofs}. 
 Despite the similarity in the two problems \eqref{eq:submeasures-inf}--\eqref{eq:submeasures-sup} and the two algorithms, items~\ref{item:A} and~\ref{item:B} require different proof arguments.

   Algorithms~\ref{algo:A}--\ref{algo:B} not only provide the optimal values for \eqref{eq:submeasures-inf}--\eqref{eq:submeasures-sup}, they also construct corresponding optimizers (not necessarily unique), which are described below.  
For any atomic coupling $\pi$  with size $a$, define its \textit{coupling multiset} as the multiset $\Gamma$ of atom locations satisfying
\begin{equation*}
    \pi = a \sum_{(x,y)\in \Gamma} \delta_{(x,y)}. 
\end{equation*}
Note that, for a given size $a$, a coupling multiset entirely defines a coupling. 
%For $\mu$ and $\nu$ atomic with the same size $a$,
%we identify an atomic coupling $\pi$ with size $a$  between $\mu$ and $\nu$ with the multiset of atom locations in $\pi$.
% We will refer to any multiset
 %of pairs between $\{x_1,\dots,x_m\}$
% and $\{y_1,\dots,y_n\}$
% as a \emph{coupling} between $\mu$ and $\nu$. 
We denote by
 $\Gamma^{\rm A}_{k,\mu,\nu}$ the coupling multiset made of pairs $(x_{i},y_{d_i}) $  produced by Algorithm~\ref{algo:A}.
 The corresponding coupling is denoted by $\gamma^{\rm A}_{k,\mu,\nu} $, that is, $$\gamma^{\rm A}_{k,\mu,\nu}  = a\sum_{(x,y)\in\Gamma^{\rm A}_{k,\mu,\nu}}\delta_{(x,y)},$$ where $a$ is the size of one atom. In terms of optimal transport, $\gamma^{\rm A}_{k,\mu,\nu}$ is an unbalanced transport plan (``unbalanced" reflects the possibility of $\mu(\R)< \nu(\R)$), and it is optimal by Theorem~\ref{th:Q-inf}.   
Similarly, let $\Gamma^{\rm B}_{k,\mu,\nu}$ be the coupling multiset produced by Algorithm~\ref{algo:B}, and $\gamma^{\rm B}_{k,\mu,\nu}$ be the corresponding transport plan.  We provide examples of $\Gamma^{\rm A}_{k,\mu,\nu}$ and $\Gamma^{\rm B}_{k,\mu,\nu}$ for specific $\mu$ and $\nu$ below. 
%The construction of $\mathcal{D}_k(\nu,\mu)$ predicates on the idea of least nuisance.

\begin{example}
\label{ex:atomic-inf}
Consider probability measures $\mu$ and $\nu$, atomic with the same size, with atoms at locations  $\mathbf{x}=\{0,0.5,2,4,4.5\}$ and $\mathbf{y} = \{3.5,4,5.5,6,7\}$.
Note that $m=n=5$. Let $k=3$. Using Algorithm~\ref{algo:A}, we obtain $Q^{\rm A}_3(\mu,\nu) = 1/5$, which, by Theorem~\ref{th:Q-inf}, is equal to $Q_{3}^{\inf}(\mu,\nu)$. The coupling multiset $\Gamma^{\rm A}_{3,\mu,\nu}$ produced in the process is 
\begin{equation*}
    \Gamma_{3,\mu,\nu}^{\rm A} = \{(4.5,7),(4,6),(2,4),(0.5,3.5),(0,5.5)\}.
\end{equation*}
This coupling multiset is depicted in Figure~\ref{fig:D} where in green are the pairs such that $y-x > 3$  and in red are the ones such that $y-x \leq 3$. Note how ${\Gamma}_{k,\mu,\nu}^{\rm A}$ yields a smaller probability of worthwhile effect compared to comonotonic and DL couplings (see Appendix~\ref{sect:dl}):  comonotonic coupling would yield $\mathbb{P}(Y-X>3) = 3/5$ and DL-coupling, $\mathbb{P}(Y-X > 3) = 2/5$, for $X\lawis \mu$, $Y\lawis \nu$. %Note that an independence coupling would contradict monotone treatment response.       
\end{example}

\begin{figure}[tb]
\centering
	\begin{minipage}{0.70\textwidth}
		\centering
		\begin{tikzpicture}
          \node at (-1.25, 0.75) [rectangle] {$\mu$};
          \node at (-1.25, 3.25) [rectangle] {$\nu$};
			\draw[Black!40,  ->, thick] (-0.5,0) -- (7.5, 0);
			
			\foreach \x in {0,1,...,7}
			\draw[Black!40,  thick] (\x cm,-2pt) node[below]{\x} -- (\x cm,2pt);
			
			\filldraw (3.5,3.25) circle (0.1 cm);
			\filldraw (4,3.25) circle (0.1 cm);
			\filldraw (5.5,3.25) circle (0.1 cm);
			\filldraw (6,3.25) circle (0.1 cm);
			\filldraw (7,3.25) circle (0.1 cm);
			
			\draw[Red, very thick, ->] (4.5, 0.75) -- (7, 3.25);
			\draw[Red, very thick, ->] (4, 0.75) -- (6, 3.25);
			\draw[Red, very thick, ->] (2, 0.75) -- (4, 3.25);
			\draw[Red, very thick, ->] (0.5, 0.75) -- (3.5, 3.25);
			\draw[LimeGreen, very thick, ->] (0,0.75) -- (5.5,3.25); 
			
			\filldraw (0,0.75) circle (0.1 cm);
			\filldraw (0.5,0.75) circle (0.1 cm);
			\filldraw (2,0.75) circle (0.1 cm);
			\filldraw (4,0.75) circle (0.1 cm);
			\filldraw (4.5,0.75) circle (0.1 cm);
		\end{tikzpicture}	
	\end{minipage}
	\caption{Illustration of the coupling multiset $\Gamma^{A}_{k,\mu,\nu}$ produced by Algorithm~\ref{algo:A},
    for $k=3$, for the two discrete distributions with same-size atoms from Example~\ref{ex:atomic-inf}. In green are the pairs such that $y-x > 3$%(the treatment has a worthwhile effect)
    ; in red are the ones such that $y-x \leq 3$. %(the treatment does not have a worthwhile effect). 
    }
	\label{fig:D}
\end{figure}
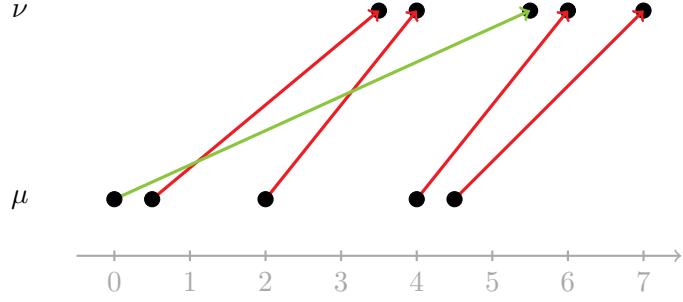

%Note that when $k\to 0$, comonotonicity

% We next present Algorithm~\ref{algo:B}, solving the problem in \eqref{eq:submeasures-sup}, in a symmetric manner to Algorithm~\ref{algo:A}.

%The following theorem is a pendant to Theorem~\ref{th:Q-inf} for the problem in \eqref{eq:submeasures-sup}. 

% \begin{theorem}
% \label{th:Q-sup}
%      Let $\mu, \nu \in \mathcal{M}$ be atomic with same size, and $\mu\leq_{\rm st}\nu$. Then, $ Q^{\sup}_k(\mu,\nu)=Q^{\rm B}_k(\mu,\nu) $ for $k\geq 0$ and $Q^{\rm B}_k(\mu,\nu)$ specified in 
%     Algorithm~\ref{algo:B}.  
% \end{theorem}

\begin{example}
\label{ex:atomic-sup}
Consider probability measures $\mu$ and $\nu$, atomic with same size, with atoms at locations $\mathbf{x}=\{0,0.5,2, 3,4,5\}$ and $\mathbf{y} = \{3,3.5,4,5,6,7\}$, and let $k=2$. From Algorithm~\ref{algo:B}, we obtain $Q^{\rm B}_2(\mu,\nu) = 2/3$. By Theorem~\ref{th:Q-inf}, we thus have $Q^{\sup}_2(\mu,\nu) = 2/3$. The coupling multiset $\Gamma_{2,\mu,\nu}^{\rm B}$ produced in the process is 
\begin{equation*}
    \Gamma_{2,\mu,\nu}^{\rm B} = \{(5,5),(4,7),(3,6),(2,3),(0.5,3.5), (0,4)\}.
\end{equation*}
This coupling multiset is depicted in Figure~\ref{fig:E} where in green are the pairs such that $y-x > 2$  and in red are the ones such that $y-x \leq 2$. The coupling multiset ${\Gamma}_{k,\mu,\nu}^{\rm B}$ yields a larger probability of worthwhile effect compared to comonotonic and DL couplings: both would produce $\mathbb{P}(Y-X>2) = 1/3$, for $X\lawis \mu$, $Y\lawis \nu$. %Again, independence would contradict monotone treatment response.       
\end{example}

\begin{figure}[tb]
\centering
	\begin{minipage}{0.70\textwidth}
		\centering
		\begin{tikzpicture}
         \node at (-1.25, 0.75) [rectangle] {$\mu$};
          \node at (-1.25, 3.25) [rectangle] {$\nu$};
			\draw[Black!40,  ->, thick] (-0.5,0) -- (7.5, 0);
			
			\foreach \x in {0,1,...,7}
			\draw[Black!40,  thick] (\x cm,-2pt) node[below]{\x} -- (\x cm,2pt);
			
			\filldraw (3,3.25) circle (0.1 cm);
			\filldraw (3.5,3.25) circle (0.1 cm);
			\filldraw (4,3.25) circle (0.1 cm);
			\filldraw (5,3.25) circle (0.1 cm);
			\filldraw (6,3.25) circle (0.1 cm);
			\filldraw (7,3.25) circle (0.1 cm);
			
			\draw[Red, very thick, ->] (5, 0.75) -- (5, 3.25);
			\draw[LimeGreen, very thick, ->] (4, 0.75) -- (7, 3.25);
			\draw[LimeGreen, very thick, ->] (3, 0.75) -- (6, 3.25);
			\draw[Red, very thick, ->] (2, 0.75) -- (3, 3.25);
			\draw[LimeGreen, very thick, ->] (0.5,0.75) -- (3.5,3.25); 
			\draw[LimeGreen, very thick, ->] (0,0.75) -- (4,3.25);
			
			\filldraw (0,0.75) circle (0.1 cm);
			\filldraw (0.5,0.75) circle (0.1 cm);
			\filldraw (2,0.75) circle (0.1 cm);
			\filldraw (3,0.75) circle (0.1 cm);
			\filldraw (4,0.75) circle (0.1 cm);
			\filldraw (5,0.75) circle (0.1 cm);
		\end{tikzpicture}	
	\end{minipage}
	\caption{Illustration of the coupling multiset $\Gamma^{B}_{k,\mu,\nu}$ produced by Algorithm~\ref{algo:B},
    for $k=2$, for the two discrete distributions with same-size atoms from Example~\ref{ex:atomic-sup}. In green are the pairs such that $y-x > 2$; in red are the ones such that $y-x \leq  2$.  }
	\label{fig:E}
\end{figure}
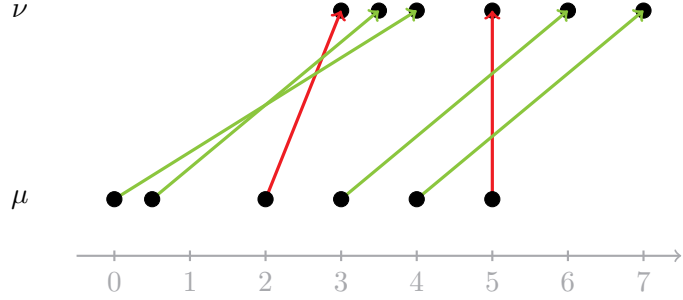

It is clear from Figure~\ref{fig:D} that $\gamma^{\rm A}_{k, \mu, \nu}$ is not the unique transport plan solving \eqref{eq:submeasures-inf}, as one could have rather coupled, in Example~\ref{ex:atomic-inf}, $x_4$ to $y_5$ and $x_5$ to $y_4$, still respecting the constraints, and it would have no incidence on the probability of worthwhile effect. Similarly, an alternative to $\gamma^{B}_{k,\mu,\nu}$ in the case of Example~\ref{ex:atomic-sup} would be the transport plan obtained by rather coupling $x_1$ to $y_2$ and $x_2$ to $y_3$.%; it is also not the unique solution.

The next result highlights some stability from the algorithms, in the sense that the coupling multiset obtained may be separated in two parts, depending on whether atom pairs meet the objective or not, and re-running the algorithm on either part will not alter the coupling multiset. 

\begin{proposition}
\label{prop:separable}
	The coupling multiset $\Gamma^{\rm A}_{k,\mu,\nu}$ is separable in the following sense: Let $\Gamma^{>} = \{(x,y) \in \Gamma^{\rm A}_{k,\mu,\nu}: y-x > k\}$  and $\Gamma^{\leq} = \Gamma^{\rm A}_{k,\mu,\nu}\backslash \Gamma^{>}$. Then, we have $\Gamma^{>} = \Gamma^{A}_{k,\mu\vert_{\Gamma^{>}} ,\nu\vert_{\Gamma^{>}}}$ and $\Gamma^{\leq} = \Gamma^{A}_{k,\mu\vert_{\Gamma^{\leq}} ,\nu\vert_{\Gamma^{\leq}}}$, where $\mu|_{\Gamma}$ (resp.~$\nu|_{\Gamma}$) means restriction to the first (resp. second) component of multiset $\Gamma$'s elements. 
The coupling multiset $\Gamma^{\rm B}_{k,\mu,\nu}$ is separable similarly. 
    %$\mathcolor{red}{P = \mathcal{D}(X+k, Y) + \mathcal{D}(X, Y)}$. 
\end{proposition}

The separability of the coupling multiset will prove useful below in showing the stability of the couplings produced by the algorithms for limit cases of their applications.

\subsection{Some results on the coupling solutions}

This subsection gathers some results on the couplings solving \eqref{eq:submeasures-inf}--\eqref{eq:submeasures-sup}; the aim is to provide further insight on the structure of such couplings and how small variations in the parameters of the problems may affect the solution. 
Although we previously only discussed solving \eqref{eq:submeasures-inf}--\eqref{eq:submeasures-sup} for atomic measures with same size, the results in this subsection hold for any types of measures. Theorems~\ref{th:stosto} and~\ref{th:tv} below will be required in solving the problems for general discrete measures and continuous measures, in the next subsections.

First, we remark that values $Q_k^{\inf}(\mu,\nu)$ and $Q_{k}^{\sup}(\mu,\nu)$ not only bound the possible probabilities of worthwhile treatment effect for $\mu$ and $\nu$, but mark the interval of all of their possible values; we have 
\begin{equation}
   \left[Q_k^{\inf}(\mu,\nu), Q_k^{\sup}(\mu,\nu) \right) \subseteq \left\{\int \id_{\{y-x>k\}}\d \pi : \pi\in\mathcal{H}(\mu,\nu) \right\} \subseteq  \left[Q_k^{\rm inf}(\mu,\nu), Q_k^{\sup}(\mu,\nu) \right]
   \label{eq:bounds}
\end{equation}
This is because the probability $\int \id_{\{y-x>k\}}\d \pi$ is linear in the coupling $\pi$ and $\mathcal{H}(\mu,\nu)$ is convex. The attainability of the lower bound is claimed in Theorem~\ref{th:stosto} just below; the upper bound may be not attained, see Appendix~\ref{sect:approx-sup}.
%Moreover, any  point  in $[Q_k^{\inf}(\mu,\nu), Q_{k}^{\sup}(\mu,\nu)] $ with the form $aq$ for a rational $q$ is attained by an atomic $\pi$.
%  Let $\mu,\nu\in\mathcal{M}$ be atomic with the same size, $\mu\leq_{\rm st}\nu$, and $k\geq 0$. 
% For any $Q^*$ such that $Q^{\inf}_k(\mu, \nu) \leq Q^* \leq Q^{\sup}_{k}(\mu,\nu)$, there exists  $\gamma^* \in \mathcal{H}(\mu,\nu)$, a transport plan such that $Q^* = \int \id_{\{(x,y)\in\mathbb{R}^2: y-x>k \}}\mathrm{d}\gamma^* $.  Moreover, if $Q^*$ is rational, $\gamma^*$ is an atomic measure.

\begin{theorem}
\label{th:stosto}
   In \eqref{eq:submeasures-inf}, the infimum is attained, and a solution $\pi^*$ can be chosen such that $P_2(\pi^*)\leq_{\rm st}P_2(\pi)$ for all $\pi\in\mathcal{H}(\mu,\nu)$. 
\end{theorem}

The selection result of Theorem~\ref{th:stosto} is noteworthy in particular since $\nu\leq_{\rm st} \nu^{\prime}$ does not imply $Q_{k}^{\inf}(\mu,\nu) \leq Q_{k}^{\inf}(\mu,\nu^{\prime})$, as in the following example. 
\begin{example}
\label{ex:0356}
Consider $\mu,\nu,\nu^{\prime}$, atomic probability measures with same size, with atoms respectively at locations $\mathbf{x} = \{0,5\}$, $\mathbf{y}=\{3,8\}$ and $\mathbf{y}^{\prime} = \{5,8\}$. Then, $\nu\leq_{\rm st}\nu^{\prime}$ but $Q_{1}^{\inf}(\mu,\nu) = 1$ and $Q_{1}^{\inf}(\mu,\nu^{\prime}) = 0.5$. Next, if we consider $\nu^{\dagger}$ a measure atomic with same size as $\mu$, with atom locations at $\mathbf{y}^{\dagger} = \{3,5,8\} $, we have $Q_1^{\inf}(\mu,\nu^{\dagger}) = 0.5$ and, by Theorem~\ref{th:stosto}, the coupling solution $\pi^*$ may effectively be taken so that $P_2(\pi^*) =: \nu^* \leq_{\rm st} P_2(\pi)$ for every $\pi\in\mathcal{H}(\mu,\nu^{\dagger})$. Namely, $\pi^*$ is described by the coupling multiset $\Gamma^* = \{(0,3), (5,5)\}$.
Let us remark that $\nu,\nu^{\prime},\nu^* \leq \nu^{\dagger}$, and that $\nu^*\leq_{\rm st} \nu$ and $\nu^{*}\leq_{\rm st}\nu^{\prime}$.
\end{example}
In the case where $\mu$ and $\nu$ are probability measures and $k=0$, the upper bound in \eqref{eq:bounds} is attained and the solution is given by the comonotonic coupling, as indicated in the following proposition. Apart from the case $k=0$, however, the optimizers do not have a simple dependence structure like comonotonicity or countermonotonicity.

\begin{proposition} \label{prop:comono} Let $\mu,\nu \in \mathcal{M}$ be probability measures, and $\mu\leq_{\rm st} \nu$. 
Then, $$Q^{\sup}_0(\mu,\nu) = 1-
\int_0^1 \id_{\{F_{\mu}^{-1}(u) = F_{\nu}^{-1}(u)\}}(u)\mathrm{d}u  = \mathbb{P}(Y - X >0 ),$$  
where  $X \lawis \mu$ and $Y\lawis \nu$ are comonotonic.
%Comonotonicity solves $\sup\{\int \id_{\{(x,y)\in\mathbb{R}^2: y>x\}}\mathrm{d}\pi : \pi\in \mathcal{H}(\mu,\nu)\}$. 
\end{proposition}

%In particular, Proposition~\ref{prop:comono} implies that the function $k\mapsto Q^{\rm sup}_k( \mu,\nu)$ is not continuous at $0$, even if $\mu$ and $\nu$ are absolutely continuous probability measures; indeed, the comonotonic coupling will not necessarily yield $\mathbb{P}(Y-X>0)=1$, yet $\mathbb{P}(Y-X\geq 0) = 1$ from the ordering constraint. 

%From the instructions in Algorithm~\ref{algo:B}, we note that $\gamma^{\rm B}_{0,\mu,\nu} = \mathrm{dl}$. \textcolor{magenta}{irrelevant : all elements in $\mathcal{H}(\mu,\nu)$ will solve.} But then, $\lim_{k\to 0} \gamma^{\rm B}_{0,\mu,\nu}$ is comonotonicity. 

For the remainder of Section~\ref{sect:EFFECTIVENESS} (that is, also for upcoming Sections~\ref{sec:discrete} and~\ref{sec:cont}), we will discuss only the problem in \eqref{eq:submeasures-inf}. While results for \eqref{eq:submeasures-sup} would be similar in essence, some adjustments to the mathematical settings are required, which we believe would overburden the discussion; rather, see~Appendix~\ref{sect:approx-sup}. The result in the following theorem constrains how much $Q_k^{\inf}$ varies when adding or removing mass in the coupling pool.

\begin{theorem}
\label{th:tv}
    Consider $\mu,\nu,\mu^{\prime},\nu^{\prime}\in\mathcal{M}$ such that $\mu\leq_{\rm st}\nu$, $\mu^{\prime}\leq_{\rm st}\nu^{\prime}$, $\mu\geq \mu^{\prime}$ and $\nu\leq \nu^{\prime}$. For every $k\geq 0$, it holds that
   \begin{equation*}
      0 \leq  Q^{\inf}_k(\mu,\nu) - Q^{\inf}_k(\mu^{\prime},\nu^{\prime})  \leq \mu(\mathbb{R}) - \mu^{\prime}(\mathbb{R}) + \nu^{\prime}(\mathbb{R}) - \nu(\mathbb{R}).
    \end{equation*}
    \end{theorem}

    % Perhaps surprisingly, the statement of Theorem~\ref{th:tv} would fail to hold without the assumption of monotone treatment response, as shown in the next example.

Theorem~\ref{th:tv} indicates that adding or removing small portions of mass in $\mu$ or $\nu$ cannot snowball into large variations for $Q_k^{\inf}(\mu,\nu)$; the disruption is at most the amount of mass added. Let us remark that the 0 lower bound is non-trivial for $\mu \neq \mu^{\prime}$. In particular, it indicates that a portion of mass in $\mu$ never leads to major compromises in satisfying the objective $y-x\leq k$ because it has to be coupled. In Appendix~\ref{sect:approx-sup}, we provide an example where this would be the case for $Q_k^{\sup}(\mu,\nu)$, and this predicates the need to tweak the setting accordingly for approximations regarding problem~\eqref{eq:submeasures-sup}.

%The upcoming subsections  to more general measures. 

\subsection{Discrete measures}
\label{sec:discrete}

%Suppose the supports of $\mu$ and $\nu$ are subsets of the integers $\mathbb{Z}$. 
We next extend our solutions to \eqref{eq:submeasures-inf}--\eqref{eq:submeasures-sup} for $\mu$ and $\nu$ that are discrete measures but not necessarily atomic with the same size. These solutions are obtained as limit cases to those for measures that are atomic with the same size. 

Let $\mu$ and $\nu$ be two discrete measures such that $|\mathrm{supp}(\nu)|<\infty$ and $\mu\leq_{\rm st}\nu$. Define $\widetilde{\mu}_r$ and $\widetilde{\nu}_r$, $r\in\mathbb{N}$, as %subprobability 
atomized versions of $\mu$ and $\nu$, described by
\begin{align}
    \widetilde{\mu}_r (\{x\}) &= \frac{1}{2^r}\left\lfloor 2^r{\mu(\{x\})} \right\rfloor, \quad \text{for every $x\in\mathrm{supp}(\mu)$}, \label{eq:atomization-mu}\\
    \widetilde{\nu}_r (\{y\}) &= \frac{1}{2^r}\left\lceil 2^{r}\nu(\{y\}) \right\rceil, \quad \text{for every $y\in\mathrm{supp}(\nu)$}, \label{eq:atomization-nu}
\end{align}
and this construction ensures  that the ordering constraint $\widetilde{\mu}_r\le_{\rm st} \widetilde{\nu}_r$  holds for every atomization level $r$. 
It also highlights the importance of allowing for $\mu(\R)<\nu(\R)$ in Theorem \ref{th:Q-inf}
because we will apply it to $    \widetilde{\mu}_r $ and $    \widetilde{\nu}_r$, which do not have the same total mass in general. We require the support of $\nu$ to be finite, otherwise $\widetilde{\nu}_r$ would have infinite mass for any $r\in\mathbb{N}$.

\begin{proposition}
\label{prop:atomization}
Let $\mu, \nu\in\mathcal{M}$ be discrete measures such that $|\mathrm{supp}(\nu)|<\infty$ and $\mu\leq_{\rm st}\nu$.
    For any $k\geq 0$, it holds that $Q_k^{\rm A}(\widetilde{\mu}_r, \widetilde{\nu}_r) \uparrow Q_k^{\inf}(\mu, \nu)$ as $r\to\infty$. 
\end{proposition}

Executing the algorithm on atomized versions of $\mu$ and $\nu$ will not only yield an approximate value of $Q^{\inf}_k(\mu,\nu)$, it most notably produces a coupling that is distributionally close to the optimal coupling,  thus giving valuable information on the dependence scheme solving \eqref{eq:submeasures-inf}. 

\begin{theorem}
\label{th:discrete}
   % For discrete $\mu, \nu$ such that $\mu \leq_{\rm st}\nu$, let $\widetilde{\mu}_n$ and $\widetilde{\nu}_n$ be their atomized versions as defined above for all $n\in\mathbb{N}$. 
   Let $\mu, \nu\in\mathcal{M}$ be discrete measures such that $|\mathrm{supp}(\nu)|<\infty$ and $\mu\leq_{\rm st}\nu$.
   For any $k\geq 0$, there exists $\gamma \in \mathcal{H}(\mu,\nu)$ such that $\gamma_{k,\widetilde{\mu}_r,\widetilde{\nu}_r}^{\rm A} \to \gamma$ as $r\to \infty$. 
\end{theorem}

For discrete $\mu$ and $\nu$, define $\gamma^{A}_{k,\mu,\nu} = \lim_{r\to\infty} \gamma^{A}_{k,\widetilde{\mu}_r,\widetilde{\nu}_r} $ as per Theorem~\ref{th:discrete}. 
Although Proposition~\ref{prop:atomization} and Theorem~\ref{th:discrete} require $|\mathrm{supp}(\nu)|<\infty$, it is possible to handle cases where the measures have infinite support. For example, in cases where $\mu$ and $\nu$ have a support that is unbounded to the right, one does not have a natural starting point to construct $\gamma^{\mathrm{A}}_{k,\mu,\nu}$ and the atomization of $\nu$ would produce unbounded measures, but a possible workaround is by finding  
$t\in\mathbb{R}$ such that 
$\mu \vert_{[t,\infty)}   \leq \nu^{[-k]}\vert_{[t+k,\infty)} $, where $\nu^{[-k]}(A)=\nu(A+k)$ for every Borel set $A$ on $\mathbb{R}$. Then, it is certain that all atom locations in $[t,\infty)$ are able to be coupled so that they satisfy the objective $y-x \leq k$, specifically by coupling every $x$ to $y=x+k$. We continue by constructing the coupling leftward from $t$.  We illustrate this in the following example.
%\textcolor{magenta}{Add the trick for $\mathcal{D}_k$}

\begin{example}
\label{ex:poisson-2}
Consider $ \mu =$ Poisson$(1)$ and $\nu =$ Poisson$(2)$. Obviously, $\mu\leq_{\rm st} \nu$; we also note that $\mu\vert_{[3,\infty)} \leq \nu^{[-1]}\vert_{[4,\infty)}$.
    Suppose $k=1$. The transport plan $\gamma^{\rm A}_{1,\mu,\nu}$  solving \eqref{eq:submeasures-inf} is depicted in Figure~\ref{fig:Dex-poisson}, and we remark that for every $x\in\{3,4,\ldots\}$, mass at that location in $\mu$ is coupled to location $y=x+1$ in $\nu$. 
    We obtain $Q_1^{\inf}(\mu,\nu) = 0.0626$.
\end{example}

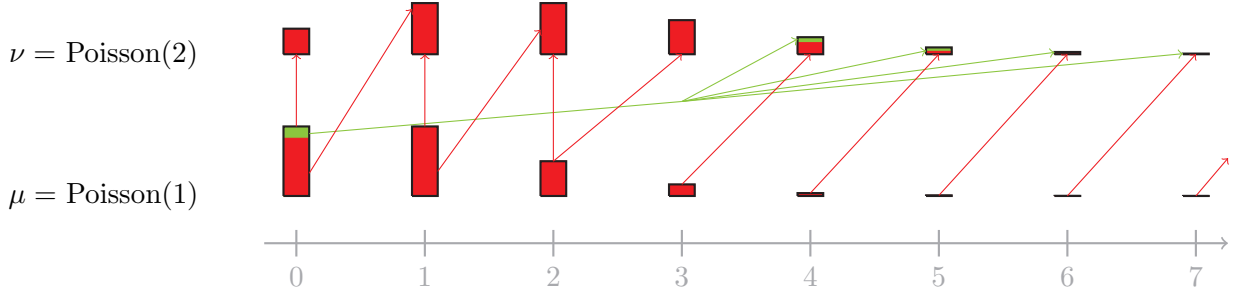
\begin{figure}[tb]
\centering
\begin{tikzpicture}[xscale=1.7, yscale=2.5]
	\node at (-1.5, 0.25) [rectangle] {$\mu = $ Poisson(1)};
	\node at (-1.5, 1) [rectangle] {$\nu = $ Poisson(2)};
	\draw[Black!40,  ->, thick] (-0.25,0) -- (7.25, 0);
			
	\foreach \x in {0,1,...,7}
   \draw[Black!40, thick] (\x cm,-2pt) node[below]{\x} -- (\x cm,2pt);

\foreach \x/\y in {
4/0.0902,
5/0.0361,
6/0.0120,
7/0.0034
}
\filldraw[LimeGreen] (\x-0.1, 1) -- (\x+0.1,1) -- (\x+0.1, \y+1) -- (\x-0.1,\y+1) -- (\x-0.1,1);

\foreach \x/\y in {
0/0.1353,
1/0.2707,
2/0.2707,
3/0.1804,
4/0.0613,
5/0.0153,
6/0.0031,
7/0.0005
}
\filldraw[Red] (\x-0.1, 1) -- (\x+0.1,1) -- (\x+0.1, \y+1) -- (\x-0.1,\y+1) -- (\x-0.1,1);

\foreach \x/\y in {
0/0.1353,
1/0.2707,
2/0.2707,
3/0.1804,
4/0.0902,
5/0.0361,
6/0.0120,
7/0.0034
}
\draw[Black, thick] (\x-0.1, 1) -- (\x+0.1,1) -- (\x+0.1, \y+1) -- (\x-0.1,\y+1) -- (\x-0.1,1);

\foreach \x/\y in {
0/0.3679,
}
\filldraw[LimeGreen] (\x-0.1, 0.25) -- (\x+0.1,0.25) -- (\x+0.1, \y+0.25) -- (\x-0.1,\y+0.25) -- (\x-0.1,0.25);

\foreach \x/\y in {
0/0.3053,
1/0.3679,
2/0.1839,
3/0.0613,
4/0.0153,
5/0.0031,
6/0.0005,
7/0.0001
}
\filldraw[Red] (\x-0.1, 0.25) -- (\x+0.1,0.25) -- (\x+0.1, \y+0.25) -- (\x-0.1,\y+0.25) -- (\x-0.1,0.25);

\foreach \x/\y in {
0/0.3679,
1/0.3679,
2/0.1839,
3/0.0613,
4/0.0153,
5/0.0031,
6/0.0005,
7/0.0001}
\draw[Black, thick] (\x-0.1, 0.25) -- (\x+0.1,0.25) -- (\x+0.1, \y+0.25) -- (\x-0.1,\y+0.25) -- (\x-0.1,0.25);

\foreach \xz/\yz/\xw/\yw in {
0.1/0.33/3/-0.25
}
\draw[LimeGreen](\xz,\yz+0.25) -- (\xw, \yw+1);

\foreach \xz/\yz/\xw/\yw in {
3/0.5/3.9/0.08,
3/0.5/4.9/0.018,
3/0.5/5.9/0.01,
3/0.5/6.9/0.003
}
\draw[LimeGreen, ->](\xz,\yz+0.25) -- (\xw, \yw+1);

\foreach \xz/\yz/\xw/\yw in {
7/0.0001/7.25/-0.55,
6/0.0005/7/0,
5/0.0031/6/0,
4/0.0153/5/0,
3/0.0613/4/0,
2/0.1839/3/0,
2/0.1839/2/0,
1.1/0.13/1.9/0.13,
1/0.3679/1/0,
0.1/0.12/0.9/0.24,
0/0.3679/0/0
}
\draw[Red, ->](\xz,\yz+0.25) -- (\xw, \yw+1);
			
	% \draw[Red, very thick, ->] (1.666, 1) -- (1.666, 3);
	% \draw[Red, very thick, ->] (1.666, 1) -- (2.333, 3);
	% \draw[Red, very thick, ->] (1.333, 1) -- (1.333, 3);
	% \draw[Red, very thick, ->] (1.333, 1) -- (2.666, 3);
	% \draw[Red, very thick, ->] (0.666, 1) -- (0.666, 3);
	% \draw[LimeGreen, very thick, ->] (0.666, 1) -- (3.666, 3);
	% \draw[Red, very thick, ->] (0.333, 1) -- (0.333, 3);
	% \draw[LimeGreen, very thick, ->] (0.333,1) -- (3.333,3); 
			
	% \filldraw[LimeGreen] (0,0.5) -- (0,0.75) -- (1, 0.75) -- (1,0.5) ;
	% \filldraw[Red] (1,0.5) -- (1,1) -- (2, 1) -- (2,0.5) ;
	% \filldraw[Red] (0,0.75) -- (0,1) -- (1, 1) -- (1,0.75) ;
	% \draw[Black, very thick] (0,0.5) -- (0,1) -- (2, 1) -- (2,0.5) -- (0,0.5) ;
		\end{tikzpicture}
        \caption{Transport plan $\gamma^{\rm A}_{1,\mu,\nu}$ from Example~\ref{ex:poisson-2}.}
        \label{fig:Dex-poisson}
\end{figure}

\subsection {Absolutely continuous measures}
\label{sec:cont}

We now let $\mu$ and $\nu$ be absolutely continuous measures, with survival functions $S_{\mu}$ and $S_{\nu}$, whose derivatives are $-f_{\mu}$ and $-f_{\nu}$; we call $f_{\mu}$ and $f_{\nu}$ the \textit{densities} of $\mu$ and $\nu$. 
We obtain the solution to \eqref{eq:submeasures-inf} as a limit case to solutions for the discrete case.
%As for the previous subsection, the case of \eqref{eq:submeasures-sup} is discussed in Appendix~\ref{sect:approx-sup}.  

Consider $\overline{\mu}_{s}$ and $\overline{\nu}_{s}$, for $s\in\mathbb{N}$, 
discretized versions of $\mu$ and $\nu$ defined by
\begin{align}
    \overline{\mu}_{s} \left(\left\{x+2^{-(s+1)}\right\}\right) &= 2^{-s} \inf\left\{ f_{\mu}(z) : z\in\left[x, x+2^{-s}\right)\right\}, \quad \text{for every $x\in\{ 2^{-s}z : z\in\mathbb{Z}\}$},  \label{eq:discretization-mu}\\
    \overline{\nu}_{s} \left(\left\{y+2^{-(s+1)}\right\}\right) &= 2^{-s} \sup\left\{ f_{\nu}(z) : z\in\left[y, y+2^{-s}\right)\right\}, \quad \text{for every $y\in\{ 2^{-s}z : z\in\mathbb{Z}\}$}. \label{eq:discretization-nu}
\end{align}
For a discretization level $s$, mass of interval $[x,x+{2}^{-s})$ in $\mu$, rounded down according to the smallest value of $f_{\mu}$ on that interval, is attributed to the middle point of the interval in $\overline{\mu}_s$; same for $\overline{\nu}_s$, but the amount of mass is rounded up according to the largest value of $f_{\nu}$ on the interval. Note how the intervals are designed such they divide in two going from discretization level $s$ to $s+1$, and thus no two discretization levels share points for their support.    
If $\mu\leq_{\rm st}\nu$, then 
 $\overline{\mu}_{s} \leq_{st} \overline{\nu}_{s}$ holds true for every $s\in\mathbb{N}$; indeed, for every $z \in\mathrm{supp}(\overline{\mu}_{s})\cup \mathrm{supp}(\overline{\nu}_{s})$,
\begin{equation*}
    \overline{\mu}_{s}\left( \left[ z, \infty\right)\right) \leq \mu\left( \left[ z - 2^{-(s+1)}, \infty\right)\right) \leq \nu\left( \left[ z - 2^{-(s+1)}, \infty\right)\right) \leq \overline{\nu}_{s}\left( \left[ z , \infty\right)\right). 
\end{equation*}

% Consider $\overline{\mu}^{\dagger}_{\varepsilon}$ and $\overline{\nu}_{\varepsilon}$, for $\varepsilon>0$, 
% discretized versions of $\mu$ and $\nu$ defined by
% \begin{align*}
%     \overline{\mu}_{\varepsilon}^{\dagger} \left(\left\{x+\frac{\varepsilon}{2}\right\}\right) &= \varepsilon \inf\left\{ f_{\mu}(z) : z\in\left[x, x+\varepsilon\right)\right\}, \quad \text{for every $x\in\{ \ldots, -2\varepsilon, -\varepsilon, 0, \varepsilon, 2\varepsilon, \ldots \}$}, \\
%     \overline{\nu}_{\varepsilon} \left(\left\{y+\frac{\varepsilon}{2}\right\}\right) &= \varepsilon \inf\left\{ f_{\nu}(z) : z\in\left[y, y+\varepsilon\right)\right\}, \quad \text{for every $y\in\{ \ldots, -2\varepsilon, -\varepsilon, 0, \varepsilon, 2\varepsilon, \ldots \}$},
% \end{align*}
% and let 
%  $\overline{\mu}_{\varepsilon}$ be given by $\overline{\mu}_{\varepsilon}((x,\infty)) = \min(\overline{\mu}_{\varepsilon}^{\dagger}((x,\infty)), \overline{\nu}_{\varepsilon}((x,\infty)))$
% to ensure that the ordering constraint continues to hold for every discretization step $\varepsilon$. 

\begin{proposition}
\label{prop:convergence}
    Let $\mu, \nu \in\mathcal{M}$ have respective densities $f_{\mu}$ and $f_{\nu}$, and assume that $\mu\leq_{\rm st}\nu$ and $\mathrm{supp}(\nu)$ is compact. If $f_{\mu}$ and $f_{\nu}$ are piecewise Lipschitz continuous, then, as $s\to \infty$, for any $k\geq 0$, $Q_k^{\inf}(\overline{\mu}_{s},\overline{\nu}_{s}) \to Q_k^{\inf}(\mu,\nu)$.  
\end{proposition}

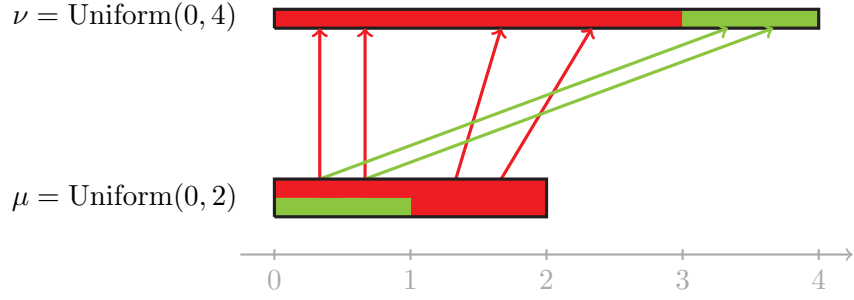
\begin{figure}[tb]
    \centering
 	\begin{minipage}{0.70\textwidth}
		\centering
			\begin{tikzpicture}[xscale=1.8]
			 \node at (-1.1, 0.75) [rectangle] {$\mu = \mathrm{Uniform}(0,2)$};
          \node at (-1.1, 3.15) [rectangle] {$\nu = \mathrm{Uniform}(0,4)$};
			\draw[Black!40,  ->, thick] (-0.25,0) -- (4.25, 0);
			
			\foreach \x in {0,1,...,4}
			\draw[Black!40, thick] (\x cm,-2pt) node[below]{\x} -- (\x cm,2pt);
			
			\filldraw[Red] (2.333,3) -- (2.333,3.25) -- (3, 3.25) -- (3,3) ;
			\filldraw[Red] (1.666,3) -- (1.666,3.25) -- (2.666, 3.25) -- (2.666,3) ;
			\filldraw[Red] (1,3) -- (1,3.25) -- (3, 3.25) -- (3,3) ;
			\filldraw[Red] (0.666,3) -- (0.666,3.25) -- (1, 3.25) -- (1,3) ;
			\filldraw[LimeGreen] (3.666,3) -- (3.666,3.25) -- (4, 3.25) -- (4,3) ;
			\filldraw[Red] (0.333,3) -- (0.333,3.25) -- (1, 3.25) -- (1,3) ;
			\filldraw[LimeGreen] (3.333,3) -- (3.333,3.25) -- (4, 3.25) -- (4,3) ;
			\filldraw[Red] (0,3) -- (0,3.25) -- (1, 3.25) -- (1,3) ;
			\filldraw[LimeGreen] (3,3) -- (3,3.25) -- (4, 3.25) -- (4,3) ;
			\draw[Black, very thick] (0,3) -- (0,3.25) -- (4, 3.25) -- (4,3) -- (0,3) ;

			\draw[Red, very thick, ->] (1.666, 1) -- (2.333, 3);
			\draw[Red, very thick, ->] (1.333, 1) -- (1.666, 3);
			\draw[Red, very thick, ->] (0.666, 1) -- (0.666, 3);
			\draw[LimeGreen, very thick, ->] (0.666, 1) -- (3.666, 3);
			\draw[Red, very thick, ->] (0.333, 1) -- (0.333, 3);
			\draw[LimeGreen, very thick, ->] (0.333,1) -- (3.333,3);

			\filldraw[Red] (1.666,0.5) -- (1.666,1) -- (2, 1) -- (2,0.5) ;
			\filldraw[Red] (1.333,0.5) -- (1.333,1) -- (2, 1) -- (2,0.5) ;
			\filldraw[Red] (1,0.5) -- (1,1) -- (2, 1) -- (2,0.5) ;
			\filldraw[Red] (0.666,0.75) -- (0.666,1) -- (1, 1) -- (1,0.75) ;
			\filldraw[LimeGreen] (0.666,0.75) -- (0.666,0.5) -- (1, 0.5) -- (1,0.75) ;
			\filldraw[Red] (0.333,0.75) -- (0.333,1) -- (1, 1) -- (1,0.75) ;
			\filldraw[LimeGreen] (0.333,0.75) -- (0.333,0.5) -- (1, 0.5) -- (1,0.75) ;
			\filldraw[Red] (0,0.75) -- (0,1) -- (0.5, 1) -- (0.5,0.75) ;
			\filldraw[LimeGreen] (0,0.75) -- (0,0.5) -- (0.5, 0.5) -- (0.5,0.75) ;			
			\draw[Black, very thick] (0,0.5) -- (0,1) -- (2, 1) -- (2,0.5) -- (0,0.5) ;
		\end{tikzpicture}	
	\end{minipage}
    \caption{Transport plan $\gamma_{1,\mu,\nu}^{\rm A}$ from Example~\ref{ex:continuous}.}
    \label{fig:D-Unif}
\end{figure}

Note that we require $f_{\mu}$ and $f_{\nu}$ to be piecewise Lipschitz continuous to have some regularity in approximating $\mu$ and $\nu$ by $\overline{\mu}_{s}$ and $\overline{\nu}_{s}$ as we let $s$ become large; piecewise Lipschitz continuity, combined with compactness of $\mathrm{supp}(\nu)$, also ensures that $\overline{\nu}_{s}$ has finite total mass for every discretization level $s\in\mathbb{N}$. 
The construction in \eqref{eq:discretization-mu}--\eqref{eq:discretization-nu} not only permits to approximate the value of $Q_k^{\inf}(\mu,\nu)$ with measures on which Algorithm~\ref{algo:A} is applicable, by Propositions~\ref{prop:atomization}--\ref{prop:convergence}, applications of Algorithm~\ref{algo:A} will moreover yield a consistent coupling as $s$ tends to infinity. 

\begin{theorem}
\label{th:continuous}
    Let $\mu,\nu\in\mathcal{M}$ with $\mu\leq_{\rm st}\nu$. Suppose that densities $f_{\mu}$ and $f_{\nu}$ are piecewise Lipschitz continuous and that $\mathrm{supp}(\nu)$ is compact. For any $k\geq 0$, there exists $\gamma \in \mathcal{H}(\mu,\nu)$ such that $\gamma_{k,\overline{\mu}_{s},\overline{\nu}_{s}
    }^{\rm A} \to \gamma$ weakly as $s\to \infty$. 
\end{theorem}

Define $\gamma_{k,\mu,\nu}^{\rm A} = \lim_{s\to\infty} \gamma_{k,\overline{\mu}_{s}, \overline{\nu}_{s}}^{\rm A}$ as per Theorem~\ref{th:continuous}. We illustrate the discretization approach in the following example. 
%Because $c$ in \eqref{eq:nutz-inf} is lower semicontinuous, the infimum is attained in \eqref{eq:submeasures-inf}, see \citet[Theorem 5.10]{V09}. %Therefore, $\gamma_{k,\mu,\nu}^{\rm A}$ is a transport plan yielding $Q^{\inf}_k(\mu,\nu)$.

\begin{example}
\label{ex:continuous}
    Let $\mu =\mathrm{Uniform}(0,2)$ and $\nu = \mathrm{Uniform}(0,4)$, and thus $\mu\leq_{\rm st}\nu$. Suppose $k=1$; we examine $\gamma_{1,\mu,\nu}^{\rm A}$. By taking small discretization steps, and after atomization, we notice that, for each atom in $\mu$ at a given location $1<x\leq 2$, Algorithm~\ref{algo:A} couples it to the two largest atoms in $\nu$, at locations $y_1,y_2$, such that $y_i-x\leq 1$, $i\in\{1,2\}$. Then, for atoms in $\mu$ located between 0 and 1, only half of the atoms at a given location can be coupled in $\nu$ such that $x=y$, and Algorithm~\ref{algo:A} couples the other half to the largest uncoupled location in $\nu$. As we let the discretization level $s\to\infty$, we obtain that $\gamma^{\rm A}_{1, \mu,\nu}$ is the transport plan illustrated in Figure~\ref{fig:D-Unif}. 
    In green are the transport portions such that $y-x > 1$; in red are the ones such that $y-x \leq 1$. By Proposition~\ref{prop:convergence}, we have $Q_1^{\rm inf}(\mu,\nu) = 0.25$. 
\end{example}

The transport plan from Example~\ref{ex:continuous} shows that in general $\gamma^{A}_{k,\mu,\nu}$ is not Monge-type. A transport plan is Monge-type if the corresponding transport map is described by a deterministic kernel. 
As illustrated by the double arrows in Figure~\ref{fig:D-Unif}, the kernel transporting from $\mu$ to $\nu$ in the example is two-atomic for $x\in [0,1]$, namely $x\mapsto (\delta_x + \delta_{3+x})/2$ on that portion.% (see \cite{NW22}, as $\mathcal{E}_k(\mu,\nu)$ is the DL coupling in this particular case). Note that $\mathcal{E}_k$ is not the DL coupling in general; one easily notices how $\mathcal{E}_k(\mu,\nu)$ performs better than ${\rm dl}(\mu,\nu)$ in the discrete example. The transport $\mathcal{D}_k$ is not Monge-type in general either. 

%\section{Limitations for non-binary treatments}
\section{Non-binary treatments}
\label{sect:NONBINARY}

Hitherto, we only discussed binary treatments, that is, treatments that may   be applied either fully or not, without gradation. One may inquire about the case where a participant may receive partial treatment, producing new marginal response distributions to account for. While still interested in the probability of worthwhile effect of the full treatment, we may impose the coupling to be monotone with respect to the partial-treatment responses as well. 
%Some hurdles, however, interfere in extending Algorithms~\ref{algo:A} and~\ref{algo:B} to handle these added constraints. 
Example~\ref{ex:constraints} below shows, in a simple case, how this constrains the admissible couplings for $X$ and $Y$.

\begin{example}
\label{ex:constraints}
Suppose that, for a given treatment, participants' responses to the treatment are $\mathbf{y} = \{2,4\}$ and responses to its absence are $\mathbf{x} = \{0,2\}$. 
The coupling multiset maximizing $\mathbb{P}(Y-X > 3)$ is $\{(2,2),(0,4)\}$: no effect half of the time, and an effect of 4 the other half. Now, suppose that we also observed responses to a partial treatment,  $\mathbf{z}=\{1,3\}$. This additional information, paired with monotonicity of the treatment, dispels the possibility that the full treatment may have no effect, and thus leaves $\{(0,2),(2,4)\}$ as the only feasible coupling multiset for full treatment. 
This is illustrated in Figure~\ref{fig:limitations}.
\label{ex:limitations}
\end{example}
    
    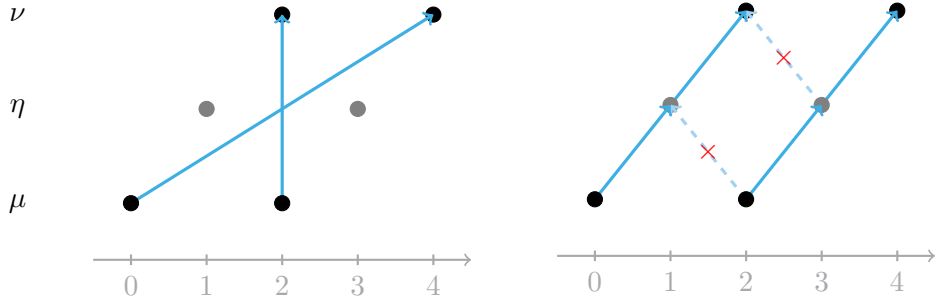
\begin{figure}[tb]
     \centering
	\begin{minipage}{0.4\textwidth}
		\centering
		\begin{tikzpicture}
        	\node at (-1.5, 0.75) [rectangle] {$\mu$};
            	\node at (-1.5, 2) [rectangle] {$\eta$};
		\node at (-1.5, 3.25) [rectangle] {$\nu$};
			\draw[Black!40,  ->, thick] (-0.5,0) -- (4.5, 0);
			
			\foreach \x in {0,1,...,4}
			\draw[Black!40,  thick] (\x cm,-2pt) node[below]{\x} -- (\x cm,2pt);
			
			\filldraw (4,3.25) circle (0.1 cm);
			\filldraw (2,3.25) circle (0.1 cm);

            \filldraw[black!50] (1,2) circle (0.1 cm);
			\filldraw[black!50] (3,2) circle (0.1 cm);
			
			\draw[CornflowerBlue, very thick, ->] (0, 0.75) -- (4, 3.25);
			\draw[CornflowerBlue, very thick, ->] (2, 0.75) -- (2, 3.25);

			\filldraw (0,0.75) circle (0.1 cm);
			\filldraw (2,0.75) circle (0.1 cm);
		
		\end{tikzpicture}	
	\end{minipage}
	\begin{minipage}{0.4\textwidth}
		\centering
				\begin{tikzpicture}
			\draw[Black!40,  ->, thick] (-0.5,0) -- (4.5, 0);
			
			\foreach \x in {0,1,...,4}
			\draw[Black!40,  thick] (\x cm,-2pt) node[below]{\x} -- (\x cm,2pt);
			
			\filldraw (4,3.25) circle (0.1 cm);
			\filldraw (2,3.25) circle (0.1 cm);

             \draw[CornflowerBlue!50, very thick, dashed, ->] (3, 2) -- node[red]{$\times$} (2, 3.25);
             \draw[CornflowerBlue, very thick, ->] (1, 2) -- (2, 3.25);
                \draw[CornflowerBlue, very thick, ->] (3, 2) -- (4, 3.25);

            \filldraw[black!50] (1,2) circle (0.1 cm);
			\filldraw[black!50] (3,2) circle (0.1 cm);
			
			\draw[CornflowerBlue!50, very thick, dashed, ->] (2, 0.75) -- node[red]{$\times$}(1, 2);
			\draw[CornflowerBlue, very thick, ->] (2, 0.75) -- (3, 2);

           \draw[CornflowerBlue, very thick, ->] (0, 0.75) -- (1, 2);

			\filldraw (0,0.75) circle (0.1 cm);
			\filldraw (2,0.75) circle (0.1 cm);
		\end{tikzpicture}	
	\end{minipage}
	\caption{The coupling in the left panel is no longer valid if one requires monotonicity of treatment response with respect to the partial treatment (whose responses' distribution is $\eta$) as well, as illustrated in the right panel.}
	\label{fig:limitations}
\end{figure}

In Example~\ref{ex:limitations}, we only considered one degree of partial treatment, but one could have more, leading to more layers of marginal distributions for which monotonicity has to be satisfied in the coupling. Let $Z_1,\ldots,Z_h$ be the responses to each of $h\in\mathbb{N}$ degrees of partial treatment, in increasing order of intensity, and their distributions are denoted by $\eta_1,\ldots, \eta_h$. Aligned with the assumptions regarding full treatment, we suppose that $\eta_1,\ldots,\eta_h$ are fully identified. Monotone treatment response translates to $X\leq Z_1\leq \cdots \leq Z_h \leq Y$. For this to be possible, it must hold that $\mu \leq_{\rm st} \eta_1 \leq_{\rm st} \cdots \leq_{\rm st} \eta_h \leq_{\rm st} \nu$.  In the general setting, we allow $\eta_1,\ldots,\eta_h\in\mathcal{M}$ to be non-probability measures, and this stochastic ordering must also hold. %As a result, we have $\mu(\mathbb{R})\leq \eta_1(\mathbb{R}) \leq \dots \leq \eta_h(\mathbb{R}) \leq \nu(\mathbb{R})$.\footnote{To consider cases where the reverse ordering holds, one may apply the negation transformation to all distributions, as in the bivariate case.} 

While we do not directly examine the effect of partial treatment, it gives valuable information by restricting the set of admissible measures in $\mathcal{H}(\mu,\nu)$. For any $h\in\mathbb{N}$, define 
\begin{equation*}
    \mathbb{H}_h = \left\{(x,z_1,\ldots,z_h,y)\in \R^{h+2} : x \leq z_1 \leq \cdots \leq z_h \leq y \right\}. 
\end{equation*}
Let $\Theta_{h}$ be the set of finite measures on $\mathbb{R}^{h+2}$ supported on $\mathbb{H}_h$.  Define
\begin{equation*}
    \mathcal{T}(\mu, \eta_1,\ldots,\eta_h,\nu) = \{\theta\in\Theta_h: ~ P_X(\theta) = \mu,~ P_{Z_{\ell}}(\theta) \leq \eta_{\ell},~ \ell\in\{1,\ldots,h\},~ P_Y(\theta) \leq \nu\}.
\end{equation*}
For $\theta\in \Theta_h$, we let $P_{XY}(\theta)$ denote the pairwise marginal of $\theta$ for the first and last dimensions. The set of admissible measures satisfying the additional partial-treatment constraints is given by
\begin{equation*}
    \mathcal{H}(\mu,\nu\mid\eta_1,\ldots,\eta_h) = \{\pi\in\Pi : \pi = P_{XY}(\theta) ~\text{for some}~ \theta\in\mathcal{T}(\mu, \eta_1,\ldots,\eta_h,\nu)\}. 
\end{equation*}
%
%Let $\mathcal{H}(\mu,\nu \mid \eta_1,\ldots,\eta_n)$ denote this restricted set of admissible measures: Considering the set of measures on $\mathbb{R}^{h+2}$ supported on $\mathbb{H}_h$ with marginals $\mu,\eta_1,\ldots,\eta_h,\nu$, then $\mathcal{H}_h(\mu,\nu \mid \eta_1,\ldots,\eta_n)$ is the set of their bivariate marginals in the dimension corresponding to $\mu$ and $\nu$. 
%on reduced set of feasible couplings, that is
%\begin{equation*}
%    \mathcal{H}(\mu,\nu\mid \eta_1,\ldots,\eta_n) = \{(X,Y) : X\lawis \mu, Y\lawis \nu, Z_1\lawis\eta_1, \ldots,Z_n\lawis\eta_d ,~ Y\geq Z_d \geq \cdots \geq Z_1 \geq X \}
%\end{equation*}
%
The worst-case and best-case probabilities of worthwhile effect become
\begin{align}
    Q^{\rm inf}_k(\mu,\nu\mid \eta_1,\ldots,\eta_h)
    &=\inf\left\{\int \id_{\{(x,y)\in\mathbb{R}^2: y-x>k\}}\mathrm{d}\pi : \pi\in \mathcal{H}(\mu,\nu \mid \eta_1,\ldots,\eta_h)\right\}; \label{eq:partialtreatment-inf}\\
    Q^{\rm sup}_k(\mu,\nu \mid  \eta_1,\ldots,\eta_h)
    &=\sup\left\{\int \id_{\{(x,y)\in\mathbb{R}^2: y-x>k\}}\mathrm{d}\pi : \pi\in \mathcal{H}(\mu,\nu \mid \eta_1,\ldots,\eta_h) \right\}. \label{eq:partialtreatment-sup}
\end{align}
Because $\mathcal{H}(\mu,\nu \mid \eta_1,\ldots,\eta_h) \subseteq \mathcal{H}(\mu,\nu)$, we have the following elementary bounds for the solutions of \eqref{eq:partialtreatment-inf}--\eqref{eq:partialtreatment-sup}:
\begin{equation*}
    Q^{\inf}_k(\mu,\nu) \leq Q^{\rm inf}_k(\mu,\nu \mid \eta_1,\ldots,\eta_h) \leq Q^{\rm sup}_k(\mu,\nu\mid  \eta_1,\ldots,\eta_h) \leq Q^{\sup}_k(\mu,\nu).
\end{equation*}
Modifying Algorithms~\ref{algo:A} and~\ref{algo:B} to consider the partial-treatment responses, we hereby present Algorithms~\ref{algo:A-eta} and~\ref{algo:B-eta}. Our second main result is that Algorithm~\ref{algo:A-eta} solves the minimization problem in \eqref{eq:partialtreatment-inf} and that Algorithm~\ref{algo:B-eta} solves the maximization problem in \eqref{eq:partialtreatment-sup}. We first discuss the intuition behind each.

Algorithm~\ref{algo:A-eta} employs the same heuristics as Algorithm~\ref{algo:A} in identifying optimal atom locations in $\nu$ to which each atom in $\mu$ should be coupled. Additional steps during each of these couplings assess whether the partial-treatment constraints may be satisfied when this optimal coupling yields $y-x\leq k$. If it is not possible, the algorithm will recommend to take the atoms in $ \eta_1,\ldots,\eta_h$ and $\nu$ at the largest locations. If it is possible, that atom in $\nu$ at this location is selected along the rightmost atoms in $\eta_1,\ldots,\eta_h$ that are to the left of the upper-degree treatment. Indeed, the directionality of monotone-treatment-response constraints makes atoms in $\eta_1,\ldots,\eta_h$ at the larger locations the more restrictive. The idea is therefore that committing the larger-location atoms in $\eta_1,\ldots,\eta_h$ to the coupling of a given atom in $\mu$ is the lesser nuisance for the atoms that are yet to be coupled. The main difficulty is to prove that, by this choice, an atom in $\mu$ satisfying the objective cannot prevent more than one 
atom from satisfying it because of these committed atoms in $\eta_1,\ldots,\eta_h$. This is the main part of the argument in the proof of Theorem~\ref{th:Q-inf-eta} in Appendix~\ref{sect:proofs}. If the algorithm recommends the largest atom in $\nu$, there is no need to verify whether the ordering constraint may be satisfied: it is guaranteed to be possible from stochastic ordering. We present an application of Algorithm~\ref{algo:A-eta} below.

Algorithms~\ref{algo:A-eta} and~\ref{algo:B-eta} differ more than just from the heuristic difference between Algorithms~\ref{algo:A} and~\ref{algo:B}; they also differ in how they verify that the partial-treatment constraint is satisfied. As stated above, the larger the location  of an atom in any of $\eta_1,\ldots,\eta_h$, the more restrictive it is in coupling $\mu$ to $\nu$. Algorithm~\ref{algo:A-eta} always tries to select the largest location for atoms in $\nu$, among either set of the atoms helping toward the objective or the others. Therefore, this selection mechanism flows in the same direction as the restrictiveness of the atoms in $\eta_1,\ldots,\eta_h$. Hence, we only need to check whether one atom in $\nu$, in the set of those meeting the objective, allows to satisfy the partial-treatment constraints; if it does not, the others in that set, at smaller locations, do not either. Algorithm~\ref{algo:B-eta}, by contrast, always aims to select atoms at smallest locations, whether it is among the set of those meeting the objective $y-x > k$ or the others, meeting $y\geq x$. Hence, its selection process is  opposite to the partial-treatment restrictiveness direction. As a result, if the favored atom in the set meeting the objective $y-x > k$ cannot satisfy the partial treatment constraint, the algorithm must check every other atom in that set, as they might satisfy it, before opting for the other set, for which the same process must be repeated. Algorithm~\ref{algo:B-eta} has therefore a higher degree of complexity than Algorithm~\ref{algo:A-eta}, from this opposite directionality.

\renewcommand{\thealgofloat}{$\mathrm{A}^{\eta}$}

\begin{algofloat}[t]
\centering

   \caption{}
    \fbox{
    \begin{minipage}{0.9\textwidth}
        %\textbf{Algorithm ${\rm A}^{\eta}$}\\
            \onehalfspacing
        For measures $\mu$, $\nu$ and $\eta_1, \ldots, \eta_h$, atomic with the same size $a>0$ and $\mu \leq_{\rm st} \eta_1 \leq_{\rm st} \dots \leq_{\rm st} \eta_h \leq_{\rm st} \nu$,
        let $x_1 \leq \cdots \leq x_m$ and $y_1 \leq \cdots \leq y_n$ be the ordered locations of atoms in $\mu$ and $\nu$ and, for each $\ell\in\{1, \ldots,h\}$, let $z_{\ell,1}\leq \cdots \leq z_{\ell, s_{\ell}}$ be the ordered locations of atoms in $\eta_{\ell}$. 
        \begin{itemize}
        \item[$\blacktriangleright$] Set $\mathcal{I}_m := \{1,\ldots,n \}$ and, for each $\ell \in \{1,\ldots,h\}$, set $\mathcal{S}_{\ell, m} := \{1,\ldots,s_{\ell} \}$.  
            \item[$\blacktriangleright$] For $i=m, m-1, \ldots, 1$, in  descending order, do:
            \begin{itemize}
                \item[$\triangleright$] Check whether $\mathcal{J}_i := \{j\in \mathcal{I}_i: x_i \leq y_j \leq x_i+k\}$ is empty. If it is non-empty, do:
                \begin{itemize}
                    \item[$>$] Set $d^*_i := \sup \mathcal{J}_i$.
                    \item[$>$] Set $c_{i,h} := \sup\{r\in\mathcal{S}_{h,i} : z_{h,r} \leq y_{d^*_i}\}$ (skip if the set is empty).
                    \item[$>$] For $\ell = h-1$, $h-2$, $\ldots$, 1, in descending order, do (skip if $h=1$): 
                    \begin{itemize}
                       \item[--] Set $c_{i,\ell} := \sup\{r\in \mathcal{S}_{\ell, i} : z_{\ell, r} \leq z_{\ell+1, c_{i,\ell+1}}\}$ (skip if the set is empty). 
                    \end{itemize}
                    \item[$>$] If $c_{i,1}$ was defined and, in which case, $z_{1,c_{i,1}}\geq x_i$, do: 
                    \begin{itemize}
                    \item[--] Set $d_i := d_i^*$.
                    \item[--] Define $\mathcal{S}_{\ell, i-1} := \mathcal{S}_{\ell, i}\backslash \{c_{i,\ell} \}$ for each $\ell\in\{1,\ldots,h\}$. 
                \end{itemize}
                \end{itemize}
                
                \item[$\triangleright$]  If $d_i$ is still undefined at this step, do: 
                \begin{itemize}
                    \item[$>$] Set $d_i := \sup \mathcal{I}_i$.
                    \item[$>$] Define $\mathcal{S}_{\ell, i-1} := \mathcal{S}_{\ell, i}\backslash \{\sup \mathcal{S}_{\ell, i}\}$ for each $\ell\in\{1,\ldots,h\}$.
                \end{itemize}
                \item[$\triangleright$] Define $\mathcal{I}_{i-1} := \mathcal{I}_i\backslash\{d_i\}$.
            \end{itemize}
            \item[$\blacktriangleright$] Return $Q^{\mathrm{A}^{\eta}}_k(\mu,\nu | \eta_1,\ldots,\eta_h) := a(\sum_{i=1}^m \id_{\{y_{d_i} - x_i > k\}})$. 
        \end{itemize}
    \end{minipage}
    }
    \label{algo:A-eta}
\end{algofloat}
% Assume that , or that $\mu(\mathbb{R})\geq \eta_1(\mathbb{R}) \geq \dots \geq \eta_h(\mathbb{R}) \geq \nu(\mathbb{R})$ and 
% but if neither of these orderings apply, the algorithm will not produce a coupling. (It will not leave some $x_i$'s out to accommodate $\eta$'s with lesser atoms.) 

\renewcommand{\thealgofloat}{$\mathrm{B}^{\eta}$}

\begin{algofloat}[t]
\centering

   \caption{}
    \fbox{
    \begin{minipage}{0.9\textwidth}
    \onehalfspacing
       For measures $\mu$, $\nu$ and $\eta_1, \ldots, \eta_h$, atomic with the same size $a>0$ and $\mu \leq_{\rm st} \eta_1 \leq_{\rm st} \dots \leq_{\rm st} \eta_h \leq_{\rm st} \nu$,
        let $x_1 \leq \cdots \leq x_m$ and $y_1 \leq \cdots \leq y_n$ be the ordered locations of atoms in $\mu$ and $\nu$ and, for each $\ell\in\{1, \ldots,h\}$, let $z_{\ell,1}\leq \cdots \leq z_{\ell, s_{\ell}}$ be the ordered locations of atoms in $\eta_{\ell}$. 
        \begin{itemize}
        \item[$\blacktriangleright$] Set $\mathcal{I}_m := \{1,\ldots,n \}$ and, for each $\ell \in \{1,\ldots,h\}$, set $\mathcal{S}_{\ell, m} := \{1,\ldots,s_{\ell} \}$.  
            \item[$\blacktriangleright$] For $i=m, m-1, \ldots, 1$, in  descending order, do:
            \begin{itemize}
                \item[$\triangleright$] For each element $j$ of $\mathcal{J}_i := \{j\in \mathcal{I}_i:  y_j \geq x_i\}$, do:
                \begin{itemize}
                    \item[$>$] Set $c_{i,j,h} := \sup\{r\in\mathcal{S}_{h,i} : z_{h,r} \leq y_j\}$ (skip if the set is empty).
                    \item[$>$] For $\ell = h-1$, $h-2$, $\ldots$, 1, in descending order, do (skip if $h=1$): 
                    \begin{itemize}
                       \item[--] Set $c_{i,j,\ell} := \sup\{r\in \mathcal{S}_{\ell, i} : z_{\ell, r} \leq z_{\ell+1, c_{i,j,\ell+1}}\}$ (skip if the set is empty).
                    \end{itemize}
                    \item[$>$] If $c_{i,j,1}$ was defined and, in which case, $z_{1,c_{i,j,1}}\geq x_i$, do:
                    \begin{itemize}
                        \item[--] If $y_j> x_i + k$, include $j$ in a set $\mathcal{J}_i^*$.
                        \item[--] Otherwise, include $j$ in a set $\mathcal{J}_i^{\prime}$.
                    \end{itemize} 
                \end{itemize}
                \item[$\triangleright$] Check whether $\mathcal{J}_i^{*}$ is empty.
                \item[$\triangleright$] If $\mathcal{J}_i^*$ is non-empty, set $d_i := \inf \mathcal{J}_i^*$.\\
                If $\mathcal{J}_i^*$ is empty, set $d_i := \inf \mathcal{J}_i^{\prime}$.
                \item[$\triangleright$]  Define $\mathcal{S}_{\ell, i-1} := \mathcal{S}_{\ell, i}\backslash \{c_{i,d_i, \ell}\}$ for each $\ell \in\{1,\ldots, h\}$.
                \item[$\triangleright$] Define $\mathcal{I}_{i-1} := \mathcal{I}_i\backslash\{d_i\}$.
            \end{itemize}
            \item[$\blacktriangleright$] Return $Q^{\mathrm{B}^{\eta}}_k(\mu,\nu | \eta_1,\ldots,\eta_h) := a(\sum_{i=1}^m \id_{\{y_{d_i} - x_i > k\}})$. 
        \end{itemize}
    \end{minipage}
    }
        \label{algo:B-eta}
\end{algofloat}

\begin{theorem}
\label{th:Q-inf-eta}
   Let $\mu,\eta_1,\ldots,\eta_h,\nu \in \mathcal{M}$ be atomic with the same size $a>0$, $\mu \leq_{\rm st} \eta_1 \leq_{\rm st} \dots \leq_{\rm st} \eta_h \leq_{\rm st} \nu$, and $k\geq 0$. It holds that
    \begin{enumerate}[label={\rm (\roman*)}, ref={\rm (\roman*)}]
        \item \label{item:A-eta}  $ Q^{\inf}_k(\mu,\nu\mid \eta_1,\ldots,\eta_h)=Q^{\mathrm{A}^{\eta}}_k(\mu,\nu\mid \eta_1,\ldots,\eta_h) $, the latter specified in 
    Algorithm~\ref{algo:A-eta};  
        \item \label{item:B-eta} $ Q^{\sup}_k(\mu,\nu\mid \eta_1,\ldots,\eta_h)=Q^{\mathrm{B}^{\eta}}_k(\mu,\nu\mid \eta_1,\ldots,\eta_h) $,  the latter specified in
    Algorithm~\ref{algo:B-eta}. 
    \end{enumerate} 
    Moreover, $Q^{\inf}_k(\mu,\nu \mid \eta_1,\ldots,\eta_h)$ and $Q^{\sup}_k(\mu,\nu \mid \eta_1,\ldots,\eta_h)$ take values in the set  $\{0,a,2a,\dots, m a\}$,
     where   $m$ is in \eqref{eq:mu-a-m}.
\end{theorem}

\begin{figure}[tb]
\centering
	\begin{minipage}{0.70\textwidth}
		\centering
		\begin{tikzpicture}
            \node at (-1.5, 0.75) [rectangle] {$\mu$};
            	\node at (-1.5, 2) [rectangle] {$\eta_1$};
		          \node at (-1.5, 3.25) [rectangle] {$\eta_2$};
          \node at (-1.5, 4.5) [rectangle] {$\nu$};
			\draw[Black!40,  ->, thick] (-0.5,0) -- (9.5, 0);
			
			\foreach \x in {0,1,...,9}
			\draw[Black!40,  thick] (\x cm,-2pt) node[below]{\x} -- (\x cm,2pt);

            \draw[LimeGreen, very thick] (7, 3.25) -- (9, 4.5);
			\draw[Red, very thick] (5, 3.25) -- (6, 4.5);
			\draw[LimeGreen, very thick] (4, 3.25) -- (4, 4.5);
			\draw[Red, very thick] (2, 3.25) -- (2, 4.5);

            \draw[LimeGreen, very thick] (7, 2) -- (7, 3.25);
			\draw[Red, very thick] (4, 2) -- (5, 3.25);
			\draw[Red, very thick] (2, 2) -- (2, 3.25);
            \draw[LimeGreen, very thick] (1, 2) -- (4, 3.25);

            \draw[LimeGreen, very thick] (6, 0.75) -- (7, 2);
			\draw[Red, very thick] (4, 0.75) -- (4, 2);
			\draw[Red, very thick] (1, 0.75) -- (2, 2);
            \draw[LimeGreen, very thick] (0, 0.75) -- (1, 2);
            
			\filldraw (2,4.5) circle (0.1 cm);
			\filldraw (4,4.5) circle (0.1 cm);
			\filldraw (6,4.5) circle (0.1 cm);
			\filldraw (9,4.5) circle (0.1 cm);

			\filldraw[black!50] (2,3.25) circle (0.1 cm);
			\filldraw[black!50] (4,3.25) circle (0.1 cm);
			\filldraw[black!50] (5,3.25) circle (0.1 cm);
			\filldraw[black!50] (7,3.25) circle (0.1 cm);

            \filldraw[black!50] (1,2) circle (0.1 cm);
			\filldraw[black!50] (2,2) circle (0.1 cm);
			\filldraw[black!50] (4,2) circle (0.1 cm);
			\filldraw[black!50] (7,2) circle (0.1 cm);

			\filldraw (0,0.75) circle (0.1 cm);
			\filldraw (1,0.75) circle (0.1 cm);
			\filldraw (4,0.75) circle (0.1 cm);
			\filldraw (6,0.75) circle (0.1 cm);
		\end{tikzpicture}	
	\end{minipage}
 	\caption{Illustration of the coupling multiset produced by Algorithm~\ref{algo:A-eta} in Example~\ref{ex:D-eta}, wherein $k=2$. In green are the vectors such that $y-x > 2$; in red are the ones such that $y-x \leq 2$. }
	\label{fig:D-eta}
\end{figure}
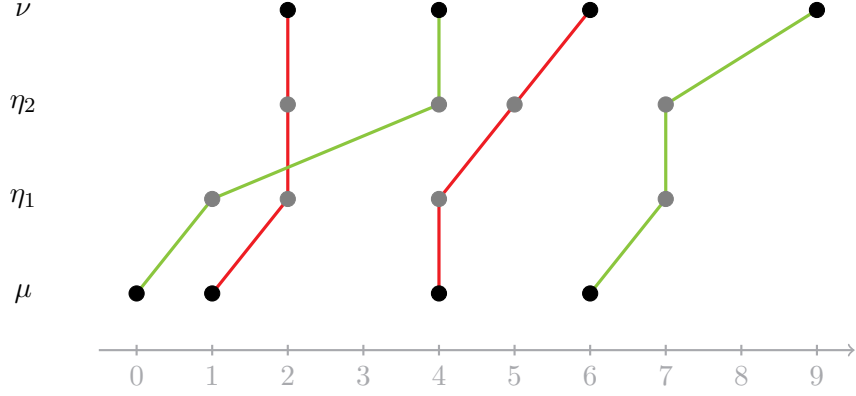

\begin{example}
\label{ex:D-eta}
    Consider probability measures $\mu$ for no-treatment responses, $\eta_1$ and $\eta_2$ for two degrees of partial-treatment responses, and $\nu$ for full-treatment responses. Each consists of equiprobable atoms at respective locations $\mathbf{x} = \{0,1,4,6 \}$, $\mathbf{z}_1 = \{1,2,4,7\}$, $\mathbf{z}_2 = \{2,4,5,7\}$ and $\mathbf{y} = \{2,4,6,9\}$. Let $k=2$. Algorithm~\ref{algo:A-eta} starts the coupling process with the atom at $x_4=6$. According to Algorithm~\ref{algo:A}, the preferred atom in $\nu$ should be $y_3 = 6$; however, Algorithm~\ref{algo:A-eta} needs to verify whether it can satisfy the partial-treatment constraint. Looking down, selecting atoms in $\eta_1$ and $\eta_2$ at largest admissible locations, we obtain $z_{1,c_{4,1}} = 4 < 6$, so it is not possible to satisfy the constraint. Therefore, instead, Algorithm~\ref{algo:A-eta} couples atom at $x_4 = 6$ to atom at $y_4 = 9$, the largest location, and commits atoms at $z_{1,4} = 7$ and $z_{2,4} = 7$ to this couple, effectively creating tuple (6,7,7,9) for the joint distribution $\pi\in\mathcal{H}(\mu,\nu\mid\eta_1,\eta_2)$. Then it is the turn of the atom at $x_3 = 4$ to be coupled. The favored atom in $\nu$ according to Algorithm~\ref{algo:A} would be $y_3 = 6$, and, after verification, it satisfies the partial-treatment constraint, by committing atoms at $z_{1,3} = 4$ and $z_{2,3} = 5$, so the algorithm selects this atom. The coupling multiset thus constructed is depicted in Figure~\ref{fig:D-eta}. Algorithm~\ref{algo:A-eta} yields $Q_2^{\mathrm{A}^{\eta}}(\mu,\nu\mid\eta_1,\eta_2) = 0.5$; we remark that it is larger than $Q_2^{\rm A}(\mu,\nu)$, without the partial-treatment constraint. By Theorem~\ref{th:Q-inf-eta}, we have $Q^{\inf}_2(\mu,\nu\mid\eta_1,\eta_2) = 0.5$. 
\end{example}

% As a pendant to Algorithm~\ref{algo:A-eta}, we present Algorithm~\ref{algo:B-eta}, adapting Algorithm~\ref{algo:B} to consider the partial-treatment constraints. Algorithm~\ref{algo:B-eta} solves the problem in \eqref{eq:partialtreatment-sup}.

%\begin{theorem}
% \label{th:Q-sup-eta}
%    Let $\mu,\eta_1,\ldots,\eta_h,\nu \in \mathcal{M}$ be atomic with same size and stochastically ordered.  Then, for any $k\geq 0$, it holds that
%    $
%     Q^{\rm B^{\eta}}_k(\mu,\nu|\eta_1,\ldots,\eta_h) = Q^{\sup}_k(\mu,\nu|\eta_1,\ldots,\eta_h)
% $.
% \end{theorem}

\begin{figure}[tb]
\centering
	\begin{minipage}{0.70\textwidth}
		\centering
		\begin{tikzpicture}
        \node at (-1.5, 0.75) [rectangle] {$\mu$};
            	\node at (-1.5, 2) [rectangle] {$\eta$};
          \node at (-1.5, 3.25) [rectangle] {$\nu$};
			\draw[Black!40,  ->, thick] (-0.5,0) -- (9.5, 0);
			
			\foreach \x in {0,1,...,9}
			\draw[Black!40,  thick] (\x cm,-2pt) node[below]{\x} -- (\x cm,2pt);

            \draw[Red, very thick] (7, 2) -- (8, 3.25);
			\draw[Red, very thick] (5, 2) -- (5, 3.25);
			\draw[LimeGreen, very thick] (4, 2) -- (6, 3.25);
            \draw[LimeGreen, very thick] (2, 2) -- (3, 3.25);

            \draw[Red, very thick] (6, 0.75) -- (7, 2);
			\draw[Red, very thick] (4, 0.75) -- (5, 2);
			\draw[LimeGreen, very thick] (2, 0.75) -- (4, 2);
            \draw[LimeGreen, very thick] (0, 0.75) -- (2, 2);

			\filldraw (3,3.25) circle (0.1 cm);
			\filldraw (5,3.25) circle (0.1 cm);
			\filldraw (6,3.25) circle (0.1 cm);
			\filldraw (8,3.25) circle (0.1 cm);

            \filldraw[black!50] (2,2) circle (0.1 cm);
			\filldraw[black!50] (4,2) circle (0.1 cm);
			\filldraw[black!50] (5,2) circle (0.1 cm);
			\filldraw[black!50] (7,2) circle (0.1 cm);

			\filldraw (0,0.75) circle (0.1 cm);
			\filldraw (2,0.75) circle (0.1 cm);
			\filldraw (4,0.75) circle (0.1 cm);
			\filldraw (6,0.75) circle (0.1 cm);
		\end{tikzpicture}	
	\end{minipage}
	\caption{Illustration of the coupling multiset produced by Algorithm~\ref{algo:B-eta} in Example~\ref{ex:E-eta}, wherein $k=2$. In green are the vectors such that $y-x > 2$; in red are the ones such that $y-x \leq 2$. }
	\label{fig:E-eta}
\end{figure}
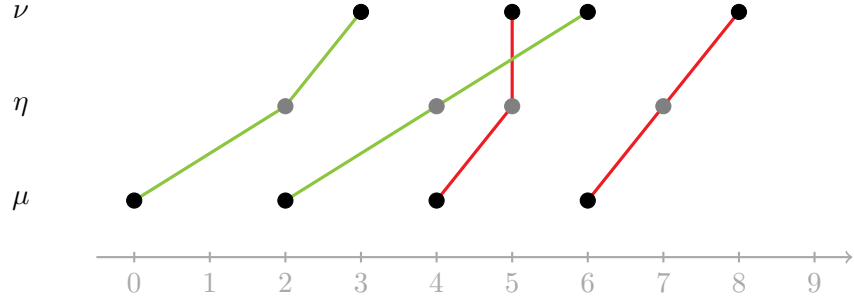

\begin{example}
\label{ex:E-eta}
Consider probability measures $\mu$  for no-treatment responses, $\eta$ for partial-treatment responses, and $\nu$ for full-treatment responses. Each comprises four atoms of the same size at respective locations $\mathbf{x}=\{0,2,4,6\}$, $\mathbf{z}=\{2,4,5,7\}$, and $\mathbf{y}=\{3,5,6,8\}$. Assume $k=2$. Algorithm~\ref{algo:B-eta} returns $Q_2^{\mathrm{B}^{\eta}}(\mu,\nu\mid\eta) = 0.5$; therefore, $Q_2^{\sup}(\mu,\nu\mid\eta) = 0.5$ by Theorem~\ref{th:Q-inf-eta}. In the beginning of the coupling process, when coupling the atom at $x_4=6$, Algorithm~\ref{algo:B-eta} selects atom at $y_4=8$ instead of the one at $y_3=6$, as would be recommended by Algorithm~\ref{algo:B}. The verification indeed yields $z_{c_{4,3}} = 4 < 6$, indicating that Algorithm~\ref{algo:B}'s choice would not allow to satisfy the partial-treatment constraints. The complete coupling multiset produced by Algorithm~\ref{algo:B-eta} is depicted in Figure~\ref{fig:E-eta}. Without the partial treatment, choosing to couple atoms at $x_4 = 6$ and $y_3 = 6$ would have allowed to couple atoms at $x_3 = 4$ and $y_4 = 8$, and atoms at $x_2 = 2$ and $y_2 = 5$, letting one more couple satisfy the objective $y-x> 2$; effectively, $Q_2^{\rm B}(\mu,\nu) = 0.75$.
\end{example}

By using the same atomization and discretization methods, we can also obtain results similar to those in Sections \ref{sec:discrete}--\ref{sec:cont} under partial-treatment constraints.  

The approach from Algorithms~\ref{algo:A-eta} and~\ref{algo:B-eta} does not extend to continuous treatments, where we would have a continuum of measures $\{\eta_u: u\in (0,1)\}$ to represent the degrees of partial treatment. Let us remark, however, that continuity of treatments conflicts with the assumption of fully identified marginal distributions (and no further structural assumptions) in practical applications: for continuous treatments, almost surely no two subjects will be administered the same level of treatment and thus only one point may be observed for each marginal distribution.

%\section{Numerical illustration}

\section{Conclusion}
\label{sect:CONCLUSION}

The probability of worthwhile effect for a treatment  requires, to be computable, identification of the joint distribution of treatment responses, unlike the average treatment effect, which is commonly employed but ill-suited to some contexts.  We examined the range of possible probabilities of worthwhile effect of a monotone-response treatment with the assumption that only marginal distributions are fully identifiable. With no further identification assumptions, Algorithms~\ref{algo:A} and~\ref{algo:B} yield the extremal probabilities of worthwhile effect for atomic distributions. Results for discrete and continuous distributions follow by convergence after atomization and discretization. Partial identification may come from responses to intermediate degrees of treatment, in which case Algorithms~\ref{algo:A-eta} and~\ref{algo:B-eta} produce the extremal probabilities of worthwhile effect. 
While we mainly grounded this paper in the statistical considerations of treatment effects, these results also transpose to other frameworks; they notably solve extant problems in optimal transport and risk aggregation under model uncertainty.

\subsection*{Acknowledgments}
 Ruodu Wang is supported by the Natural Sciences and Engineering Research Council of Canada (RGPIN-2018-03823, RGPAS-2018-522590) and Canada Research Chairs (CRC-2022-00141).

\appendix

% \section{Algorithms}

\setcounter{table}{0}
\setcounter{figure}{0}
\setcounter{equation}{0}
\renewcommand{\thetable}{S.\arabic{table}}
\renewcommand{\thefigure}{S.\arabic{figure}}
\renewcommand{\theequation}{S.\arabic{equation}}

\setcounter{theorem}{0}
\setcounter{proposition}{0}
\renewcommand{\thetheorem}{S.\arabic{theorem}}
\renewcommand{\theproposition}{S.\arabic{proposition}}
\setcounter{lemma}{0}
\renewcommand{\thelemma}{S.\arabic{lemma}}

\setcounter{corollary}{0}
\renewcommand{\thecorollary}{S.\arabic{corollary}}

\setcounter{remark}{0}
\renewcommand{\theremark}{S.\arabic{remark}}
\setcounter{definition}{0}
\renewcommand{\thedefinition}{S.\arabic{definition}}
\setcounter{example}{0}
\renewcommand{\theexample}{S.\arabic{example}}

\newpage
\begin{center}
\Large 
Supplementary Appendices
\\
``Probability of worthwhile effect of monotone-response treatments" 
\end{center}

~

There are six supplementary appendices. 
Appendix~\ref{sect:proofs}  presents all proofs.
Appendix~\ref{sect:dl}
explains the background on direction optimal transport.
Appendix~\ref{sect:CHEN} explains the corresponding robust risk aggregation problem to our main quantities of interest.
Appendix~\ref{sect:coupling-order} discusses an alternative ordering in the design of the main algorithms. 
Appendix~\ref{sect:approx-sup} 
presents the approximation arguments for the sup problem, complementing the results in Sections~\ref{sec:discrete}--\ref{sec:cont} stated for the inf problem.
Appendix~\ref{sect:duality} discusses duality based on some known results.

\section{Proofs}
\label{sect:proofs}

\subsection{Proof of Theorem~\ref{th:Q-inf}}

Let $\mathcal{X}= \{x_1\leq x_2\leq \dots \leq x_m\}$, $m\in\mathbb{N}$, be the multiset of the locations of the probability atoms in $\mu$, and $\mathcal{Y}= \{y_1\leq y_2\leq \dots \leq y_n\}$, $n\in\mathbb{N}$, be the multiset of the locations of the probability atoms in~$\nu$. Each of these atoms has size $a>0$. %As in the main text, we will write $x$ and $y$ (non indexed) to refer to generic elements of $\mathcal{X}$ and $\mathcal{Y}$, respectively; when referring to specific elements of either sets,  the symbols will be indexed. 

First, we claim that we can reduce our focus to couplings that are atomic with size $a$ (meaning that each atom in $\mu$ is coupled fully to only one atom in $\nu$) when arguing the optimality of the couplings produced by the algorithms. The last part of Theorem~\ref{th:Q-inf}'s statement follows directly from this observation.

For item~\ref{item:A}: Consider any $\pi_0\in\mathcal{H}(\mu,\nu)$ that is not atomic with size $a$; we can construct, as follows, a coupling $\pi_a$ atomic with size $a$ such that $\int \id_{\{y-x>k\}} \mathrm{d}\pi_a\leq \int \id_{\{y-x>k\}} \mathrm{d}\pi_0$. Let $K_0$ be the kernel describing $\pi_0$, meaning $\pi_0(\mathrm{d}x,\mathrm{d}y) = K_0(x,\mathrm{d}y)\mu(\mathrm{d}x)$. We proceed right to left in terms of atom locations $x_i\in\mathcal{X}$, $i\in\{m,\ldots,1\}$. Let $K_0^{[m+1]} = K_0$. At step $i$,   
couple $x_i$ to the largest $y\in\mathcal{Y}$, say $y^*$, such that $K_0^{[i-1]}(x,\{y^*\})>0$ and $y^*-x_i\leq k$, and if none of them satisfy this, simply the largest satisfying $K_0^{[i-1]}(x_i,\{y^*\})>0$. That is, we let $K_0^{[i]}(x_i,\cdot) = \delta_{y^*}$.  For all other $x\in\mathcal{X}$ such that $K_0^{[i-1]}(x,\{y^*\})>0$ as well, let \begin{equation}
K_0^{[i]}(x, \cdot) = K_0^{[i-1]}(x,\cdot)|_{\mathcal{Y}\backslash\{y^*\}} + \frac{K_0^{[i-1]}(x,\{y^*\})}{1-K_0^{[i-1]}(x_i,\{y^*\})}K_0^{[i-1]}(x_i,\cdot)|_{\mathcal{Y}\backslash\{y^*\}}; 
\label{eq:reattrib-inf}
\end{equation}
temporarily until their turn to be processed comes, so the coupling remains well-defined at every step; for the other $x\in\mathcal{X}$, let $K_0^{[i]}(x,\cdot) = K_0^{[i-1]}(x,\cdot)$.  
In other words, the locations impacted by the change of coupling in $x_i$ inherit, proportionally, the other locations to which $x_i$ was coupled. Note that because we proceed in this fashion right to left, all locations $x\in\mathcal{X}$ impacted by the change are smaller than $x_i$ (the larger ones have become atomic from this very process, at previous steps, and thus could not have been partially coupled to $y^*$). Because of this, it is clear that $\int \id_{\{y-x>k\}} K_0^{[i]}(x,\mathrm{d}y)\mu(\mathrm{d}x)\leq \int \id_{\{y-x>k\}} K_0^{[i-1]}(x,\mathrm{d}y)\mu(\mathrm{d}x)$: if $y^*-x_i\leq k$ then all of 
the portion of mass that newly leads the other altered $x$ to satisfy $y-x>k$ due to the alteration in coupling was already satisfying $y-x>k$ with $x_i$, because $y^*$ was taken to be the largest satisfying $y^*-x_i\leq k$; if $y^*-x_i>k$, then this means that all the mass from $\mu(\{x_i\})$ already belonged to $\{y-x>k\}$, and all other altered $x$ are now partially paired to smaller locations in $\mathcal{Y}$. 
Let $\pi_a$ be described by the kernel $K_0^{[1]}$; it is atomic with size $a$.    

For item~\ref{item:B}: Consider any $\pi_0^{\prime}\in\mathcal{H}(\mu,\nu)$ that is not atomic with size $a$; we can construct, as follows, a coupling $\pi_a^{\prime}$ atomic with size $a$ such that $\int \id_{\{y-x>k\}} \mathrm{d}\pi^{\prime}_a\geq \int \id_{\{y-x>k\}} \mathrm{d}\pi^{\prime}_0$. Let $K_0^{\prime}$ be the kernel characterizing $\pi_0^{\prime}$; as for the previous item, we transform through its kernel iteratively going from largest to lowest atom locations $x_i\in\mathcal{X}$, $i\in\{m,\ldots,1\}$. We start with $K_0^{\prime[m+1]} = K_0$. Then, for every $i\in\{m,\ldots,1\}$, at step $i$ let $K_0^{\prime[i]}(x_i,\cdot) = \delta_{y^*}$ where $y^*$ is the smallest $y\in\mathcal{Y}$ satisfying $K_0^{\prime[i-1]}(x_i,\{y\})>0$ and $y-x_i>k$, and if none does, the smallest $y\in\mathcal{Y}$ satisfying $K_0^{\prime[i-1]}(x_i,\{y\})>0$. For all other $x\in\mathcal{X}$ such that $K_0^{\prime [i-1]}(x,\{y^*\})>0$, let 
\begin{equation}
K_0^{\prime[i]}(x, \cdot) = K_0^{\prime[i-1]}(x,\cdot)|_{\mathcal{Y}\backslash\{y^*\}} + \frac{K_0^{\prime[i-1]}(x,\{y^*\})}{1-K_0^{\prime[i-1]}(x_i,\{y^*\})}K_0^{\prime[i-1]}(x_i,\cdot)|_{\mathcal{Y}\backslash\{y^*\}}; 
\label{eq:reattrib-sup}
\end{equation}
and for the other $x\in\mathcal{X}$, let $K_0^{[i]}(x,\cdot) = K_0^{[i-1]}(x,\cdot)$. Again, because we proceed right to left in $\mathcal{X}$, locations collaterally impacted at step $i$ are all smaller than $x_i$. From this, we can see that $\int \id_{\{y-x>k\}} K_0^{[i]}(x,\mathrm{d}y)\mu(\mathrm{d}x)\geq \int \id_{\{y-x>k\}} K_0^{[i-1]}(x,\mathrm{d}y)\mu(\mathrm{d}x)$: if $y^* - x_i > k$, then all the portion of mass that was previously coupling $x_i$ to other locations, and now couples the other altered $x$ to locations such that $y-x\leq k$, must already have been satisfying $y-x<k$ because $y^*$ was taken to be the smallest such that $y^* - x_i \leq k$; if $y^* - x_i \leq k$, then all other altered $x$ are then partially paired to larger locations $y\in\mathcal{Y}$. By letting $\pi_a^{\prime}$ be described by the kernel $K_0^{\prime[1]}$, we obtain the desired coupling that is atomic with size $a$.

% It is straightforward to see that $\mathcal{H}(\mu,\nu)$ is convex: for $\pi_1,\pi_2\in\mathcal{H}(\mu,\nu)$, $\lambda \pi_1 + (1-\lambda)\pi_2 \in\mathcal{H}(\mu,\nu)$ for all $\lambda\in(0,1)$. Because $\int \id_{\{(x,y)\in\mathbb{R}^2: y-x>k\}}\mathrm{d}\pi$ is linear in $\pi$, there must be some optimal couplings solving \eqref{eq:submeasures-inf} and \eqref{eq:submeasures-sup} that are extreme points of $\mathcal{H}(\mu,\nu)$, \textcolor{magenta}{and thus  }

%As stated in the main text, the algorithm is well-defined for general measures when $m\leq n$; we described above how a negation transformation of the random variables allow to employ the algorithm for the case $m>n$ so we will not argue it further. 

Next, 
we argue that the couplings produced by the algorithms are optimal among couplings that are atomic with size $a$.
Let $\mathcal{X}_{[i]}=\{x_i \leq \cdots \leq x_m\}\subseteq \mathcal{X}$, for $i\in\{1,\ldots,m\}$. Note that, as $i$ decreases, more atom locations are added to the left, so that $\{x_m\} = \mathcal{X}_{[m]} \subseteq \mathcal{X}_{[m-1]}  \subseteq \dots \subseteq \mathcal{X}_{[1]} = \mathcal{X}$. The algorithm is well-defined when changing $\mu$ to $\mu\vert_{\mathcal{X}_{[i]}}$. We write $\Gamma_\ell$ for the coupling multiset produced by the algorithm for $\mu\vert_{\mathcal{X}_{[\ell]}}$ and $\nu$, for $\ell\in\{i,i+1\}$. (This is $\Gamma^{\rm A}_{k,\mu\vert_{\mathcal{X}_{[\ell]}}, \nu}$ in the main text.) For both items~\ref{item:A} and~\ref{item:B}, we will prove the existence and optimality of the coupling multiset produced by the algorithm inductively on $i$: we prove for the base case $i=m$, and, by assuming it holds for $i+1$, we prove it for a general $i\in\{1,\ldots,m-1\}$. 

For item~\ref{item:A}: For the base step: Let $i=m$, then $\mathcal{X}_{[i]} = \{x_m\}$; existence and optimality are thus trivial. 
For the inductive step: Let $i\in\{1,\ldots,m-1\}$ and assume that the coupling multiset produced by the algorithm is valid and optimal for $\mu\vert_{\mathcal{X}_{[i+1]}}$ and $\nu$. Because the coupling process is made iteratively in decreasing locations of atoms in $\mu$, adding an atom leftmost will not change the coupled element of the other atoms, which are processed before. Therefore, $\Gamma_i = \Gamma_{i+1}\cup\{(x_i,y_d)\}$ where $d$ is determined by the algorithm.  
If $\Gamma_{i+1}$ is valid, then $\Gamma_i$ is valid as well: because $\mu\vert_{\mathcal{X}_{[i]}}((x,\infty)) \leq \nu((x,\infty))$ for all $x\in\mathbb{R}$, there is at least one uncoupled $y\in\mathcal{Y}$ such that $x_i\leq y$ and the algorithm will select among them. For optimality, let $\varphi_{\Gamma} = \sum_{(x,y)\in\Gamma} a\id_{\{y-x>k\}}$ for any coupling multiset $\Gamma$; we know that $\varphi_{\Gamma_{i+1}}$ is smallest for all coupling multisets of $\mathcal{X}_{[i+1]}$ and $\mathcal{Y}$, and we need to prove that $\varphi_{\Gamma_i}$ is  smallest for all coupling multisets of $\mathcal{X}_{[i]}$ and $\mathcal{Y}$. 

From the optimality of $\varphi_{\Gamma_{i+1}}$, we have that $\varphi_{\Gamma}\geq \varphi_{\Gamma_{i+1}}$ for all possible coupling multisets $\Gamma$ of $\mathcal{X}_{[i]}$ and $\mathcal{Y}$. For $\Gamma_i$, two situations may occur. First possibility: $d$ is such that $y_d - x_i \leq k$. Then $\varphi_{\Gamma_i} = \varphi_{\Gamma_{i+1}}$, and therefore $\varphi_{\Gamma_{i}}$ is smallest among all coupling multisets of $\mathcal{X}_{[i]}$ and $\mathcal{Y}$. Second possibility: $d$ is such that $y_d - x_i > k$. Then, $\varphi_{\Gamma_i} = \varphi_{\Gamma_{i+1}} + a$; we need to argue that no coupling multiset $\Gamma$ of $\mathcal{X}_{[i]}$ and $\mathcal{Y}$ could produce $\varphi_{\Gamma} = \varphi_{\Gamma_{i+1}}$. Suppose that there exists such a coupling multiset, say $\Gamma^{\dagger}$. By the optimality of $\varphi_{\Gamma_{i+1}}$, it must hold that $y_{d^{\dagger}} - x_i \leq k$ for $y_{d^{\dagger}}$ such that $(x_i,y_{d^{\dagger}})\in\Gamma^{\dagger}$; otherwise, we would have $\varphi_{\Gamma^{\dagger}\backslash \{(x_i, y_{d^{\dagger}})\}} = \varphi_{\Gamma_{i+1}} - a$. Therefore, $y_{d^{\dagger}}$ is an element of $\mathcal{Y}$ that was already coupled to another $x\in\Gamma_{i+1}$, that is $y_{d^{\dagger}}\in\{y:(x,y)\in \Gamma_{i+1}\}$. Indeed, if there was a then-uncoupled element of $\mathcal{Y}$ satisfying this, the algorithm would have selected it for $x_i$. Let $x^*$ be such that $(x^*, y_{d^{\dagger}})\in\Gamma_{i+1}$. Since, in $\Gamma^{\dagger}$, $x_i$ steals $y_{d^{\dagger}}$ from $x^*$, the location $x^*$ must be coupled to another $y\in\mathcal{Y}$. Because $x^*\geq x_i$, it must be that $y_{d^{\dagger}} - x^* \leq k$ as well. Hence, $x^*$ must find another element of $\mathcal{Y}$, say $y^*$, such that $y^*-x^* \leq k$, for $\varphi_{\Gamma^{\dagger}} = \varphi_{\Gamma_{i+1}}$ to hold. Since $\Gamma_{i+1}$ was constructed from the algorithm, $y_{d^{\dagger}}$ is the largest available $y\in\mathcal{Y}$ such that $y- x^*\leq k$. Hence, $y^* \leq y_{d^{\dagger}}$, but this $y^*$ must therefore also be coupled in $\Gamma_{i+1}$
by the same argument as before. 
Thus emerges a contradiction conflicting the existence of $\Gamma^{\dagger}$ as, from the arguments deployed, no $y\in\mathcal{Y}$ that was uncoupled in $\Gamma_{i+1}$ can be coupled in $\Gamma^{\dagger}$, but the latter requires one more element.

For item \ref{item:B}:
%The set up for this proof is the same as for the one of Theorem~1; we will reuse the notation. The argument for the proof is very similar in nature: for $\mu\vert_{\mathcal{X}_{[i]}}$, defined as therein, we prove existence and optimality of its coupling with $\nu$ as per Algorithm~\ref{algo:B} inductively on $i$. 
For the base step: Again, for $i=m$, we have $\mathcal{X}_{[i]} = \{x_m\}$, so existence and optimality are trivial. For the inductive step, we set $i\in\{1,\ldots,m-1\}$, and assuming that the coupling for $\mu\vert_{\mathcal{X}_{[i+1]}}$ and $\nu$ is valid and optimal, we argue that it is the case also for $\mu\vert_{\mathcal{X}_{[i]}}$ and $\nu$. As for Algorithm~\ref{algo:A}, Algorithm~\ref{algo:B} constructs the coupling iteratively, decreasingly in terms of elements of $\mathcal{X}$. Hence, adding an atom leftmost to $\mu\vert_{\mathcal{X}_[i+1]}$ will not change the coupled elements of the other atoms, which are processed before. Therefore, $\Gamma_i = \Gamma_{i+1}\cup\{(x_i, y_d)\}$, where $d$ is determined by the algorithm. Since $\Gamma_{i+1}$ is valid, $\Gamma_i$ is also valid, as the constraint $\mu\vert_{\mathcal{X}_{[i]}}((x,\infty)) \leq \nu((x,\infty))$ ensures that there remains at least one $y\in\mathcal{Y}$, uncoupled in $\Gamma_{i+1}$, such that $x_i\leq y$. This proves existence. For optimality, let $\psi_{\Gamma} = \sum_{(x,y)\in\Gamma} a\id_{\{y-x \leq k\}}$ for any coupling multiset $\Gamma$. The objective is to minimize $\psi_{\Gamma}$, so that $\sum_{(x,y)\in\Gamma} a\id_{\{y-x > k\}}$ is maximized, as the number of couples is fixed for compared coupling multisets. 

We have, from the inductive assumption, that $\psi_{\Gamma_{i+1}}$ is smallest for coupling multisets of $\mathcal{X}_{[i+1]}$ and $\mathcal{Y}$; we need to prove that $\psi_{\Gamma_i}$ is smallest for coupling multisets of $\mathcal{X}_{[i]}$ and $\mathcal{Y}$. It holds that $\psi_{\Gamma}\geq \psi_{\Gamma_{i+1}}$ for all possible coupling multisets $\Gamma$ of $\mathcal{X}_{[i]}$ and $\mathcal{Y}$. For $\Gamma_i$, two situations may occur. First situation: the $d$ selected by the algorithm is such that $y_d - x_i > k$. As a result, $\psi_{\Gamma_{i}} = \psi_{\Gamma_{i+1}}$ and $\psi_{\Gamma_i}$ is smallest for coupling multisets of $\mathcal{X}_{[i]}$ and $\mathcal{Y}$. Second situation: the $d$ selected by the algorithm is such that $y_d - x_i \leq k$, so $\psi_{\Gamma_i} = \psi_{\Gamma_{i+1}} + a$. We need to argue that no coupling multiset $\Gamma$ of $\mathcal{X}_{[i]}$ and $\mathcal{Y}$ could produce $\psi_{\Gamma} = \psi_{\Gamma_{i+1}}$. Suppose $\Gamma^{\dagger}$ is such a coupling multiset, and we will argue that it cannot exist. By optimality of $\psi_{\Gamma_{i+1}}$, it must hold that $y_{d^{\dagger}} - x_i > k$ where $y_{d^{\dagger}}\in \mathcal{Y}$ and $(x_i, y_{d^{\dagger}})\in\Gamma^{\dagger}$, because otherwise 
$\psi_{\Gamma^{\dagger}\backslash\{(x_i,y_{d^{\dagger}})\}} = \psi_{\Gamma_{i+1}} - a$, contradicting this optimality. We thus deduce that $y_{d^{\dagger}} \in \{y:(x,y)\in\Gamma_{i+1} \}$, it is already coupled in $\Gamma_{i+1}$. Indeed, if there was $y\in\mathcal{Y}$ satisfying $y-x_i> k$, the algorithm would have selected it as $y_d$. Let $x^*$ be such that $(x^*, y_{d^{\dagger}})\in\Gamma_{i+1}$. In $\Gamma^{\dagger}$, the location $x_i$ steals $y_{d^{ \dagger}}$ from $x^*$, so $x^*$ must be coupled to another element of $\mathcal{Y}$, say $y^*$. Now, recall that $(x^*, y_{d^{\dagger}})\in \Gamma_{i+1}$ was coupled by Algorithm~\ref{algo:B}. Two situations may have occurred. First, if $y_{d^{\dagger}} - x^* \leq k$, then elements $y$ of $\mathcal{Y}$ in $\Gamma_{i+1}$ such that $y-x^*> k$ or $x^*\leq y\leq y_{d^\dagger}$ were already coupled. Therefore, $y^*$ must either be stolen from another couple in $\Gamma_{i+1}$ or be an uncoupled $y\in[y_{d^{\dagger}}, k+x^*)$, the latter contradicting that there were no uncoupled $y_d$ such that $y_d - x_i > k$, as we assumed that $y_{d^{\dagger}} - x_i > k$.   
Second, if $y_{d^{\dagger}} - x^* > k$, then it must be that $y^* - x^*> k$ as well, so that $\psi_{\Gamma^{\dagger}} = \psi_{\Gamma_{i+1}}$. From the algorithm, all $y\in (x^*+k, y_{d^{\dagger}})\cap \mathcal{Y}$ are coupled in $\Gamma_{i+1}$, so $y^*$ must either be stolen from another couple in $\Gamma_{i+1}$ or be an uncoupled $y\in[y_{d^{\dagger}}, \infty)$. The latter contradicts that there were no uncoupled $y_d$ such that $y_d - x_i > k$, because we assumed that $y_{d^{\dagger}} - x_i > k$ . Combining these arguments, we have a contradiction for the construction of $\Gamma^{\dagger}$, as it does not allow an additional $y\in\mathcal{Y}$, uncoupled in $\Gamma_{i+1}$, to be coupled, although requiring one more element. \hfill $\square$

% Consider the same notation as in the proof of Theorem~\ref{th:Q-inf}.

% We prove the optimality of the coupling inductively on $m$ in a similar fashion; we add atoms in $\mu$ from largest location to smallest.

\subsection{Proof of Proposition~\ref{prop:separable}}

We have $Q^{\inf}_k(\mu\vert_{\Gamma^{\leq}},\nu\vert_{\Gamma^{\leq}}) = 0$ and thus it cannot be smaller.  
    If it were possible to reduce $Q^{\inf}_k(\mu\vert_{\Gamma^{>}},\nu\vert_{\Gamma^{>}})$, say with a coupling multiset $\Gamma^{\dagger}$, then the coupling multiset $\Gamma^{\dagger}\cup \Gamma^{\leq}$ would yield $Q^{\dagger}_k(\mu,\nu) < Q^{\rm A}_k(\mu,\nu)$, and this would contradict Theorem~\ref{th:Q-inf}. Because all couples in $\Gamma^{>}$ were constructed with the same principle of choosing the largest remaining location in $\nu$ for each atom in $\mu$, running the algorithm again on these atoms will produce the same result. Similarly, the coupling instruction for all couples in $\Gamma^{\leq}$ would be the same, that of choosing the largest remaining atom location such that $y-x\leq k$, and thus Algorithm~A produces the same coupling multiset on this restricted set of atoms also. A similar argument supports the claim for Algorithm~\ref{algo:B}. 
\hfill $\square$

\subsection{Proof of Theorem~\ref{th:stosto}}
\label{sect:lem-stosto}

This proof relies on Lemma~\ref{lem:construction}, stated below. 
    We first prove the attainability of the solution. If $\mu(\mathbb{R}) = \nu(\mathbb{R})$, then by the optimal transport representation in \eqref{eq:nutz-inf}, the fact that the cost function $c$ therein is lower semicontinuous, and \citet[Theorem 5.10]{V09}, the solution is attained. We nonetheless need to argue for the general case $\mu(\mathbb{R})\leq \nu(\mathbb{R})$. First, let us remark that $\mathcal{H}(\mu,\nu)$ is tight because, for every of its elements, the first marginal is fixed and the second marginal is dominated by $\nu$, which has finite total mass. Also, $\mathcal{H}(\mu,\nu)$ is weakly closed since all three of its defining properties are preserved under weak convergence: Consider any weakly converging sequence $\{\pi_j:j\in\mathbb{N}\}\subseteq\mathcal{H}(\mu,\nu)$, converging to $\pi\in\Pi$. Since $\mathrm{supp}(\pi_j)\subseteq \mathbb{H}$  for every $j\in\mathbb{N}$ and $\mathbb{H}$ is closed, then $\mathrm{supp}(\pi)\subseteq\mathbb{H}$. Since $P_1(\pi_j) = \mu$ and $P_1(\pi_j)\to P_1(\pi)$ weakly, then $P_1(\pi) = \mu$. Since $P_2(\pi_2) \leq \nu$ and $P_1(\pi_j)\to P_1(\pi)$ weakly, then $P_1(\pi) \leq \nu$.\footnote{The order $\leq$ is preserved under weak convergence: by Portmanteau's theorem, $\pi(A) \leq \liminf_{j->\infty} \pi_j(A) \leq \nu(A)$ for all open Borel set $A\subseteq\mathbb{R}^2$. } Combining tightness and weak closure yields weak compactness of $\mathcal{H}(\mu,\nu)$. 
    Consider any sequence $\{\pi_j: j\in\mathbb{N}\}\subseteq\mathcal{H}(\mu,\nu)$ such that $\liminf_{j\to\infty} \pi_j(\{y-x>k\}) = Q^{\inf}_k(\mu,\nu)$. By weak compactness, there exists a weakly converging subsequence $\{\pi_{j_\ell} : \ell \in\mathbb{N}\}$, converging to some $\pi\in\mathcal{H}(\mu,\nu)$. Because $\{(x,y)\in\mathbb{R}^2 : y-x>k\}$ is open, by Portmanteau theorem we have 
        \begin{equation*}
            \pi(\{y-x>k\}) \leq \liminf_{j\to\infty} \pi_j(\{y-x>k\}) = Q^{\inf}_k(\mu,\nu).
        \end{equation*}
    Also, $Q_{k}^{\inf}(\mu,\nu) \leq \pi(\{y-x>k\})$ since $\pi\in\mathcal{H}(\mu,\nu)$. Thus, $Q_k^{\inf}(\mu,\nu)$ is attained.

    % For the case $\mu(\mathbb{R}) < \nu(\mathbb{R})$, we argue by contradiction. Suppose the infimum is unattainable: because $(\mathbb{R}, \leq)$ is totally ordered and dense, this means that there is a sequence $\{\pi_1,\pi_2,\ldots \in \mathcal{H}(\mu,\nu)\}
    % $ approaching the optimal value, but $\lim_{i\to\infty} \pi_i \not\in \mathcal{H}(\mu,\nu)$. By Helly's selection theorem and because $P_2(\pi_i)\leq \nu$ for all $i\in\mathbb{N}$, there exists $\nu^*\in\mathcal{M}$ and a subsequence $P_2(\pi_{i_1}), P_2(\pi_{i_2}), \ldots$ such that $P_2(\pi_{i_j})\uparrow\nu^*$ as $j\to\infty$. Note that $\nu^*\leq\nu$, and hence $\mathcal{H}(\mu,\nu^*) \subseteq \mathcal{H}(\mu,\nu)$, and $\nu^*(\mathbb{R}) = \mu(\mathbb{R})$. %Since $\lim_{j\to\infty} \pi_{i_j}$ solves \eqref{eq:submeasures-inf} for $\mathcal{H}(\mu,\nu)$ and $\pi_{i_1}, \pi_{i_2}, \ldots \in \mathcal{H}(\mu,\nu^*)$, it must be that $\lim_{j\to\infty} \pi_{i_j}$ also solves \eqref{eq:submeasures-inf} but for $\mathcal{H}(\mu,\nu^*)$. 
    % Because $\mu(\mathbb{R}) = \nu^*(\mathbb{R})$ and the solution to \eqref{eq:submeasures-inf} is attained in that case, as argued just above, there must be some $\pi^{\dagger}\in\mathcal{H}(\mu,\nu^*)$ yielding the same value as the one approached by the sequence $\{\pi_1,\pi_2,\ldots \in \mathcal{H}(\mu,\nu)\}
    % $. But then, $\pi^{\dagger}\in\mathcal{H}(\mu,\nu)$ as well, contradicting that the solution could not be attained. 

We prove the second part of the claim. Define $\mathcal{M}_{\mu} = \{\theta \in\mathcal{M}:\theta(\mathbb{R}) = \mu (\mathbb{R})\}$, and note that, from the definition of $\mathcal{H}(\mu,\nu)$ in \eqref{eq:semicoupling}, it holds that $P_2(\pi)\in\Theta_{\mu}$ for all $\pi\in\mathcal{H}(\mu,\nu)$. 
The space $(\mathcal{M}_{\mu}, \leq_{\rm st})$ is a lattice; see \citet[Remark 1.A.18]{SS07}. 
Let 
\begin{equation*}
    \nu^* = \bigwedge_{\rm st}\{P_2(\pi) : \pi\in\mathcal{H}(\mu,\nu) \},
\end{equation*}
and it is well-defined because $(\mathcal{M}_{\mu}, \leq_{\rm st})$ is a lattice, $\leq_{\rm st}$ is closed under weak convergence, and a lower bound for $\nu^*$ is $\mu$.
Let us remark that, for any $\boldsymbol{\theta}=\{\theta_i:i\in\mathbb{N}\}\subseteq\mathcal{M}_{\mu}$, it holds that $\bigwedge_{\rm st} \boldsymbol{\theta} \leq \bigvee \boldsymbol{\theta}$; therefore, 
\begin{equation*}
    \nu^* =  \bigwedge_{\rm st}\{P_2(\pi) : \pi\in\mathcal{H}(\mu,\nu) \} \leq   \bigvee \{P_2(\pi) : \pi\in\mathcal{H}(\mu,\nu) \} \leq \nu. 
\end{equation*}
Thus, $\mathcal{H}(\mu,\nu^*) \subseteq \mathcal{H}(\mu,\nu)$. Because $\mu\leq_{\rm st} \nu^*$, the set $\mathcal{H}(\mu,\nu^*)$ is non-empty. 
Choose $\pi_1$ as an element of $\mathcal{H}(\mu,\nu^*)$ and $\pi_2$ as a coupling solving \eqref{eq:submeasures-inf} (recall that the solution is attained). 
The construction in Lemma~\ref{lem:construction} then yields a coupling $\pi^*\in\mathcal{H}(\mu,\nu)$ that satisfies $P_2(\pi^*) = \nu^*$, by Lemma~\ref{lem:construction}\ref{item:lem1-i2}, and also solves \eqref{eq:submeasures-inf}, by Lemma~\ref{lem:construction}\ref{item:lem1-i3}. 
\hfill$\square$

 Next we state and prove Lemma~\ref{lem:construction}. For any $\mu,\nu\in\mathcal{M}$ such that $\mu\leq_{\rm st}\nu$, define the reverse-DL-coupling as the unique $\pi^{\rm rev\text{-}dl}_{\mu,\nu}\in\mathcal{H}(\mu,\nu)$ 
 coupling, for every $x\in\mathbb{R}$, $\mu\vert_{(-\infty,x]}$ to $\nu_x$, the latter described by its cdf as such:
\begin{equation}
	F_{\nu_x} = \sup\{ F_\theta : \theta\in \mathfrak{S}^{\rm rev}_x\} \quad \text{where } \mathfrak{S}^{\rm rev}_x = \{\theta: \mu\vert_{(-\infty,x]} \leq_{\rm st} \theta \leq \nu\}. \label{eq:rev-dl}
\end{equation}
The difference with the DL-coupling introduced in~\cite{NW22} is that it is constructed by coupling $\mu$ to $\nu$ from left-to-right instead of right-to-left. On the DL-coupling, see Appendix~\ref{sect:dl}.
In the case where $\mu(\mathbb{R}) = \nu(\mathbb{R})$, the reverse-DL-coupling coincides with the comonotonic coupling.

       Consider $\mu,\nu\in\mathcal{M}$ such that $\mu\leq_{\rm st}\nu$ and any $\pi_1,\pi_2\in\mathcal{H}(\mu,\nu)$. Define the coupling $\pi^*$ as follows: 
\begin{enumerate}
\item Separate $\pi_2$ into $\pi_2 = \theta_2 + \tau_2$  where $\theta_2,\tau_2\in\mathcal{M}$, $\mathrm{supp}(\theta_2) \subseteq \{y-x\leq k\}$ and $\tau_2(\{y-x\leq k\}) = 0$; 
\item Define $\theta^* = \pi^{\rm rev\text{-}dl}_{P_1(\theta_2), P_2(\pi_1)\vee P_2(\pi_2)}$, which is the reverse-DL-coupling between $P_1(\theta_2)$ and $P_2(\pi_1)\vee P_2(\pi_2)$; 
\item Define $\tau^* = \pi^{\rm rev\text{-}dl}_{P_1(\tau_2), P_2(\pi_1) - P_2(\theta^*)\wedge P_2(\pi_1)}$, which is the reverse-DL-coupling between $P_1(\tau_2)$ and $P_2(\pi_1) - P_2(\theta^*)\wedge P_2(\pi_1)$;
\item Let $\pi^* = \theta^* + \tau^*$. 
\end{enumerate}

\begin{lemma}
    \label{lem:construction}
For the coupling $\pi^*$ above with $\mu,\nu\in\mathcal{M}$ such that $\mu\leq_{\rm st}\nu$ and   $\pi_1,\pi_2\in\mathcal{H}(\mu,\nu)$, the following holds:
        \begin{enumerate}[label={\rm (\roman*)}, ref={\rm (\roman*)}]
        \item \label{item:lem1-i1} $\pi^*\in\mathcal{H}(\mu,\nu)$; 
        \item \label{item:lem1-i2} $P_2(\pi^*) \leq_{\rm st} P_2(\pi_1)$;
        \item \label{item:lem1-i3} $\int \id_{\{y-x>k \}} \mathrm{d}\pi^* \leq \int \id_{\{y-x>k \}} \mathrm{d}\pi_2$.
        \end{enumerate}
\end{lemma}
\begin{proof}
 For item~\ref{item:lem1-i1}: 
 It is straightforward to see that $P_1(\theta_2) \leq_{\rm st} P_2(\pi_1)\vee P_2(\pi_2)$, and hence the coupling in Step 2 of the construction is well defined. 
 Because we take the supremum over all cdfs in \eqref{eq:rev-dl}, $\nu_x$ is the smallest element in terms of $\leq_{\rm st}$ over the set $\mathfrak{S}_x^{\rm rev}$ for every $x\in\mathbb{R}$. 
Hence, 
\begin{equation}
P_2(\theta^*)\leq_{\rm st} P_2(\theta) \quad \text{for all }\theta\in\mathcal{H}(P_1(\theta_2), P_2(\pi_1)\vee P_2(\pi_2)).  \label{eq:smallest-theta-star}
\end{equation}
Separate $\pi_1$ into $\pi_1 = \theta_1 + \tau_1$ where $P_1(\theta_1) = P_1(\theta_2)$ and $P_1(\tau_1) = P_1(\tau_2)$; note that, because $\pi_1$ is a valid coupling, we have
\begin{equation}
\label{eq:p1-tau1}
    P_1(\tau_1) \leq_{\rm st} P_2(\pi_1) - P_2(\theta_1)\wedge P_2(\pi_1).
\end{equation}
Because $P_2(\theta^*)\leq_{\rm st} P_2(\theta_1)$ by \eqref{eq:smallest-theta-star}, it follows from \eqref{eq:p1-tau1} that $P_1(\tau_2) \leq_{\rm st} P_2(\pi_1) - P_2(\theta^*)\wedge P_2(\pi_1)$. The coupling in Step~3 of the construction is therefore well-defined. By definition of the reverse-DL-coupling above, the supports of both $\theta^*$ and $\tau^*$ are comprised in $\mathbb{H}$; thus, $\mathrm{supp}(\pi^*)\subseteq \mathbb{H}$. It remains to verify the marginal constraints. For the first marginal, 
\begin{equation*}
    P_1(\pi^*) = P_1(\theta^*) + P_1(\tau^*) = P_1(\theta_2) + P_1(\tau_2) = P_1(\pi_2) = \mu.
\end{equation*}
For the second marginal, 
\begin{equation*}
    P_2(\pi^*) = P_2(\theta^*) + P_2(\tau^*) \leq P_2(\theta^*) + P_2(\pi_1) - P_2(\theta^*)\wedge P_2(\pi_1) \leq P_2(\pi_1)\vee P_2(\pi_2) \leq \nu.
\end{equation*}
 For item~\ref{item:lem1-i2}: By \eqref{eq:smallest-theta-star}, $P_2(\theta^*) \leq_{\rm st} P_2(\theta_1)$, and this means 
 \begin{equation}
 \label{eq:survival-theta}
     S_{P_2(\theta_1)}(y) - S_{P_2(\theta_*)}(y) \geq 0 \quad  \text{for all }y\in\mathbb{R}.
 \end{equation}
 If $P_2(\tau^*)\leq_{\rm st}P_2(\tau_1)$, the claim follows directly. Otherwise, we remark that, since $\nu_x$ is the smallest element in terms of $\leq_{\rm st}$ over the set $\mathfrak{S}_x^{\rm rev}$ for every $x\in\mathbb{R}$,
\begin{equation}
P_2(\tau^*)\leq_{\rm st} P_2(\tau) \quad \text{for all }\tau\in\mathcal{H}(P_1(\tau_1), P_2(\pi_1) - P_2(\theta^*)\wedge P_2(\pi_1)).  \label{eq:smallest-tau-star}
\end{equation}
As $\tau_1 \in\mathcal{H}(P_1(\tau_1), P_2(\pi_1) - P_2(\theta_1))$, by \eqref{eq:smallest-tau-star} all the possible stochastic improvement for $\tau^*$ over $\tau_1$ must come as a reflection of a stochastic improvement of $\theta_1$ over $\theta^*$, that is, 
\begin{equation*}
    S_{P_2(\tau^*)}(y) - S_{P_2(\tau_1)}(y) \leq S_{P_2(\theta_1)}(y) - S_{P_2(\theta^*)}(y), \quad \text{for all }y\in\mathbb{R}. 
\end{equation*}
We have, by this and \eqref{eq:survival-theta}, 
\begin{equation*}
    S_{P_2(\pi_1)}(y) - S_{P_2(\pi^*)}(y) \geq S_{P_2(\theta_1)}(y) - S_{P_2(\theta^*)}(y) - \left(S_{P_2(\tau)^*}(y) - S_{P_2(\tau_1)}(y) \right) \geq 0, \quad \text{for all }y\in\mathbb{R}.    
\end{equation*}
 For item~\ref{item:lem1-i3}: By \eqref{eq:smallest-theta-star}, $P_2(\theta^*)\leq_{\rm st}P_2(\theta_2)$; also recall that $P_1(\theta^*) = P_1(\theta_2)$.  If $\theta_2(\mathbb{R}^2) = 0$, then $\int \id_{\{y-x>k \}} \mathrm{d}\pi_2 = \mu(\mathbb{R})$ and the claim holds directly. Thus, assume $\theta_2(\mathbb{R}^2)>0$. Note that $\theta^*$ and $\theta_2$ have the same total mass, so we may examine them in terms of distributions of random vectors: let
 \begin{equation}
     (X^*,Y^*) \lawis \theta^*/\theta^*(\mathbb{R}^2) \text{ and } (X,Y) \lawis \theta_2/\theta_2(\mathbb{R}^2). 
 \end{equation}
Note that $X^*$ has the same distribution as $X$, and that $Y^*\leq_{\rm st}Y$.
Because $\theta^*$ is the reverse-DL-coupling, $(X^*,Y^*)$ are comonotonic. Let $(X^{\rm co}, Y^{\rm co})$ be a comonotonic version of $(X,Y)$, meaning components have the same respective distributions, but the dependence relation tying the vector is comonotonicity. By \citet[Corollary 3.28]{R13} (and \citet[Theorem~3]{CW26} if the means of $X$ and $Y$ are infinite),\footnote{Note that if $(X,Y)$ is comonotonic, $(-X,Y)$ is countermonotonic.} we have 
\begin{equation}
    Y^{\rm co}-X^{\rm co} \leq_{\rm cx} Y - X, 
\label{eq:cx-lemma1}
\end{equation}
where $\leq_{\rm cx}$ is the convex order.
Since we have $\mathbb{P}(Y-X \leq k) = 1$ by definition of $\theta_2$, and $\mathbb{P}(Y-X\geq 0)$ from the monotone-response support constraint,
\begin{equation*}
    \essinf(Y^{\rm co} - X^{\rm co}) \geq  \essinf(Y - X) \geq 0\quad \text{and} \quad \esssup(Y^{\rm co} - X^{\rm co}) \leq  \esssup(Y - X) \leq k,
\end{equation*}
where the relations between $\essinf$'s and $\esssup$'s follow from \eqref{eq:cx-lemma1}; see \citet[p.~111]{SS07}. Therefore, $\mathbb{P}(Y^{\rm co}-X^{\rm co}\leq k) = 1$. By comonotonicity and since $Y^*\leq_{\rm st} Y^{\rm co}$ and $X^*\laweq X^{\rm co}$ (meaning $F^{-1}_{Y^*}(u) - F^{-1}_{X^*} (u) \leq F^{-1}_{Y^{\rm co}}(u) - F^{-1}_{X^*} (u)$ for every $u\in(0,1)$), we have $Y^* - X^* \leq_{\rm st} Y^{\rm co} - X^{\rm co}$. Hence, 
\begin{equation*}
    \mathbb{P}(Y^* - X^*\leq k) \geq \mathbb{P}(Y^{\rm co} - X^{\rm co}\leq k) = 1.  
\end{equation*}
This translates to $\theta^*(\{y-x\leq k\}) = \theta_2(\{y-x\leq k\})$. Finally, 
\begin{align*}
 \int \id_{\{y-x \leq k \}} \mathrm{d}\pi^* &- \int \id_{\{y-x\leq k \}} \mathrm{d}\pi_2 \\ &= \theta^*(\{y-x\leq k\}) - \theta_2(\{y-x\leq k\}) + \tau^*(\{y-x\leq k\}) - \tau_2(\{y-x\leq k\}) \\ &= \tau^*(\{y-x\leq k\}) \geq 0,
\end{align*}
as $\tau_2(\{y-x\leq k\}) = 0$ by definition. The desired result follows directly given $\pi^*(\mathbb{R}^2) = \pi_2(\mathbb{R}^2)$. 
\end{proof}

\subsection{Proof of Proposition~\ref{prop:comono}}
	Consider $X\lawis \mu$ and $Y\lawis \nu$, and let $Z=-X$. If $(Y,Z)$ is counter-monotonic, then $Y+Z$ is the smallest element with respect to the convex order $\leq_{\rm cx}$; see \citet[Corollary 3.28]{R13}, and \citet[Theorem~3]{CW26} if the means of $X$ and $Y$ are infinite. %\textcolor{magenta}{this holds true even if means are undefined, see Theorem~3 of \cite{CW25}, so we do not have to worry about that.} 
    This signifies that $Y-X$ is smallest with respect to $\leq_{\rm cx}$ when $(X,Y)$ is comonotonic. Note that comonotonicity will always satisfy the constraint $Y\geq X$, since 
    $\mu \leq_{\rm st} \nu$.  The function $x\mapsto\id_{\{x=0\}}$ is convex on $[0,\infty)$, the support of $Y-X$.\footnote{The fact that it is not convex for a larger interval has no incidence in the case of $\leq_{\rm cx}$, see \cite{WW25}.} Therefore, $\mathbb{E}[\id_{\{Y-X=0\}}] = \mathbb{P}(Y-X = 0)$ is smallest when $(X,Y)$ is comonotonic. Comonotonic coupling attributes all its mass to the set $\{(F^{-1}_{\mu}(u),F^{-1}_{\nu}(u)): u\in(0,1)\}$. The integral in the proposition's statement thus corresponds to the portion of mass obligated to be on $\{(x,y): x=y\}$ because of the monotone-response requirement.   \hfill $\square$

\subsection{Proof of Theorem~\ref{th:tv}} 
This proof requires Lemmas~\ref{lem:inf-equality}--\ref{lem:decompo} below. 
Since $\mu\geq \mu^{\prime}$, it is straightforward to see that every element $\pi$ of $\mathcal{H}(\mu,\nu)$ is of the form 
\begin{equation}
    \pi = \pi^* + \pi^{\dagger} ~~\text{for some }\pi^*\in\mathcal{H}(\mu^{\prime},\nu)\text{ and }\pi^{\dagger}\in\mathcal{H}(\mu - \mu^{\prime}, \nu). \label{eq:decompo-mu}
\end{equation}
 Hence, 
 \begin{align}
    &Q^{\inf}_k(\mu ,\nu)\notag \\
    &= \inf\left\{\int \id_{\{y-x>k\}} \,\mathrm{d}(\pi^* +  \pi^{\dagger}): \pi^* + \pi^{\dagger} = \pi,~ \pi\in\mathcal{H}(\mu,\nu),~ \pi^*\in\mathcal{H}(\mu^{\prime},\nu),~ \pi^{\dagger}\in\mathcal{H}(\mu - \mu^{\prime}, \nu) \right\}\notag\\
    &\geq \inf\left\{\int \id_{\{y-x>k\}} \,\mathrm{d}\pi^* : \pi^* \leq \pi,~ \pi\in\mathcal{H}(\mu,\nu),~ \pi^*\in\mathcal{H}(\mu^{\prime},\nu)  \right\}\notag\\
    &= Q^{\inf}_k(\mu^{\prime}, \nu),  \label{eq:TV-Q-inf-sideA}
\end{align}
where the second equality follows from Lemma~\ref{lem:inf-equality}.  Also, 
\begin{align}
    &Q^{\inf}_k(\mu ,\nu)\notag \\
    &= \inf\left\{\int \id_{\{y-x>k\}} \,\mathrm{d}(\pi^* +  \pi^{\dagger}): \pi^* + \pi^{\dagger} = \pi,~ \pi\in\mathcal{H}(\mu,\nu),~ \pi^*\in\mathcal{H}(\mu^{\prime},\nu),~ \pi^{\dagger}\in\mathcal{H}(\mu - \mu^{\prime}, \nu) \right\}\notag\\
    &\leq \inf\left\{\int \id_{\{y-x>k\}} \,\mathrm{d}\pi^* : \pi^* \leq \pi,~ \pi\in\mathcal{H}(\mu,\nu),~ \pi^*\in\mathcal{H}(\mu^{\prime},\nu)  \right\} + \mu(\mathbb{R}) - \mu^{\prime}(\mathbb{R})\notag\\
    &= Q^{\inf}_k(\mu^{\prime}, \nu) + \mu(\mathbb{R}) - \mu^{\prime}(\mathbb{R}), \label{eq:TV-Q-inf-sideB}
\end{align}
where the inequality comes from noting that $\pi^{\dagger} \in \mathcal{H}(\mu-\mu^{\prime},\nu)$ implies  $\pi^{\dagger}(\mathbb{R}) = \mu(\mathbb{R}) - \mu^{\prime}(\mathbb{R})$, 
and the second equality follows from Lemma~\ref{lem:inf-equality}.  
Combining \eqref{eq:TV-Q-inf-sideA} and \eqref{eq:TV-Q-inf-sideB} yields 
\begin{equation}
0\leq Q^{\inf}_k(\mu ,\nu) - Q^{\inf}_k(\mu^{\prime} ,\nu) \leq \mu(\mathbb{R}) - \mu^{\prime}(\mathbb{R}).
    \label{eq:TV-Q-inf-goalA}
\end{equation}
Next, because $\nu \leq \nu^{\prime}$, we have $\mathcal{H}(\mu^{\prime},\nu)\subseteq \mathcal{H}(\mu^{\prime},\nu^{\prime})$, and so 
\begin{equation}
Q^{\inf}_k(\mu^{\prime},\nu) \geq Q^{\inf}_k(\mu^{\prime}, \nu^{\prime}).
\label{eq:TV-Q-inf-sideC}
\end{equation}
% By Lemma~\ref{lem:decompo}, 
% \begin{equation}
%     \mathcal{H}(\mu^{\prime},\nu^{\prime}) = \bigcup_{\substack{ \mu_1,\mu_2 \in \mathcal{M}\\ \mu_1 + \mu_2 = \mu^{\prime}}} \{\pi_1 + \pi_2 : \pi_1 \in \mathcal{H}(\mu_1, \nu),~\pi_2 \in \mathcal{H}(\mu_2,\nu^{\prime} - \nu)  \}.\label{eq:decompo-nu}
% \end{equation} 
% We require $\mu_2 \leq_{\rm st} \nu^{\prime} - \nu$ for $\mathcal{H}(\mu_2, \nu^{\prime}-\nu)$ to be nonempty; this implies 
% \begin{equation}
% \mu^{\prime}(\mathbb{R}) - \mu_1(\mathbb{R}) \leq \nu^{\prime}(\mathbb{R}) - \nu(\mathbb{R}).\label{eq:mu-mu-nu-nu}
% \end{equation}
Moreover, 
\begin{align}
    &Q^{\inf}(\mu^{\prime}, \nu^{\prime})\notag\\
    &= \inf\left\{ \int \id_{\{y-x>k\}}\,\mathrm{d}(\pi_1 + \pi_2): \pi_1 \in \mathcal{H}(\mu_1,\nu),~ \pi_2 \in \mathcal{H}(\mu_2,\nu^{\prime} - \nu),~\mu_1+\mu_2=\mu^{\prime},~\mu_1,\mu_2\in\mathcal{M} \right\}\notag\\
    &\geq \inf\left\{ \int \id_{\{y-x>k\}}\,\mathrm{d}\pi_1 : \pi_1 \in \mathcal{H}(\mu_1,\nu),\mu_2\leq_{\rm st} \nu^{\prime}- \nu,~\mu_1+\mu_2=\mu^{\prime},~\mu_1,\mu_2\in\mathcal{M} \right\}\notag\\
    &= \inf\left\{Q^{\inf}_k(\mu_1,\nu) : \mu_2\leq_{\rm st} \nu^{\prime}- \nu,~\mu_1+\mu_2=\mu^{\prime},~\mu_1,\mu_2\in\mathcal{M} \right\}\notag\\
    &\geq Q^{\inf}_k(\mu^{\prime},\nu) -\nu^{\prime}(\mathbb{R}) + \nu(\mathbb{R}),\label{eq:TV-Q-inf-sideD}
\end{align}
where the first equality follows by Lemma~\ref{lem:decompo} and  the last inequality follows from \eqref{eq:TV-Q-inf-sideB} and that  $\mu_2 \leq_{\rm st} \nu^{\prime} - \nu$ implies $\mu^{\prime}(\mathbb{R}) - \mu_1(\mathbb{R}) \leq \nu^{\prime}(\mathbb{R}) - \nu(\mathbb{R})$. By \eqref{eq:TV-Q-inf-sideC} and \eqref{eq:TV-Q-inf-sideD}, 
\begin{equation}
    0 \leq Q^{\inf}_k(\mu^{\prime},\nu) - Q^{\inf}_k(\mu^{\prime},\nu^{\prime}) \leq \nu^{\prime}(\mathbb{R}) - \nu(\mathbb{R}).  
    \label{eq:TV-Q-inf-goalB}
\end{equation}
Combining \eqref{eq:TV-Q-inf-goalA} and \eqref{eq:TV-Q-inf-goalB} yields the desired result.

\begin{lemma}
\label{lem:inf-equality}
Consider $\mu,\nu,\mu^{\prime}\in\mathcal{M}$ such that $\mu\leq_{\rm st}\nu$ and $\mu \geq \mu^{\prime}$. It holds that
\begin{equation}
    \label{eq:inf-equality}
Q^{\rm inf}_k(\mu^{\prime}, \nu) = \inf\left\{\int \id_{\{y-x>k \}}\,\mathrm{d}\pi^{\prime} : \pi^{\prime}\leq \pi,~\pi\in\mathcal{H}(\mu,\nu),~\pi^{\prime}\in\mathcal{H}(\mu^{\prime},\nu)  \right\}.
\end{equation}
\end{lemma}

\begin{proof}
By Theorem~\ref{th:stosto}, among the couplings solving $Q^{\inf}_k(\mu^{\prime}, \nu)$, one (say $\pi^*$) is such that $P_2(\pi^*) \leq_{\rm st} P_2(\pi^{\prime})$ for all $\pi^{\prime}\in\mathcal{H}(\mu^{\prime},\nu)$. Therefore, 
\begin{equation}
\nu - P_2(\pi^{\prime})\leq_{\rm st} \nu - P_2(\pi^*)~~\text{for all }\pi^{\prime}\in \mathcal{H}(\mu^{\prime},\nu).
\label{eq:nu-minus-p2pi}
\end{equation}
Note that the requirement in \eqref{eq:inf-equality} that $\pi^{\prime} \leq \pi$ for some $\pi\in\mathcal{H}(\mu,\nu)$ is the same as requiring $\pi-\pi^{\prime}\in \mathcal{H}(\mu - \mu^{\prime}, \nu - P_2(\pi^{\prime}))$ for some $\pi\in\mathcal{H}(\mu,\nu)$. 
Since $\mathcal{H}(\mu,\nu)$ is non-empty, $\mathcal{H}(\mu - \mu^{\prime}, \nu - P_2(\pi^{\prime}))$ must be non-empty for at least one $ \pi^{\prime}\in \mathcal{H}(\mu^{\prime},\nu)$. The necessary requirement for that is $\mu - \mu^{\prime} \leq_{\rm st} \nu - P_2(\pi^{\prime})$. By \eqref{eq:nu-minus-p2pi} and the transitivity of $\leq_{\rm st}$, it follows that $\mathcal{H}(\mu - \mu^{\prime}, \nu - P_2(\pi^*))$ is non-empty, and thus
 $\pi^*\in\{\pi^{\prime}\in\mathcal{H}(\mu^{\prime},\nu):\pi^{\prime}\leq \pi,~\pi\in\mathcal{H}(\mu,\nu)\}$. 
\end{proof}

\begin{lemma}\label{lem:decompo}
Consider $\mu^{\prime},\nu,\nu^{\prime}\in\mathcal{M}$ such that $\mu^{\prime}\leq_{\rm st}\nu$ and $\nu\leq \nu^{\prime}$. It holds that
\begin{equation}
  \mathcal{H}(\mu^{\prime},\nu^{\prime}) = \bigcup_{\substack{ \mu_1,\mu_2 \in \mathcal{M}\\ \mu_1 + \mu_2 = \mu^{\prime}}} \{\pi_1 + \pi_2 : \pi_1 \in \mathcal{H}(\mu_1, \nu),~\pi_2 \in \mathcal{H}(\mu_2,\nu^{\prime} - \nu)  \}.
    \label{eq:lem3-decompo}
\end{equation}
\end{lemma}

\begin{proof}
Let $\mathcal{H}^*$ denote the right-hand side of \eqref{eq:lem3-decompo}. It is straightforward to see that $\mathcal{H}^*\subseteq \mathcal{H}(\mu^{\prime}, \nu^{\prime})$. We argue that $\mathcal{H}(\mu^{\prime},  \nu^{\prime}) \subseteq \mathcal{H}^*$ by showing that each of its components can be expressed as a member of $\mathcal{H}^*$. 
    Because $\nu \leq \nu^{\prime}$, then $\nu \ll \nu^{\prime}$ and the Radon-Nikodym derivative $\mathrm{d}\nu/\mathrm{d}\nu^{\prime}$ is well-defined. Consider any $\pi\in\mathcal{H}(\mu^{\prime},  \nu^{\prime})$. Define $\pi_1,\pi_2\in\Pi$ as follows
    \begin{equation*}
        \pi_1(A) = \int_A \frac{\mathrm{d}\nu}{\mathrm{d}\nu^{\prime}}(y) \pi(\mathrm{d}x,\mathrm{d}y); \quad  \pi_2(A) = \int_A \frac{\mathrm{d}(\nu^{\prime}-\nu)}{\mathrm{d}\nu^{\prime}}(y)\pi(\mathrm{d}x,\mathrm{d}y), \quad \text{for every Borel set }A\subseteq\mathbb{R}^2.
    \end{equation*}
     Letting $\mu_1 = P_1(\pi_1)$ and $\mu_2 = P_1(\pi_2)$, we have $\pi = \pi_1 + \pi_2$ and $\pi_1 \in \mathcal{H}(\mu_1, \nu)$, $\pi_2 \in \mathcal{H}(\mu_2,\nu^{\prime} - \nu)$. Note also that $\mu_1 + \mu_2 = \mu^{\prime}$. Hence, $\pi\in\mathcal{H}^*$.
\end{proof}

\subsection{Proof of Proposition~\ref{prop:atomization}}

From the construction of $\widetilde{\mu}_{r}$ and $\widetilde{\nu}_r$ in \eqref{eq:atomization-mu}--\eqref{eq:atomization-nu}, it is clear that $ \widetilde{\mu}_{r} \leq \widetilde{\mu}_{r+1} \leq \mu$ and $\widetilde{\nu}_r \geq \widetilde{\nu}_{r+1} \geq \nu$ for every $r\in\mathbb{N}$, and $\widetilde{\mu}_r \leq_{\rm st} \widetilde{\nu}_r$ as we stated in the main text. 
We have $\widetilde{\mu}_r(\mathbb{R}) \uparrow \mu(\mathbb{R})$ and, because $|\mathrm{supp}(\nu)|<\infty$, also $\widetilde{\nu}_r(\mathbb{R}) \downarrow \nu(\mathbb{R})$ as $r\to\infty$; the claim follows directly by Theorem~\ref{th:tv}. 
\hfill $\square$

% As the supports $\mu$ and $\nu$ are subsets of $\mathbb{Z}$, there are at most $2n+1$ atoms' locations $x\in\mathbb{Z}$ in each such that $|x|\leq n$. For all such locations $x$, we have
% \begin{equation*}
%     \mu(\{x\}) - \widetilde{\mu}^{\dagger}_n(\{x\}) \leq \frac{1}{2^n}\quad\quad\text{and}\quad\quad \nu(\{x\}) - \widetilde{\nu}_n(\{x\}) \leq \frac{1}{2^n}. 
% \end{equation*}
% As such, 
% \begin{equation*}
%     \mu(\{x\}) - \widetilde{\mu}_n(\{x\}) \leq 2\frac{1}{2^n} = \frac{1}{2^{n-1}},\quad \text{for all }|x|\leq n. 
% \end{equation*}
% It is straightforward to see that the difference in optimal probabilities of worthwhile effect is smaller than the total variation distance between the two pairs of measures. As a result, 
% \begin{equation*}
%     \left| Q^{\inf}_k(\widetilde{\mu}_n,\widetilde{\nu}_n) - Q^{\inf}_k(\mu,\nu) \right| \leq (2n+1) \frac{1}{2^{n-1}} + (2n+1) \frac{1}{2^{n}} + \max\left(\int\id_{\{|x|\geq n\}}\,\mathrm{d}\mu,  \int\id_{\{|x|\geq n\}}\,\mathrm{d}\nu\right).
% \end{equation*}
% which tends to 0 as $n$ grows to infinity. 
% By Theorem~\ref{th:Q-inf}, for any $n\in\mathbb{N}$, we have $Q^{\rm A}_k(\widetilde{\mu}_n,\widetilde{\nu}_n) = Q^{\inf}_k(\widetilde{\mu}_n,\widetilde{\nu}_n)$; hence, $Q^{\rm A}_k(\widetilde{\mu}_n,\widetilde{\nu}_n) \to Q^{\inf}_k(\mu,\nu)$.
% A similar argument holds for $Q^{\sup}_k(\mu,\nu)$. 

\subsection{Proof of Theorem~\ref{th:discrete}}

Because $\widetilde{\mu}_1 \leq \widetilde{\mu}_2 \leq \dots \leq \mu$, measures $\widetilde{\mu}_r$, $r\in\mathbb{N}$, are all absolutely continuous with respect to $\mu$; their Radon-Nikodym derivatives $g_r := \mathrm{d}\widetilde{\mu}_r/\mathrm{d}\mu$, $r\in\mathbb{N}$, are thus well defined. For every $r\in\mathbb{N}$, we can write $\widetilde{\mu}_r = \int g_r \mathrm{d}\mu$, and we note that $g_r\uparrow 1$. 

For notational convenience, let us write $\widetilde{\gamma}_r := \gamma^{\rm A}_{k,\widetilde{\mu}_r, \widetilde{\nu}_r}$. We want to prove $\widetilde{\gamma}_r \to \gamma$ weakly as $r\to\infty$, for some $\gamma\in\mathcal{H}(\mu,\nu)$. Note that $\widetilde{\gamma}_r$ is well-defined for every $r\in\mathbb{N}$, because $\widetilde{\nu}_r(\mathbb{R})<\infty$ for every $r$ since $|\mathrm{supp}(\nu)|<\infty$. 
For every coupling $\widetilde{\gamma}_r$, $r\in\mathbb{N}$, let $K_r$ be the associated kernel, that is, $K_r$ is the probability measure such that  $\widetilde{\gamma}_r(\mathrm{d}x,\mathrm{d}y) = \mu_r(\mathrm{d}x) K_r(x,\mathrm{d}y)$. 

We will prove that, for $\widetilde{\mu}_r$-almost every $x\in\mathbb{R}$, 
\begin{equation}
K_r(x,\cdot) \longrightarrow K(x,\cdot) \quad \text{weakly},
    \label{eq:kernel-conv}
\end{equation}
for some kernel $K$. At atomization level $r\in\mathbb{N}$, for every $x\in\mathrm{supp}(\widetilde{\mu}_r) \subseteq \mathrm{supp}(\mu)$, Algorithm~\ref{algo:A} commands to couple atoms at $x$ to the largest remaining locations $y$ such that $y-x\leq k$, and if impossible, to the largest remaining locations $y$ without constraint. The measure $K_r(x,\cdot)\in\mathcal{M}$ is therefore of the form 
\begin{equation*}
    K_r(x,\cdot) = \theta_r(x,\cdot) + \tau_r(x,\cdot),
\end{equation*}
for $\theta_r(x,\cdot), \tau_r(x,\cdot)\in\mathcal{M}$ satisfying 
\begin{align}
    \theta_r(x,\cdot) &\leq \frac{1}{\widetilde{\mu}_r(\{x\})}\left(\widetilde{\nu}_r - \int_{(x,\infty)}K_r(x,\cdot)\widetilde{\mu}_r(\mathrm{d}x)\right) \quad \text{and} \quad \mathrm{supp}(\theta_r(x,\cdot))\subseteq\{y-x\leq k\};    
\label{eq:th4-constr-theta}\\
    \tau_r(x,\cdot) &\leq \frac{1}{\widetilde{\mu}_r(\{x\})}\left(\widetilde{\nu}_r - \int_{(x,\infty)}K_r(x,\cdot)\widetilde{\mu}_r(\mathrm{d}x)\right) \quad \text{and} \quad \tau_r(x,\{y-x\leq k\}) = 0,
\label{eq:th4-constr-tau}
\end{align}
and $\theta_r(x,\cdot)$ is the largest measure, with respect to $\leq_{\rm st}$, and with largest total mass, satisfying \eqref{eq:th4-constr-theta}, and $\tau_r(x,\cdot)$ is the largest measure, also with respect to $\leq_{\rm st}$, satisfying \eqref{eq:th4-constr-tau} and $\tau(x,\mathbb{R}) = 1- \theta_r(x,\mathbb{R})$. This stochastic-optimality requirement of each, for every $\mathrm{supp}(\mu)$, (and $\widetilde{\nu}_r\geq\widetilde{\nu}_{r+1}$ and $\widetilde{\mu}_r\leq\widetilde{\mu}_{r+1}$, 
\begin{equation}
    \frac{1}{\widetilde{\mu}_{r+1}(\{x\})}\left(\widetilde{\nu}_{r+1} - \int_{(x,\infty)}K_{r+1}(x,\cdot)\widetilde{\mu}_{r+1}(\mathrm{d}x)\right) \leq_{\rm st} \frac{1}{\widetilde{\mu}_r(\{x\})}\left(\widetilde{\nu}_r - \int_{(x,\infty)}K_r(x,\cdot)\widetilde{\mu}_r(\mathrm{d}x)\right).
    \label{eq:th4-improv-integr}
\end{equation}
Relation \eqref{eq:th4-improv-integr} directly implies, from the total-mass optimality requirement for $\theta_r(x,\cdot)$, 
\begin{equation}
\theta_{r+1}(x,\mathbb{R}) \leq \theta_{r}(x,\mathbb{R}), \quad \text{for every }r\in\mathbb{N},\text{for every }x\in\mathrm{supp}(\widetilde{\mu}_r). 
\label{eq:th4-total-theta}
\end{equation}
From this, we have $\theta_{r}(x,\mathbb{R})\downarrow a_x$ for some $a_x\geq 0$. Combining \eqref{eq:th4-constr-theta}, \eqref{eq:th4-improv-integr}, \eqref{eq:th4-total-theta} and the stochastic optimality requirement for $\theta_r(x,\cdot)$, we have
\begin{equation}
    \theta_{r+1}(x, \cdot) \leq_{\rm st} \theta_r(x,\cdot), \quad \text{for every }r\in\mathbb{N},\;x\in\mathrm{supp}(\widetilde{\mu}_r). \label{eq:stodom-theta}
\end{equation}
Then, for each $x\in\mathrm{supp}(\mu)$, because the sequence $\{\theta_r(x,\cdot):r\in\mathbb{N}\}$ is uniformly lower bounded, in the sense of $\leq_{\rm st}$, by $a_x\delta_x$ due to the monotone-response constraint, it must converge weakly:
    $\theta_r(x,\cdot) \to \theta(x,\cdot)$ weakly as $r\to\infty$,
for some $\theta(x,\cdot)\in\mathcal{M}$. 
The parallel argument for $\tau_r(x,\cdot)$ requires to employ an alternative definition of stochastic order: Define $\leq_{\rm st}^*$, a relation on $\mathcal{M}$,  by letting $\mu_1 \leq_{\rm st}^* \mu_2$ when $F_{\mu_1} \leq F_{\mu_2}$ for $\mu_1,\mu_2\in\mathcal{M}$, with $F_{\mu}$ being the cumulative distribution of $\mu\in\mathcal{M}$: $F_{\mu}(x) = \mu((-\infty,x])$. By \eqref{eq:th4-total-theta}, 
$\tau_{r+1}(x,\mathbb{R}) \geq \tau_{r}(x,\mathbb{R})$ for every $x\in\mathrm{supp}(\mu)$ and $r\in\mathbb{N}$. Combining this with \eqref{eq:th4-constr-theta} and \eqref{eq:th4-improv-integr} and the stochastic optimality requirement for $\tau_r(x,\cdot)$, we have 
\begin{equation}
    \tau_{r+1}(x, \cdot) \leq_{\rm st}^* \tau_r(x,\cdot), \quad \text{for every }r\in\mathbb{N},\;x\in\mathrm{supp}(\widetilde{\mu}_r). 
    \label{eq:stodom-tau}
\end{equation}
For each $x\in\mathrm{supp}(\mu)$, the sequence $\{\tau_r(x,\cdot):r\in\mathbb{N}\}$ is uniformly lower bounded, with respect to $\leq_{\rm st}^*$ by $\delta_{x+k}$. Hence it must converge weakly, meaning $\tau_r(x,\cdot) \to \tau(x,\cdot)$ weakly as $r\to\infty$, for some $\tau(x,\cdot)\in\mathcal{M}$. 
Letting $K(x,\cdot) = \theta(x,\cdot) + \tau(x,\cdot)$, we have 
$K_r(x,\cdot) \to K(x,\cdot)$ weakly, for every $x\in\mathrm{supp}(\widetilde{\mu}_r)$. 

We can finally argue for the weak convergence of $\{\widetilde{\gamma}_r$: $r\in\mathbb{N}\}$. For every bounded continuous function $\varphi:\mathbb{R}^2\to\mathbb{R}$, relation~\eqref{eq:kernel-conv} implies that 
\begin{equation*}
    \int \varphi(x,y) K_r(x,\mathrm{d}y) \longrightarrow \int \varphi(x,y) K(x,\mathrm{d}y), \quad \text{as }r\to\infty, \text{ for $\mu$-a.e. $x\in\mathbb{R}$.}
\end{equation*}
Therefore,
\begin{align}
\int \varphi(x,y)\mathrm{d}\widetilde{\gamma}_r &= \int \left( \int \varphi(x,y)K_r(x,\mathrm{d}y)\right) \mathrm{d}\widetilde{\mu}_r \notag\\
&= \int \left( \int \varphi(x,y)K_r(x,\mathrm{d}y)\right) g_r \mathrm{d}\mu
\longrightarrow  \int \left( \int \varphi(x,y)K(x,\mathrm{d}y)\right)  \mathrm{d}\mu \quad \text{as }r\to\infty,
\label{eq:weak-conv-GAMMA}
\end{align}
by dominated convergence since $g_r\uparrow 1$. Define $\gamma$ from the kernel $K$: let $\gamma(\mathrm{d}x,\mathrm{d}y) = \mu(\mathrm{d}x)K(x,\mathrm{d}y)$. This directly means that $P_1(\gamma) = \mu$, and it is straightforward to see that $P_2(\gamma_r)$ converges weakly to a measure dominated by $\nu$. Moreover, because $\mathrm{supp}(\widetilde{\gamma}_r)\subseteq \mathbb{H}$ for all $r\in\mathbb{N}$ and $\mathbb{H}$ is closed, we have $\mathrm{supp}(\gamma)\subseteq \mathbb{H}$; hence, $\gamma\in\mathcal{H}(\mu,\nu)$. Relation \eqref{eq:weak-conv-GAMMA} indicates that $\widetilde{\gamma}_r \to \gamma$ weakly as $r\to\infty$. \hfill$\square$

\subsection{Proof of Proposition~\ref{prop:convergence}}

As the discretization level changes, the difference between $\overline{\mu}_s$ and $\overline{\mu}_{s+1}$ or between $\overline{\nu}_{s}$ and $\overline{\nu}_{s+1}$ is twofold: the support changes and the amount of mass captured at each point varies. We want to account for each of these two sources separately. To this end, for every $t\leq s$, define $\widehat{\mu}_s^{(t)}$ and $\widehat{\nu}_s^{(t)}$ as follows:
\begin{align}
\widehat{\mu}_s^{(t)} = \sum_{i\in\mathbb{Z}} \frac{2^t}{2^s} \overline{\mu}_t(\{2^{-t}(i+1/2)\}) \sum_{x \in  \mathrm{supp}(\overline{\mu}_s) \cap  [2^{-t} i, 2^{-t} (i+1))} \delta_x; \label{eq:hat-mu}\\
\widehat{\nu}_s^{(t)} = \sum_{i\in\mathbb{Z}} \frac{2^t}{2^s} \overline{\nu}_t(\{2^{-t}(i+1/2)\}) \sum_{y \in  \mathrm{supp}(\overline{\nu}_s) \cap  [2^{-t} i, 2^{-t} (i+1))} \delta_y. \label{eq:hat-nu}
\end{align}
The measure $\widehat{\mu}_s^{(t)}$ has the same support as $\overline{\mu}_s$, but the mass captured is the same as for $\overline{\mu}_t$ and it is spread out evenly on the finer support; a similar interpretation holds for $\widehat{\nu}_s^{(t)}$. Uniformly for every $t\in\mathbb{N}$, as $s\to \infty$, we have $\widehat{\mu}_{s}^{(t)}\to \widehat{\mu}^{(t)}$ and $\widehat{\nu}_{s}^{(t)}\to \widehat{\nu}^{(t)}$ weakly where $\widehat{\mu}^{(t)},\widehat{\nu}^{(t)}\in\mathcal{M}$ have respective density functions
\begin{align}
   f_{\widehat{\mu}^{(t)}}(x) = \sum_{i\in\mathbb{Z}} 2^{t} \overline{\mu}_{t}\left(\left\{2^{-t}\left(i + 1/2\right)\right\} \right)\id_{[2^{-t} i, 2^{-t} (i+1))}(x), \quad \text{for all }x\in\mathbb{R}; \label{eq:hat-mu-dens}\\ 
    f_{\widehat{\nu}^{(t)}}(x) = \sum_{i\in\mathbb{Z}} 2^{t} \overline{\nu}_{t}\left(\left\{2^{-t}\left(i + 1/2\right)\right\} \right)\id_{[2^{-t} i, 2^{-t} (i+1))}(x), \quad \text{for all }x\in\mathbb{R}. \label{eq:hat-nu-dens}
\end{align}
First, we argue that, for every $s,t\in\mathbb{N}$, $t\leq s$, 
\begin{equation}
    Q_k^{\inf}\left(\widehat{\mu}_s^{(t)}, \widehat{\nu}_s^{(t)}\right) \geq Q^{\inf}_k\left(\widehat{\mu}^{(t)}, \widehat{\nu}^{(t)}\right). 
    \label{eq:prop4-goal1}
\end{equation}
For any $\pi^*_s$ solving $Q_k^{\inf}(\widehat{\mu}_s^{(t)}, \widehat{\nu}_s^{(t)})$ (recall that the solution is attained by Theorem~\ref{th:stosto}), let $K^*_s$ be the kernel describing $\pi^*_s$, that is, it satisfies $\pi^*_s(\mathrm{d}x,\mathrm{d}y) = K^*_s(x,\mathrm{d}y)\widehat{\mu}_s^{(t)}(\mathrm{d}x)$. Define the kernel $K^{*}$, by letting, for every $x\in\mathrm{supp}(\widehat{\mu}_{s})$ and every $\varepsilon \in [-2^{-(s+1)}, 2^{-(s+1)})$, 
\begin{align*}
    K^{*}(x + \varepsilon, \cdot) = \delta_{v_{x,\varepsilon}} \quad  \text{where } 
    v_{x,\varepsilon} &= q_{x,\varepsilon} - 2^{-(s+1)} + \frac{2^{s} (\varepsilon + 2^{-(s+1)}) - K^{*}_s(x,(-\infty,v)))}{K^*_s(x,\{v\})},\\
    \text{and }
    q_{x,\varepsilon} &= F^{-1}_{K^*_s(x, \cdot)}( 2^{s} (\varepsilon + 2^{-(s+1)})).
\end{align*}
Then, $K^{*}$ produces a valid coupling $\pi^{*}\in\mathcal{H}(\widehat{\mu}^{(t)}, \widehat{\nu}^{(t)})$ and $\int \id_{\{y-x>k\}}\mathrm{d}\pi^{\dagger} = \int \id_{\{y-x>k\}}\mathrm{d}\pi^{*} =  Q_k^{\inf}(\widehat{\mu}_s^{(t)}, \widehat{\nu}_s^{(t)})$; thus, \eqref{eq:prop4-goal1} holds. 

Next, for any $\pi^{\dagger}$ solving $Q_k^{\inf}(\widehat{\mu}^{(t)},\widehat{\nu}^{(t)})$ (again, the solution is attained, by Theorem~\ref{th:stosto}), let $K^{\dagger}$ be the kernel characterizing $\pi^{\dagger}$. For every $s\in\mathbb{N}$, define the kernel $K_s^{\dagger}$ by letting, for every $x\in\mathrm{supp}(\overline{\mu}_s)$,
\begin{equation}
K^{\dagger}_s(x, \{y\}) = 2^{s}\int_{x-2^{-(s+1)}}^{x+2^{-(s+1)}} K^{\dagger}( u , [y-2^{-(s+1)}, y+2^{-(s+1)}))  \mathrm{d}u
\label{eq:kernel-dagger}
\end{equation}
Notably because, for each $x\in\mathrm{supp}(\overline{\mu}_s)$, 
\begin{equation}
    \widehat{\mu}^{(t)}_s(\{x\}) = \mu^{(t)}([x-2^{-(s+1)}, x+2^{-(s+1)})),
    \label{eq:mu-parenthese-t-relation}
\end{equation}
we can see that $K^{\dagger}_s$ yields a coupling $\pi^{\dagger}_s\in \mathcal{H}(\widehat{\mu}^{(t)}_s, \widehat{\nu}^{(t)}_s)$. 
Let us remark that, for any $x,y\in\mathbb{R}$, if $y-x > k+ 2^{-s}$, then
\begin{equation*}
    y^* - x^* > k \text{ for all }x^*\in[x-2^{-(s+1)}, x+2^{-(s+1)}),\;y^*\in[y-2^{-(s+1)}, y+2^{-(s+1)}).
\end{equation*}
Combining this observation with the construction in \eqref{eq:kernel-dagger}, and by  \eqref{eq:prop4-goal1} and \eqref{eq:mu-parenthese-t-relation}, we have
\begin{align*}
    \int\id_{\{y-x>k\}} K_s^{\dagger}(x, \mathrm{d}y) \mu^{(t)}_s(\mathrm{d}x) - \int\id_{\{y-x>k\}} K^{\dagger}(x, \mathrm{d}y) \mu^{(t)}(\mathrm{d}x)\leq \int_{A_s} K^{\dagger}(x, \mathrm{d}y) \mu^{(t)}(\mathrm{d}x)=\pi^{\dagger}(A_s).
\end{align*}
where $A_s = \{(x,y)\in\mathbb{R}^2 : k < y-x \leq k + 2^{-s} \}$, $s\in\mathbb{N}$. 
Because $\pi^{\dagger}$ was a solution to $Q^{\inf}_k(\widehat{\mu}^{(t)},\widehat{\nu}^{(t)})$, it follows that, for every $s\in\mathbb{N}$,   
\begin{equation}
    0 \leq Q_k^{\inf}(\widehat{\mu}_s^{(t)}, \widehat{\nu}_s^{(t)}) - Q_k^{\inf}(\widehat{\mu}^{(t)}, \widehat{\nu}^{(t)})  \leq \pi^{\dagger}(A_s), \label{eq:prop4-goal-15}
\end{equation}
where the lower bound is \eqref{eq:prop4-goal1}.
Write $\Pi^{(t)}_{\rm sol}$ for the set of couplings solving $Q^{\inf}_k(\widehat{\mu}^{(t)}, \widehat{\nu}^{(t)})$. Since the choice of $\pi^{\dagger}$ was arbitrary among $\Pi^{(t)}_{\rm sol}$, bounds in \eqref{eq:prop4-goal-15} become
\begin{equation}
    0 \leq Q_k^{\inf}(\widehat{\mu}_s^{(t)}, \widehat{\nu}_s^{(t)}) - Q_k^{\inf}(\widehat{\mu}^{(t)}, \widehat{\nu}^{(t)})  \leq \inf\left\{ \pi(A_s) : \pi\in\Pi^{(t)}_{\rm sol}\right\}. 
    \label{eq:prop4-goal-16}
\end{equation}
For every $s,t\in\mathbb{N}$, $t\leq s$, we have 
\begin{align}
    \inf\left\{ \pi(A_s) : \pi\in\Pi^{(t)}_{\rm sol}\right\} &= \inf\left\{\int \id_{\{y-x>k\}}\mathrm{d}\pi - \int \id_{\{y-x>k+2^{-s}\}}\mathrm{d}\pi : \pi\in\Pi_{\rm sol}^{(t)} \right\} \notag \\
    &= Q_k^{\inf}\left(\widehat{\mu}^{(t)}, \widehat{\nu}^{(t)}\right) - \sup\left\{\int \id_{\{y-x>k+2^{-s}\}}\mathrm{d}\pi : \pi\in\Pi_{\rm sol}^{(t)} \right\} \notag \\
    &\leq Q_k^{\inf}\left(\widehat{\mu}^{(t)}, \widehat{\nu}^{(t)}\right) - Q_{k+2^{-s}}^{\inf}\left(\widehat{\mu}^{(t)}, \widehat{\nu}^{(t)}\right).
    \label{eq:prop4-goal-18}
\end{align}
It is straightforward to see that, for any $r\in\mathbb{R}$, and measures $\theta_1,\theta_2\in\mathcal{M}$, with $\theta_1\leq_{\rm st}\theta_2$, we have $Q_{k+r}^{\inf}(\theta_1, \theta_2^{[+r]}) \leq Q_{k}^{\inf}(\theta_1, \theta_2)$ where $\theta_2^{[+k]}(A)=\theta(A-k)$ for every Borel set $A$ on $\mathbb{R}$ (this is the same definition as in the main text), because the translation makes the monotone-response constraint less stringent. From this observation, we have
\begin{align}
    Q_k^{\inf}\left(\widehat{\mu}^{(t)}, \widehat{\nu}^{(t)}\right) - Q_{k+2^{-s}}^{\inf}\left(\widehat{\mu}^{(t)}, \widehat{\nu}^{(t)}\right) &\le  Q_{k+2^{-s}}^{\inf}\left(\widehat{\mu}^{(t)}, \widehat{\nu}^{(t)[+2^{-s}]}\right) - Q_{k}^{\inf}\left(\widehat{\mu}^{(t)}, \widehat{\nu}^{(t)}\right)\notag\\
    &\leq Q_k^{\inf}\left(\widehat{\mu}^{(t)}, \widehat{\nu}^{(t)[+2^{-s}]}\right) - Q_{k}^{\inf}\left(\widehat{\mu}^{(t)},  \widehat{\nu}^{(t)}\vee\widehat{\nu}^{(t)[+2^{-s}]}\right) \notag\\
    &\leq \widehat{\nu}^{(t)}\vee\widehat{\nu}^{(t)[+2^{-s}]}(\mathbb{R}) - \widehat{\nu}^{(t)}(\mathbb{R}) \label{eq:Q-Q-Q-Q} 
\end{align}
where the first inequality is because $\mathcal{H}(\widehat{\mu}^{(t)}, \widehat{\nu}^{(t)[+2^{-s}]}) \subseteq \mathcal{H}(\widehat{\mu}^{(t)}, \widehat{\nu}^{(t)}\vee\widehat{\nu}^{(t)[+2^{-s}]})$ since $\widehat{\nu}^{(t)[+2^{-s}]}\leq\widehat{\nu}^{(t)}\vee\widehat{\nu}^{(t)[+2^{-s}]}$, and the second inequality is by Theorem~\ref{th:tv} (and $\widehat{\nu}^{(t)}(\mathbb{R}) = \widehat{\nu}^{(t)[+2^{-s}]}(\mathbb{R})$).
Moreover, note that the discretization construction makes it so that $\widehat{\nu}^{(t)}\vee\widehat{\nu}^{(t)[+2^{-s}]}(\mathbb{R}) - \widehat{\nu}^{(t)}(\mathbb{R}) \leq \nu\vee\nu^{[+2^{-s}]}(\mathbb{R}) - \nu(\mathbb{R})$, for all $s,t\in\mathbb{N}$, $t<s$. Combining this with \eqref{eq:prop4-goal-16}, \eqref{eq:prop4-goal-18} and \eqref{eq:Q-Q-Q-Q} yields 
\begin{equation}
      0 \leq Q_k^{\inf}(\widehat{\mu}_s^{(t)}, \widehat{\nu}_s^{(t)}) - Q_k^{\inf}(\widehat{\mu}^{(t)}, \widehat{\nu}^{(t)})  \leq \nu\vee\nu^{[+2^{-s}]}(\mathbb{R}) - \nu(\mathbb{R}). 
    \label{eq:st-graal-uniform}
\end{equation}
Because $f_{\nu}$ is piecewise Lipschitz continuous, $ \nu\vee\nu^{[+2^{-s}]}(\mathbb{R}) - \nu(\mathbb{R})  \to 0$ as $s\to\infty$. 
Bounds in \eqref{eq:st-graal-uniform} thus imply
\begin{equation}
Q_k^{\inf}\left(\widehat{\mu}_s^{(t)}, \widehat{\nu}_s^{(t)}\right) \longrightarrow Q_k^{\inf}\left(\widehat{\mu}^{(t)},\widehat{\nu}^{(t)}\right) \quad \text{as }s\to\infty, \text{uniformly for all $t\in\mathbb{N}$}; 
    \label{eq:prop4-goal2}
\end{equation}
and the uniform convergence is seen from the fact that the upper bound in \eqref{eq:st-graal-uniform} does not depend on $t$. 
Next, we remark that $\widehat{\mu}^{(t)}\leq \mu$ and $\widehat{\nu}^{(t)}\geq \nu$ for every $t\in\mathbb{N}$; this follows from the infimum and supremum in \eqref{eq:discretization-mu} and \eqref{eq:discretization-nu}, and can be seen directly from the density functions: $f_{\widehat{\mu}^{(t)}}\leq f_{\mu}$ and $f_{\widehat{\nu}^{(t)}}\geq f_{\nu}$. Because $f_{\mu}$ and $f_{\nu}$ are piecewise Lipschitz continuous, $\widehat{\mu}^{(t)}(\mathbb{R}) \to \mu(\mathbb{R})$ and $\widehat{\nu}^{(t)}(\mathbb{R}) \to \nu(\mathbb{R})$ as $t\to \infty$. Thus, by Theorem~\ref{th:tv}, 
\begin{equation}
      Q_{k}^{\inf}(\widehat{\mu}^{(t)}, \widehat{\nu}^{(t)}) \longrightarrow Q_k^{\inf}(\mu,\nu) \quad \text{as }t\to\infty. 
      \label{eq:dist-hat-blank}
\end{equation}
 We conclude
\begin{equation*}
Q_k^{\inf}\left(\overline{\mu}_s,\overline{\nu}_s\right) =  Q_k^{\inf}\left(\widehat{\mu}^{(s)}_s,\widehat{\nu}^{(s)}_s\right)  \longrightarrow Q_k^{\inf}(\mu,\nu), \quad \text{as }s\longrightarrow\infty,   
\end{equation*}
where the convergence follows by combining \eqref{eq:prop4-goal2} and \eqref{eq:dist-hat-blank}. 
\hfill $\square$

\subsection{Proof of Theorem~\ref{th:continuous}}

Varying the discretization level $s$ induces changes in both the support and the mass captured at each point of the support. As in the proof of Proposition~\ref{prop:convergence}, we want to account for each of these two sources separately: for every $s,t\in\mathbb{N}$, $t\leq s$, define $\widehat{\mu}_s^{(t)}$ and $\widehat{\nu}_s^{(t)}$ as in \eqref{eq:hat-mu}--\eqref{eq:hat-nu}. Note that $\overline{\mu}_s = \widehat{\mu}^{(s)}_s$ and $\overline{\nu}_s = \widehat{\nu}^{(s)}_{s}$ for every $s\in\mathbb{N}$;  and that $\widehat{\mu}_s^{(t)}\to \widehat{\mu}^{(t)}$ and $\widehat{\nu}_s^{(t)}\to \widehat{\nu}^{(t)}$ weakly as $s\to\infty$, where $\widehat{\mu}^{(t)}$ and $\widehat{\nu}^{(t)}$ are defined in \eqref{eq:hat-mu-dens}--\eqref{eq:hat-nu-dens}. For notational convenience, let us write $\overline{\gamma}_{s} := \gamma_{k,\overline{\mu}_{s}, \overline{\nu}_{s}}^{\rm A}$ and   $\widehat{\gamma}^{(t)}_{s} := \gamma_{k,\widehat{\mu}_{s}^{(t)}, \widehat{\nu}_{s}^{(t)}}^{\rm A}$ for any $s,t\in\mathbb{N}$, $t\leq s$, and $\widehat{\gamma}^{(t)} := \gamma_{k,\widehat{\mu}^{(t)}, \widehat{\nu}^{(t)}}^{\rm A}$ for any $t\in\mathbb{N}$.  We aim to prove that $\overline{\gamma}_s \to \gamma$ as $s\to\infty$ for some $\gamma\in\mathcal{H}(\mu,\nu)$. 

Consider the Prokhorov distance $d_{\mathrm{P}}$, defined by 
\begin{equation*}
    d_{\mathrm{P}}(\pi,\pi^{\prime}) = \inf\{\varepsilon>0:\pi(A) \leq \pi^{\prime}(A^{\varepsilon}) + \varepsilon \quad \text{for every Borel set }A\subseteq\mathbb{R}^2\},\quad \pi,\pi^{\prime}\in\Pi 
\end{equation*}
where $A^{\varepsilon}  = \{z\in\mathbb{R}^2 : \text{there exists }z^*\in A \text{ with }d(z,z^*)<\varepsilon\}$. 

For any $s,t\in\mathbb{N}$, $t\leq s$, define 
\begin{equation*}
    \breve{\gamma}^{(t)}_s = \sum_{(i,j) \in\mathrm{supp}(\widehat{\gamma}^{(t)}_s)} \widehat{\gamma}^{(t)}_s(\{(i,j)\})\left(\frac{1}{2} \delta_{(i-2^{-(s+2)}, j - 2^{-(s+2)})} + \frac{1}{2} \delta_{(i+2^{-(s+2)}, j + 2^{-(s+2)})} \right).
\end{equation*}
The purpose of $\breve{\gamma}^{(t)}_s$ is to replicate the coupling $\widehat{\gamma}_s^{(t)}$ on the finer support of $\widehat{\gamma}^{(t)}_{s+1}$. The difference between $\widehat{\gamma}_s^{(t)}$ and $\breve{\gamma}_s^{(t)}$ is therefore just due to the change of support; it is straightforward to see that 
\begin{equation}
    d_{\mathrm{P}}(\widehat{\gamma}_s^{(t)},\breve{\gamma}_s^{(t)}) \leq \sqrt{\left(2^{-(s+2)}\right)^2 + \left(2^{-(s+2)}\right)^2} \leq 2^{-(s+1)},\quad s,t\in\mathbb{N}, t\leq s. \label{eq:th5-prokhorov-2s}
\end{equation}
Meanwhile, the difference between $\breve{\gamma}^{(t)}_s$ and $\widehat{\gamma}_{s+1}^{(t)}$ comes solely from the changes in coupling, dictated by the algorithm, given the finer support.

From Algorithm~\ref{algo:A}'s commands (and its limiting behaviour, from Theorem~\ref{th:discrete} if $\widehat{\mu}^{(t)}_s$ and $\widehat{\nu}^{(t)}_s$, are not atomic with the same size), we have the following observations concerning $\widehat{\gamma}_s^{(t)}$: If a portion of mass in $\widehat{\gamma}_s^{(t)}$ is coupling location $x$ in $\widehat{\mu}_s^{(t)}$ to location $y$ in $\widehat{\gamma}_s^{(t)}$ such that $y-x\leq k$, then the algorithm, for $\widehat{\gamma}_{s+1}^{(t)}$, will couple first the mass at $x+2^{-(s+2)}$ to the available mass at the largest location satisfying the objective, and such mass is available at either $y-2^{-(s+2)}$ or $y+2^{-(s+2)}$ (since $\widehat{\gamma}^{(t)}_s$'s construction was predicated on Algorithm~\ref{algo:A}'s commands), and then couple the mass at $x-2^{-(s+2)}$ also to the largest available mass (again, either at $y-2^{-(s+2)}$ or $y+2^{-(s+2)}$). Whether the algorithm will couple each actually at $y-2^{-(s+2)}$ or $y+2^{-(s+2)}$ depends on whether other portions of mass coupled to $y$ in $\widehat{\gamma}_s^{(t)}$ is at location larger or smaller than $x$, but this does not matter for the sake of this argument: what is important to note is that all the mass at $(x,y)$ remains in the $2^{-(s+1)}$ vicinity of that point in $\widehat{\gamma}_{s+1}^{(t)}$. Therefore, $d_P(\breve{\gamma}_{s}^{(t)}|_{(x\pm2^{-(s+2)},y\pm2^{-(s+2)})},\widehat{\gamma}_{s+1}^{(t)}|_{(x\pm2^{-(s+2)},y\pm2^{-(s+2)})} \leq \widehat{\gamma}_s^{(t)}(\{(x,y)\})2^{-s}$. This would also hold true when $y-x\geq k$ unless there is available mass at some other location, say $y^*<y$, such that $k+2^{-(s+2)}<y^*-x\leq k+2^{-(s+1)}$. Indeed, for $\widehat{\gamma}_{s+1}^{(t)}$, the algorithm would instead couple the mass at $x+2^{-(s+2)}$ to $y^*-2^{-(s+2)}$ since it would now satisfy the objective. The mass at $x-2^{-(s+2)}$ would now be coupled to $y+2^{-(s+2)}$, and thus continue to not be favourable toward the objective. Note that $y^*$ may have been coupled to some $x^*\leq x$ in $\widehat{\gamma}^{(t)}_s$. In that case, for $\gamma^{(t)}_{s+1}$, because $y^*-x^*> k$, mass at $x^*+2^{-(s+2)}$ is now coupled to $y-2^{-(s+2)}$, and mass at $x^*-2^{-(s+2)}$ is now coupled to $y+2^{-(s+2)}$ (since it must be the mass at the largest available location otherwise $x^*$ would not have been coupled to $y^*$ in $\widehat{\gamma}^{(t)}_s$).  As a result, half of the portion of mass originally coupled at $(x,y)$ now satisfies the objective $y-x\leq k$, and $(x^*,y^*)$ did not, and still does not, satisfy the objective. Thus, of all the mass which moves from $\breve{\gamma}_{s}^{(t)}$ to $\widehat{\gamma}_{s+1}^{(t)}$ by more than $2^{-(s+1)}$ in terms of locations, no less than a fourth of which goes in reducing the value of $\int \id_{\{y-x>k\}} \mathrm{d}\widehat{\gamma}^{(t)}_s$ (which is optimal for its discretization level, by Theorems~\ref{th:Q-inf} and \ref{th:discrete}). Hence, 
\begin{align}
\sum_{j=s}^{\infty} d_{\mathrm{P}}\left(\breve{\gamma}_{j}^{(t)}, \widehat{\gamma}_{j+1}^{(t)}\right) &\leq  \left(\sum_{j=s}^{\infty} 2^{-j}\right) + 4\left(Q_k^{\inf}\left(\widehat{\mu}_s^{(t)},\widehat{\nu}^{(t)}_s\right)- Q_k^{\inf}\left(\widehat{\mu}^{(t)}, \widehat{\nu}^{(t)}\right)\right),\notag\\
&\leq 2^{1-s} + 4\left(Q_k^{\inf}\left(\widehat{\mu}_s^{(t)},\widehat{\nu}^{(t)}_s\right)- Q_k^{\inf}\left(\widehat{\mu}^{(t)}, \widehat{\nu}^{(t)}\right)\right)\quad s,t\in\mathbb{N},\,t\leq s. \label{eq:4-diff}
\end{align}
 For every $s,t\in\mathbb{N}$, $t\leq s$, let $K_s^{(t)}$ be the kernel describing $\widehat{\gamma}_{s}^{(t)}$. The explanation above also indicates that, for $\widehat{\mu}_s^{(t)}$-almost every $x$, $K_s^{(t)}(x,\cdot)$ converges, since the kernel at $x$ will only change finitely many times from $\breve{\gamma}_s^{(t)}$ to $\widehat{\gamma}_{s+1}^{(t)}$; and the difference in kernels from $\widehat{\gamma}_{s}^{(t)}$ to $\breve{\gamma}_s^{(t)}$ is a simple translation of the order $2^{-(s+2)}$.  
For any $s,t\in\mathbb{N}$, $t\leq s$, we have
\begin{align}
    \sum_{j=s}^{\infty} d_{\mathrm{P}}\left(\widehat{\gamma}_{j}^{(t)}, \widehat{\gamma}_{j+1}^{(t)}\right) &\leq  \sum_{j=s}^{\infty} d_{\mathrm{P}}\left(\widehat{\gamma}_{j}^{(t)}, \breve{\gamma}_{j}^{(t)}\right) + \sum_{j=s}^{\infty} d_{\mathrm{P}}\left(\breve{\gamma}_{j}^{(t)}, \widehat{\gamma}_{j+1}^{(t)}\right)\notag\\
    &\leq \sum_{j=s}^{\infty} 2^{-(j+1)} + 2^{1-s} + 4\left(Q_k^{\inf}\left(\widehat{\mu}_s^{(t)},\widehat{\nu}^{(t)}_s\right) - Q_k^{\inf}\left(\widehat{\mu}^{(t)}, \widehat{\nu}^{(t)}\right) \right)\notag\\
    &\leq (3)2^{-s} + \nu\wedge\nu^{[+2^{-s}]}(\mathbb{R}) - \nu(\mathbb{R}), \label{eq:prokhorov-total}
\end{align}
where the second inequality follows from \eqref{eq:th5-prokhorov-2s} and \eqref{eq:4-diff}, while the third inequality is due to \eqref{eq:st-graal-uniform}.  As $s\to\infty$, the upper bound in \eqref{eq:prokhorov-total} vanishes, and thus, $\sum_{j=s}^{\infty} d_{\mathrm{P}}\left(\widehat{\gamma}_{j}^{(t)}, \widehat{\gamma}_{j+1}^{(t)}\right)\to 0$. Because the Prokhorov distance metrizes weak convergence for measures of same total mass, this means there exists $\{\widehat{\gamma}^{(t)}: t\in\mathbb{N}\}$ such that 
\begin{equation}
    \widehat{\gamma}_s^{(t)} \to \widehat{\gamma}^{(t)}\quad \text{weakly}, \quad \text{as }s\to\infty, \quad \text{uniformly for all }t\in\mathbb{N}, \label{eq:th5-goal-s}
\end{equation}
where the convergence holds uniformly for all $t$ because the upper bound in \eqref{eq:prokhorov-total} is independent of~$t$.

  %Because $\widehat{\gamma}_{s}^{(t)} \to \widehat{\gamma}^{(t)}$ weakly, it must be that $K_s^{(t)}(x,\cdot) \to K^{(t)}(x,\cdot)$ for $\widehat{\mu}^{(t)}$-almost every $x\in\mathbb{R}$. 
 As argued in the proof of Theorem~\ref{th:discrete}, the algorithms commands make it so that, for every $x\in\mathrm{supp}(\mu_{s}^{(t)})$, the measure $K_s^{(t)}(x,\cdot)\in\mathcal{M}$ is of the form
\begin{equation*}
    K_s^{(t)}(x,\cdot ) =  \theta_s^{(t)}(x,\cdot) + \tau_s^{(t)}(x,\cdot)
\end{equation*}
for $\theta_s^{(t)}(x,\cdot)$, $\tau_s^{(t)}(x,\cdot)\in\mathcal{M}$ satisfying \eqref{eq:th4-constr-theta}--\eqref{eq:th4-constr-tau} and the conditions thereafter. For every $t\in\mathbb{N}$, let $K^{(t)}$ be the kernel describing $\widehat{\gamma}^{(t)}$.
Because, from the explanation above, for each $t\in\mathbb{N}$, $K_s^{(t)}(x,\cdot)\to K^{(t)}(x,\cdot)$ weakly, and $\theta_s^{(t)}$'s are all required to have the largest possible total mass, there must exist $\theta^{(t)}(x,\cdot),\tau^{(t)}(x,\cdot)\in\mathcal{M}$ such that
\begin{align}
   K^{(t)}(x,\cdot) &= \theta^{(t)}(x,\cdot) + \tau^{(t)}(x,\cdot), \notag\\
    \theta_s^{(t)}(x,\cdot) &\longrightarrow \theta^{(t)}(x,\cdot) \quad \text{and} \quad \tau_s^{(t)}(x,\cdot) \longrightarrow \tau^{(t)}(x,\cdot)\quad \text{weakly}\quad \text{as }s\to\infty, \label{eq:conv-theta-tau}
\end{align}
for $\mu$-almost every $x\in\mathbb{R}$.
For every $s,t\in\mathbb{N}$, $t<s$, we have $\theta_s^{(t+1)}(x,\cdot)\leq_{\rm st}\theta_s^{(t)}(x,\cdot)$ and $\tau_s^{(t+1)}(x,\cdot)\leq_{\rm st}^*\tau_s^{(t)}(x,\cdot)$ as per \eqref{eq:stodom-theta} and \eqref{eq:stodom-tau}. By the closure of $\leq_{\rm st}$ on convergence in distribution,  
this combined with \eqref{eq:conv-theta-tau} implies that 
\begin{equation}
    \theta^{(t+1)}(x,\cdot) \leq_{\rm st} \theta^{(t)}(x,\cdot) \quad \text{and} \quad \tau^{(t+1)}(x,\cdot) \leq_{\rm st}^* \tau^{(t)}(x,\cdot),\label{eq:sto-sto-theta-tau}
\end{equation}
for $\mu$-almost every $x\in\mathbb{R}$. Let $b_x = \lim_{t\to\infty}\theta^{(t)}(x,\mathbb{R})$, which is well defined (as per \eqref{eq:th4-total-theta}). Because the sequence $\{\theta^{(t)}(x,\cdot): t\in\mathbb{N}\}$ is uniformly lower bounded in terms of $\leq_{\rm st}$ by $b_x\delta_x$ given the monotone-response constraint, by \eqref{eq:sto-sto-theta-tau} it must converge weakly to some $\theta(x,\cdot)\in\mathcal{M}$. Similarly, the sequence 
$\{\tau^{(t)}(x,\cdot): t\in\mathbb{N}\}$ is uniformly lower bounded, in terms of $\leq_{\rm st}$ by $\delta_x$, so by \eqref{eq:sto-sto-theta-tau} it must converge weakly to some $\tau(x, \cdot)$. Let $K(x,\cdot) = \theta(x,\cdot) + \tau(x,\cdot)$ for every $x\in\mathrm{supp}(\mu)$. 
This means that, for any bounded continuous function $\varphi:\mathbb{R}^2\to \mathbb{R}$, we have
\begin{equation}
   \int \varphi(x,y) K^{(t)}(x,\mathrm{d}y) \longrightarrow   \int \varphi(x,y) K(x,\mathrm{d}y), \quad \text{as }t\to\infty, \text{for $\mu$-a.e.~} x\in\mathbb{R}.\label{eq:th5-conv-int}  
\end{equation}
The specific discretization construction \eqref{eq:discretization-mu} yields that $\overline{\mu}_1 \leq \overline{\mu}_{2}\leq \dots \leq \mu$; thus $\overline{\mu}_s$ is absolutely continuous with respect to $\mu$ for every $s\in\mathbb{N}$ and its Radon-Nikodym derivative $g_s:=\mathrm{d}\overline{\mu}_s/\mathrm{d}\overline{\mu}$ is well-defined. We have $g_s\uparrow 1$. Then, for every bounded continuous function $\varphi:\mathbb{R}^2\to\mathbb{R}$, 
\begin{align}
&\int \varphi(x,y)\mathrm{d}\widehat{\gamma}^{(t)} = \int \left( \int \varphi(x,y)K_s(x,\mathrm{d}y)\right) \mathrm{d}\overline{\mu}_s \notag\\
&\quad= \int \left( \int \varphi(x,y)K_r(x,\mathrm{d}y)\right) g_s \mathrm{d}\mu
\longrightarrow  \int \left( \int \varphi(x,y)K(x,\mathrm{d}y)\right)  \mathrm{d}\mu \quad \text{as }r\to\infty,
\label{eq:weak-conv-GAMMA-th5}
\end{align}
by dominated convergence, given \eqref{eq:th5-conv-int} and since $g_s\uparrow 1$.
Let $\gamma$ be characterized by the kernel $K$, and \eqref{eq:weak-conv-GAMMA-th5} then implies
\begin{equation}
    \widehat{\gamma}^{(t)} \longrightarrow \gamma \quad \text{weakly}, \quad \text{as }s\longrightarrow\infty. 
    \label{eq:th5-goal-t}
\end{equation}
Since $\mathrm{supp}(\overline{\gamma}_s) \subseteq \mathbb{H}$ for all $s\in\mathbb{N}$ and $\mathbb{H}$ is closed, we have $\mathrm{supp}(\gamma)\subseteq\mathbb{H}$. Also, $P_1(\overline{\gamma}_s) = \overline{\mu}_s \to \mu$ and $P_2(\overline{\gamma}_s) \leq \overline{\nu}_{s} \to \nu$ weakly as $s \to \infty$, and therefore $\gamma\in\mathcal{H}(\mu,\nu)$.
We conclude 
\begin{equation*}
    \overline{\gamma}_s = \widehat{\gamma}^{(s)}_s \longrightarrow \gamma \quad \text{weakly}, \quad \text{as }s\longrightarrow\infty,
\end{equation*}
by combining \eqref{eq:th5-goal-s} and \eqref{eq:th5-goal-t}.\hfill$\square$

\subsection{Proof of Theorem~\ref{th:Q-inf-eta}}

Define $\mathcal{Z}_\ell = \{z_{\ell,1} \leq \cdots \leq z_{\ell, s_l} \}$ as the collection of the locations of atoms in $\eta_{\ell}$, for every $\ell\in\{1,\dots,h\}$. 
Let us reuse the same set-up and notation as in the proof of Theorem~\ref{th:Q-inf}, but let each element of coupling multisets include atoms in $\mathcal{Z}_1,\ldots, \mathcal{Z}_h$. In other words, for a coupling multiset $\Gamma$ of $\mathcal{X}$, $\mathcal{Z}_1,\ldots, \mathcal{Z}_h$ and $\mathcal{Y}$, its elements take the form $(x,z_1,\ldots, z_h,y)$. 

The argument letting the discussion on optimality reduce to atomic couplings is very similar to the one in the proof of Theorem~\ref{th:Q-inf}. For brevity, we will not expose it to its entirety: adapting the argument to accomodate the further constraint given by the partial treatement simply requires, both for items \ref{item:A-eta} and \eqref{item:B-eta}, to fully commit the largest sequence of partial treatment locations $z_1 \leq \dots \leq z_\ell$ partially committed to $(x_i,y^*)$ by the kernel $K^{[i-1]}_0$ when constructing $K^{[i]}_0$, and proportionnally reattribute the other partially committed partial treatment locations in the same fashion as in \eqref{eq:reattrib-inf}--\eqref{eq:reattrib-sup}. The same arguments in the validity of the reattribution in \eqref{eq:reattrib-inf}--\eqref{eq:reattrib-sup} hold for the reattribution of the partial treatment locations. 

The remaining part of the proof for each of items~\ref{item:A-eta} and~\ref{item:B-eta} is predicated on the same idea as in the one of Theorem~\ref{th:Q-inf}: proving inductively on $i$ that the algorithm yields a valid coupling for $\mu\vert_{\mathcal{X}_{[i]}}$, $i\in\{1,\ldots,m\}$. However, we have to account for the additional constraint given by $\eta_1,\ldots,\eta_h$. 

For item~\ref{item:A-eta}: For the base step: 
Letting $i=m$, we have $\mathcal{X}_{[i]} = \{x_m\}$; existence and optimality are straightforward to see. For the inductive step: Letting $i\in\{1,\ldots,m-1\}$, we assume that the coupling multiset for $\mathcal{X}_{[i+1]}, \mathcal{Z}_1,\ldots, \mathcal{Z}_{h},\mathcal{Y}$ produced by the algorithm, denoted by $\Gamma_{i+1}$, is valid and optimal. We need to prove that it is the case for $\Gamma_i$. 
For validity: Because the atoms in $\mu\vert_{\mathcal{X}_{[i]}}$ are processed from right to left, adding an atom leftmost does not alter the tuple of other atoms, processed before. 
Therefore, $\Gamma_i = \Gamma_{i+1}\cup \{(x_i, z_{1,c_1},\ldots,z_{h, c_h}, y_d)\}$ where $c_1,\ldots, c_h$ and $d$ are selected by the algorithm. 
Note that because $\mu((x,\infty)) \leq \eta_1((x,\infty))$, for all $x\in\mathbb{R}$, there is necessarily an uncoupled $z_{1,c_1}\geq x_i$. Next, $\eta_1((z_{1,c_1},\infty)) \leq \eta_2((z_{1,c_1},\infty))$ so there is at least one more uncoupled atom in $\eta_2$; it must be uncoupled because, otherwise, this would mean that one of the atoms in $\eta_2\vert_{(z_{1,c_1},\infty)}$ is coupled to atom left of $z_{1,c_1}$, contradicting that the algorithm selected the rightmost admissible atom in $\eta_1$ at that step. Repeat the preceding argument for $\eta_2$ and $\eta_3$, $\eta_3$ and $\eta_4$, up to $\eta_h$ and $\nu$. 

For optimality: Let $\varphi_{\Gamma} = \sum_{(x,z_1,\ldots,z_h, y) \in \Gamma} a \id_{\{y-x >k\}}$ for any coupling multiset $\Gamma$, where we recall $a$ is the mass of an atom; we know that $\varphi_{\Gamma_{i+1}}$ is smallest for coupling multisets of $\mathcal{X}_{[i+1]}$, $\mathcal{Z}_1,\ldots,\mathcal{Z}_h$ and $\mathcal{Y}$. We need to prove that $\varphi_{\Gamma_{i}}$ is smallest for coupling multisets of $\mathcal{X}_{[i]}$ with the others. Note that any such coupling multiset $\Gamma$ contains one more tuple than $\Gamma_{i+1}$, so $\varphi_{\Gamma}\geq \varphi_{\Gamma_{i+1}}$ by the latter's optimality. For $\Gamma_i$, two situations may occur. First possibility: $y_d - x_i \leq k$. Then $\varphi_{\Gamma_i} = \varphi_{\Gamma_{i+1}}$, and it is thus smallest. Second possibility: $d$ is such that $y_d - x_i > k $. Then, $\varphi_{\Gamma_i} = \varphi_{\Gamma_{i+1}} + a$.  We need to argue that no coupling multiset $\Gamma$ of $\mathcal{X}_{[i]}$ with the others could produce $\varphi_{\Gamma} = \varphi_{\Gamma_{i+1}}$. Let us suppose that such a coupling multiset exists, denoted $\Gamma^{\dagger}$. We will show that this is absurd by revealing a contradiction. 

There are $c_1^{\dagger},\ldots,c_h^{\dagger}, d^{\dagger}$ such that $(x_i,z_{1,c_1^{\dagger}}, \ldots,z_{h,c_h^{\dagger}}, y_{d^{\dagger}})\in \Gamma^{\dagger}$. It must hold that $y_{d^{\dagger}} - x_i \leq k$; otherwise, we would have $\varphi_{\Gamma^{\dagger}\backslash \{(x_i,z_{1,c_1^{\dagger}}, \ldots,z_{h,c_h^{\dagger}}, y_{d^{\dagger}})\}} = \varphi_{\Gamma_{i+1}} - a$, contradicting the optimality of $\varphi_{\Gamma_{i+1}}$. Therefore, at least one of $z_{1,c_1^{\dagger}}, \ldots, z_{h,c_h^{\dagger}}, y_{d^{\dagger}}$ must be coupled in $\Gamma_{i+1}$, or else the algorithm would have selected that tuple for $\Gamma_i$.

Suppose that any of $z_{1,c_1^{\dagger}}, \ldots, z_{h,c_{h}^{\dagger}}$, say $z_{\ell,c_{\ell}^{\dagger}}$, was already coupled in $\Gamma_{i+1}$, then there is $x^*\geq x_i$, and $z_1^*,\ldots,z_h^*, y^*$ such that $(x^*, \ldots, z_{\ell,c_{\ell}^{\dagger}}, \ldots, y^*)\in \Gamma_{i+1}$. It must be that $y^*-x^*\leq k$ since otherwise the algorithm would have coupled $y_{d^{\dagger}}$ with $x^*$ in $\Gamma_{i+1}$. As $z_{\ell,c_{\ell}^{\dagger}}$ got stolen by another tuple in $\Gamma^{\dagger}$, one possibility is for $x^*$ to couple with another $y^{**}\in\mathcal{Y}$ and some partial treatment atom locations such that $y^{**}-x^* \leq k$, but then one of them must already be coupled in $\Gamma_{i+1}$ from the same explanation as above. 
%From this reasoning, it is impossible to create an augmenting path with respect to satisfying the objective $y-x\leq k$: no uncoupled element of $\Gamma_{i+1}$ can close the path. 
Another possibility is for $x^*$ and $y^*$ to find an admissible replacement $z^*_{\ell}\in\mathcal{Z}_{\ell}$, $z^*\geq z_{\ell-1}$, for $z_{\ell,c_{\ell}^{\dagger}}$. It must be that $z^*_{\ell}\leq z_{\ell,c_{\ell}^{\dagger}}$ because the algorithm picked the largest admissible element of $\mathcal{Z}_{\ell}$ for $\Gamma_{i+1}$. However, $z^*_{\ell}\leq z_{\ell,c_{\ell}^{\dagger}}$ means that $(x_i,\ldots,z_{\ell}^*, \ldots, y_{d^{\dagger}})$ would be a valid tuple and the algorithm would have chosen it for $\Gamma_i$; this is a contradiction (as $y_{d^{\dagger}}\neq y_d$).

If all of $z_{1,c_1^{\dagger}}, \ldots, z_{h,c_{h}^{\dagger}}$, say $z_{\ell,c_{\ell}^{\dagger}}$ were not coupled in $\Gamma_{i+1}$ but $y_{d^{\dagger}}$ was, let $x^*, z_{1}^*, \ldots, z_h^*$ be such that $(x^*, z_1^*,\ldots,z_h^*, y_{d^{\dagger}})\in \Gamma_{i+1}$. Because $x^*\geq x_i$, it must hold that $y_{d^{\dagger}} - x^*\leq k$ as well. Hence, for $\Gamma^{\dagger}$,  the atom at $x^*$ must find some $y^*\in\mathcal{Y}$ and partial-treatment atom locations to be coupled to such that $y^* - x^* \leq k$, for $\varphi_{\Gamma^{\dagger}} = \varphi_{\Gamma_{i+1}}$ to hold. By the same explanation as above, at least one of $y^*$ or these partial-treatment atom locations must be already coupled. %Hence, no possibility of creating an augmenting path with regards to satisfying the objective $y-x\leq k$. 
In other words, $\Gamma^{\dagger}$ does not allow for another element satisfying $y-x\leq k$, but it would require one more so that $\varphi_{\Gamma^{\dagger}} = \varphi_{\Gamma_{i+1}}$; hence, a contradiction.

For item \ref{item:B-eta}: 
 For the base step: Let $i=m$; thus $\mathcal{X}_{[i]} = \{x_m\}$, and existence and optimality are trivial. For the inductive step: For any $i\in\{1, \ldots, m-1\}$, we assume that $\Gamma_{i+1}$, the coupling multiset of $\mathcal{X}_{[i+1]}$, $\mathcal{Z}_1,\ldots,\mathcal{Z}_h$ and $\mathcal{Y}$ produced by the algorithm, is valid and optimal. We need to prove it for $\Gamma_i$, the coupling multiset of $\mathcal{X}_{[i]}$ and the others produced by the algorithm. The argument for validity is the same as for item~\ref{item:A-eta} and we omit it. For optimality: For any coupling multiset $\Gamma$,  define $\psi_{\Gamma} = \sum_{(x,z_1,\ldots,z_h, y)\in \Gamma} a \id_{\{y-x\leq k\}}$ where $a$ is the mass of an atom. The maximal value for $\sum_{(x,z_1,\ldots,z_h, y)\in \Gamma} a \id_{\{y-x > k\}}$ is attained if $\psi_{\Gamma}$ is smallest. Therefore, we have that $\psi_{\Gamma_{i+1}}$ is smallest for coupling multisets of $\mathcal{X}_{[i+1]}$, $\mathcal{Z}_1,\ldots,\mathcal{Z}_h$ and $\mathcal{Y}$, and we need to prove that $\psi_{\Gamma_i}$ is smallest for coupling multisets of $\mathcal{X}_{[i]}$ with the others. Because these coupling multisets require one more element than $\Gamma_{i+1}$, we have $\psi_{\Gamma} \geq \psi_{\Gamma_{i+1}}$ for any $\Gamma$ of those. Let $c_1, \ldots, c_h, d$ be the indices chosen by Algorithm~\ref{algo:B-eta} so $(x_i, z_{1,c_1}, \ldots, z_{h,c_h}, y_d)\in \Gamma_i$. There are two possibilities regarding $y_d$. Either $y_d-x_i> k$, in which case $\psi_{\Gamma_i} = \psi_{\Gamma_{i+1}}$ and $\psi_{\Gamma_i}$ is therefore smallest, or $y_d - x_i \leq k$, for which we need to further argue optimality. In that latter case, $\psi_{\Gamma_i} = \psi_{\Gamma_{i+1}} + a$. Define $\Gamma^{\dagger}$, a coupling multiset such that $\psi_{\Gamma^{\dagger}} = \psi_{\Gamma_{i+1}}$; we argue by contradiction that $\Gamma^{\dagger}$ cannot exist.  

Let $c_1^{\dagger}, \ldots, c_h^{\dagger}, d^{\dagger}$ be the indices such that $(x_i, z_{1,c_1^{\dagger}}, \ldots, z_{h,c_{h}^{\dagger}}, y_{d^{\dagger}})\in \Gamma^{\dagger}$. It must hold that $y_{d^{\dagger}} - x_i > k$; else, we would have $\psi_{\Gamma^{\dagger}\backslash \{(x_i,z_{1,c_1^{\dagger}}, \ldots,z_{h,c_h^{\dagger}}, y_{d^{\dagger}})\}} = \psi_{\Gamma_{i+1}} - a$, contradicting the optimality of $\psi_{\Gamma_{i+1}}$.
At least one of $z_{1,c_1^{\dagger}}, \ldots, z_{h,c_h^{\dagger}}, y_{d^{\dagger}}$ must be coupled in $\Gamma_{i+1}$ because otherwise the algorithm would have selected that tuple for $\Gamma_{i}$. Since $y_d - x_i \leq k$ and $y_{d^{\dagger}} - x_i > k$, we have $y_{d}< y_{d^{\dagger}}$. Therefore, the reason for which the algorithm could not select $y_{d^{\dagger}}$ must be because it was the one that was already coupled --- otherwise $z_{1,c_1}, \ldots, z_{h, c_h}$ would be admissible to be coupled with $x_i$ and $y_{d^{\dagger}}$ in $\Gamma_i$,  and the algorithm would have selected that. 

Let $x^*, z_1^*, \ldots, z_h^*$ be such that $(x^*, z_1^*, \ldots, z_h^*, y_{d^{\dagger}})\in \Gamma_{i+1}$. In $\Gamma^{\dagger}$, the atom at $x^*$ is coupled to another atom location in $\mathcal{Y}$, say $y^*$. Recall that $\Gamma_{i+1}$ was constructed by Algorithm~\ref{algo:B-eta}, and there are two possibilities regarding $x^*$ and $y_{d^{\dagger}}$: either $y_{d^{\dagger}} - x^* \leq k$ or $y_{d^{\dagger}} - x^* > k$. If $y_{d^{\dagger}} - x^* \leq k$, this means that all $y\in\mathcal{Y}$ such that $y- x^* > k$ were already coupled, otherwise the algorithm could have selected $y$ since $(x^*, z_1^*, \ldots, z_h^*, y)$ would be valid, as $y>y_{d^{\dagger}}$. 
Also, it means that all $y\in\mathcal{Y}$ such that $x^*\leq y\leq y_{d^{\dagger}}$ were already coupled or could not satisfy the partial treatment constraint. Hence, this leaves that $y^*$ must be coupled in $\Gamma_{i+1}$, or that $y^*$ is uncoupled and $y^*\in[y_{d^{\dagger}}, k+x^*]\cap \mathcal{Y}$. This latter option contradicts that there were no uncoupled admissible $y_d$ such that $y_d - x_i > k$; indeed $y^* - x_i > k$, and $z_{1}^*,\ldots z_h^*$ could have committed to that coupling. So $y^*$ must be coupled in $\Gamma_{i+1}$ in that case. 
If $y_{d^{\dagger}} - x^* > k$, then it must be that $y^* - x^*> k$ for $\psi_{\Gamma^{\dagger}} = \psi_{\Gamma_{i+1}}$ to hold. This means, from the same explanation as above, that $y^*$ was already coupled to another location in $\mathcal{X}$. Combining the arguments for both the cases $y_{d^{\dagger}} - x^* \leq k$ and $y_{d^{\dagger}}- x^* > k$, we have a contradiction to the existence of $\Gamma^{\dagger}$, as no elements of $\mathcal{Y}$ uncoupled in $\Gamma_{i+1}$ may be coupled in $\Gamma^{\dagger}$, but it would require one more element. \hfill $\square$

\section{Directional optimal transport}
\label{sect:dl}

Directional optimal transport is a notion closely related to the set $\mathcal{H}(\mu,\nu)$ in the case where $\mu$ and $\nu$ are probability measures; it allows a better understanding of the feasible dependence schemes on that set. For a measure $\pi$ on $\mathbb{R}^2$, its joint cdf is given by $F_\pi(x,y)= \pi((-\infty,x], (-\infty,y])$. 

\citet[Theorem~6]{AMZ20} state that the joint cdf of any element $\pi$ of $\mathcal{H}(\mu,\nu)$, where $\mu,\nu\in\mathcal{M}$ are probability measures, is bounded as such:
\begin{align}
	L_{\mu,\nu}(x,y) \leq F_{\pi}(x,y)&\leq \min(F_\mu(x), F_{\nu}(y)), 
	\label{eq:arnold}\\
	\text{where}\quad L_{\mu,\nu}(x,y) &= \begin{cases}
		F_{\nu}(y), & x\geq y,\\
		F_{\mu}(x) - \inf_{x\leq s \leq y}\{F_{\mu}(s) - F_{\nu}(s)\} &x<y, 
	\end{cases} \quad \text{for all } x,y\in\mathbb{R}.\notag
\end{align}
Thus, the dependence schemes in $\mathcal{H}(\mu,\nu)$ are bounded with respect to the concordance order. (On the concordance order, see for example \citet[Chapter 9]{SS07}.) 
The upper-bound in \eqref{eq:arnold} corresponds to comonotonicity.   
Both bounds are attained.

\begin{figure}[tb]
		\centering
		\begin{tikzpicture}
			\draw[Black!40,  ->, thick] (-0.5,0) -- (7.5, 0);
			  \node at (-1.25, 0.75) [rectangle] {$X\lawis \mu$};
          \node at (-1.25, 3.25) [rectangle] {$Y \lawis \nu$};
			\foreach \x in {0,1,...,7}
			\draw[Black!40,  thick] (\x cm,-2pt) node[below]{\x} -- (\x cm,2pt);
			
			\filldraw (3.5,3.25) circle (0.1 cm);
			\filldraw (4,3.25) circle (0.1 cm);
			\filldraw (5.5,3.25) circle (0.1 cm);
			\filldraw (6,3.25) circle (0.1 cm);
			\filldraw (7,3.25) circle (0.1 cm);
			\filldraw (3,3.25) circle (0.1 cm);
			
			\draw[CornflowerBlue, very thick, ->] (4.5, 0.75) -- (5.5, 3.25);
			\draw[CornflowerBlue, very thick, ->] (4, 0.75) -- (4, 3.25);
			\draw[CornflowerBlue, very thick, ->] (2, 0.75) -- (3.5, 3.25);
			\draw[CornflowerBlue, very thick, ->] (0.5, 0.75) -- (6, 3.25);
			\draw[CornflowerBlue, very thick, ->] (0,0.75) -- (7,3.25); 
			\draw[CornflowerBlue, very thick, ->] (3,0.75) -- (3,3.25); 
			
			\filldraw (3,0.75) circle (0.1 cm);
			\filldraw (0,0.75) circle (0.1 cm);
			\filldraw (0.5,0.75) circle (0.1 cm);
			\filldraw (2,0.75) circle (0.1 cm);
			\filldraw (4,0.75) circle (0.1 cm);
			\filldraw (4.5,0.75) circle (0.1 cm);
		\end{tikzpicture}	
	\caption{Illustration of DL-coupling for two discrete distributions with atoms of the same size.}
	\label{fig:dl}
\end{figure}
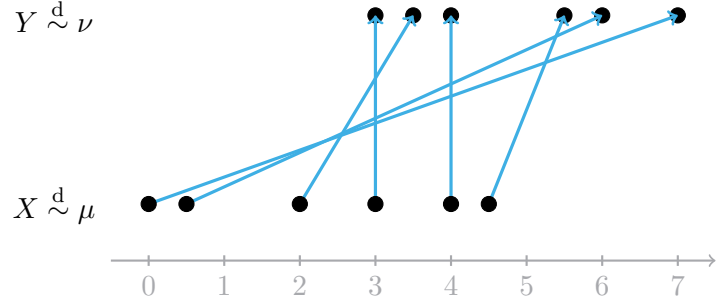

In \cite{NW22}, the authors study $\mathcal{H}(\mu,\nu)$, and particularly $L_{\mu,\nu}$, from the perspective of optimal transport. They introduce directional optimal transport to describe the class of transports from $\mu$ to $\nu$ yielding elements of $\mathcal{H}(\mu,\nu)$. They call the specific coupling rendering $L_{\mu,\nu}$ the DL-coupling and provide a construction for it, transcribed in the following proposition.

\begin{proposition}[Theorems 2.1--2.2(a) of \cite{NW22}, reformulated]
Consider probability measures $\mu,\nu\in\mathcal{M}$ such that $\mu\leq_{\rm st}\nu$. There exists a unique coupling $\pi^{\rm dl}_{\mu,\nu}$ supported on $\mathbb{H}$  
 defined by coupling, for every $x\in\mathbb{R}$, $\mu\vert_{(x,\infty)}$ to $\nu_x$, the latter described by its cdf as such:
\begin{equation*}
	F_{\nu_x} = \sup\{ F_\theta : \theta\in \mathfrak{S}_x\} \quad \text{where } \mathfrak{S}_x = \{\theta: \mu\vert_{(x,\infty)} \leq_{\rm st} \theta \leq \nu\}.
\end{equation*}
Then, the joint cdf of $\pi^{\mathrm{dl}}_{\mu,\nu}$ is $L_{\mu,\nu}$. Moreover, for every $x\in\mathbb{R}$, $\nu_x$ is the unique minimal element of $S_x$ with respect to the order $\leq_{\rm st}$. 
\end{proposition}

For atomic distributions, the transport may be interpreted as follows: going iteratively from largest to smallest atomic location in $\mu$, one couples each atom in $\mu$, say at location $x$, to the atom in $\nu$ with the smallest location $y$ such that $y\geq x$. Figure~\ref{fig:dl} provides an example of the DL-coupling for such distributions.

\section{Risk aggregation under DU and an order constraint}
\label{sect:CHEN}

% \citet[Section 5.3]{CLW22} examine the best- and worst-case probabilities of excess for an aggregate risk position with an order constraint: for $k\in\mathbb{R}$, 
%     \begin{align}
% 	M_k^{\inf}(\mu,\nu) &= \inf\{\mathbb{P}(X+Y > k): X\lawis \mu, Y\lawis\nu, Y\geq X  \};
% 		\label{eq:chenprob-inf}\\
% 	M_k^{\sup}(\mu,\nu) &=	\sup\{\mathbb{P}(X+Y \geq k): X\lawis \mu, Y\lawis\nu, Y \geq X \}.
% 		\label{eq:chenprob-sup}
% \end{align}

% Problems \eqref{eq:chenprob-inf}--\eqref{eq:chenprob-sup} are very similar to \eqref{eq:chenminus-inf}--\eqref{eq:chenminus-sup}; the only difference is a sign flip in the ordering constraint. This sign flip, however, makes it so that we cannot apply the solution from \cite{CLW22} to solve \eqref{eq:chenminus-inf}--\eqref{eq:chenminus-sup}.

In \cite[Section 5.3]{CLW22}, the authors solve \eqref{eq:chenminus-inf}--\eqref{eq:chenminus-sup} by obtaining the best-case and worst-case $\mathrm{VaR}_p$ with an explicit stochastic construction, for every value of $p$, and then inverting. For that matter, they rely on the fact that worst-case $\mathrm{VaR}_p(X)$ is solved by $p$-concentrated random vectors, see e.g.~\cite{LW21}.\footnote{A random vector is $p$-concentrated if its components share a common $p$-tail event; an event $A\in\mathcal{F}$ is a $p$-tail event of a random variable $Z$ if $\mathbb{P}(A) = 1-p$ and, for every $\omega\in A$ and $\omega^{\prime}\in\Omega\backslash A$, we have $Z(\omega)\geq Z(\omega^{\prime})$. See \cite{WZ21} on both of these concepts.  } 
This is made possible by the fact that $p$-concentration of $(X,Y)$ is always possible while satisfying the constraint $X\leq Y$, unlike $(-X, Y)$.

While the solution from \cite{CLW22} does not apply to \eqref{eq:monotoneeffectiveness-inf}--\eqref{eq:monotoneeffectiveness-sup}, the coupling heuristic from Algorithms~\ref{algo:A} and~\ref{algo:B}, on the other hand, allows to solve \eqref{eq:chenminus-inf}--\eqref{eq:chenminus-sup}. A small adjustment, however, is required to account for the sign flip. This adjustment is presented in Algorithms~\ref{algo:A-star} and~\ref{algo:B-star} below. Note that the algorithms also solve analogues to \eqref{eq:chenminus-inf}--\eqref{eq:chenminus-sup} for non-probability measures.

%as we explained in Section~\ref{sect:problems34}. 

\renewcommand{\thealgofloat}{$\mathrm{A}^*$}

\begin{algofloat}
\centering

   \caption{}
    \fbox{
    \begin{minipage}{0.9\textwidth}
        %\textbf{Algorithm A*}\\
        \onehalfspacing
        Execute Algorithm~\ref{algo:A}, but change the first repeated substep to 
        \begin{itemize}
                \item[$\triangleright$] Check whether $\mathcal{J}_i := \{j\in \mathcal{I}_i:  x_i\leq y_j\leq k - x_i\}$ is empty.
            \end{itemize}
        and the last step to
        \begin{itemize}
            \item[$\blacktriangleright$] Return $M^{\mathrm{A}^*}_k(\mu,\nu) := a(\sum_{i=1}^m \id_{\{y_{d_i} + x_i > k\}})$. 
        \end{itemize}
    \end{minipage}
    }
    \label{algo:A-star}
\end{algofloat}

\renewcommand{\thealgofloat}{$\mathrm{B}^*$}

\begin{algofloat}
\centering

   \caption{}
    \fbox{
    \begin{minipage}{0.9\textwidth}
        %\textbf{Algorithm B*}\\
        \onehalfspacing
        Execute Algorithm~\ref{algo:B}, but change the first repeated substep to 
        \begin{itemize}
                \item[$\triangleright$] Check whether $\mathcal{J}_i := \{j\in \mathcal{I}_i:   y_j\geq k - x_i, \, y_j \geq x_i\}$ is empty.
            \end{itemize}
        and the last step to
        \begin{itemize}
            \item[$\blacktriangleright$] Return $M^{\mathrm{B}^*}_k(\mu,\nu) := a(\sum_{i=1}^m \id_{\{y_{d_i} + x_i \geq k\}})$. 
        \end{itemize}
    \end{minipage}
    }
    \label{algo:B-star}
\end{algofloat}

%The outcomes of Algorithms~\ref{algo:A-star} and \ref{algo:B-star} are the solutions to problems in \eqref{eq:chenprob-inf} and \eqref{eq:chenprob-sup} for atomic $\mu$ and $\nu$, as we indicate in the next result. 

\begin{proposition}
    Let $\mu,\nu\in\mathcal{M}$ be atomic with the same size, and suppose that $\mu\leq_{\rm st}\nu$. Then, $M^{\inf}_k(\mu,\nu) = M^{\mathrm{A}^*}_k(\mu,\nu)$ and $M^{\sup}_k(\mu,\nu) = M^{\mathrm{B}^*}_k(\mu,\nu)$, for any $k\geq 0$. 
\end{proposition}

\begin{proof}
    The proof is very similar to that  of Theorem~\ref{th:Q-inf} and thus it is omitted. 
\end{proof}

Algorithms \ref{algo:A-star} and \ref{algo:B-star} are more efficient than the method of \cite{CLW22}, as they do not involve   numerical inversion. The results of Propositions~\ref{prop:atomization} and~\ref{prop:convergence} can straightforwardly be extended to this case to accommodate non-atomic distributions.

 We present in Figure~\ref{fig:D-chen} two illustrations of the transport plan produced by Algorithm~\ref{algo:A-star}. The coupling principle is the same as for Algorithm~\ref{algo:A}: going right to left through atoms of $\mu$, we couple each to the yet-uncoupled atom of $\nu$ at the largest location such that $x+y\leq k$, if possible, and otherwise to the yet-uncoupled atom of $\nu$ at the largest location of all. In a similar vein, 
Figure~\ref{fig:E-chen} provides two illustrations of the transport plan produced by Algorithm~\ref{algo:B-star}. The coupling principle is the same as for Algorithm~\ref{algo:B}: going right to left, each atom of $\mu$ was coupled to the yet-uncoupled atom of $\nu$ at the smallest location such that $x+y\geq k$ and $y\geq x$, and, if not possible, to the one at the smallest location such that $y\geq x$. 
Colours in both figures make explicit that indeed $\mu$ and $\nu$ are $p$-concentrated in these examples.

 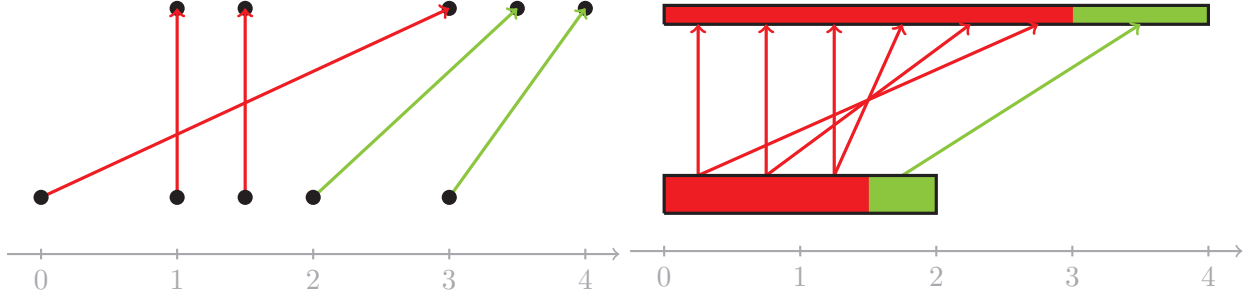
\begin{figure}[tb]
 	\begin{minipage}{0.49\textwidth}
 		\centering
 		\begin{tikzpicture}[xscale=1.8]
 			\draw[Black!40,  ->, thick] (-0.25,0) -- (4.25, 0);
 			
 			\foreach \x in {0,1,...,4}
 			\draw[Black!40,  thick] (\x cm,-2pt) node[below]{\x} -- (\x cm,2pt);
 			
 			\node  at (1,3.25) [circle, fill = Black, minimum size=0.2 cm, inner sep=0pt, draw=none] {};
 			\node  at (1.5,3.25) [circle, fill = Black, minimum size=0.2 cm, inner sep=0pt, draw=none] {};
 			\node  at (3,3.25) [circle, fill = Black, minimum size=0.2 cm, inner sep=0pt, draw=none] {};
 			\node  at (3.5,3.25) [circle, fill = Black, minimum size=0.2 cm, inner sep=0pt, draw=none] {};
 			\node  at (4,3.25) [circle, fill = Black, minimum size=0.2 cm, inner sep=0pt, draw=none] {};
 			
 			\draw[LimeGreen, very thick, ->] (3, 0.75) -- (4, 3.25);
 			\draw[LimeGreen, very thick, ->] (2, 0.75) -- (3.5, 3.25);
 			\draw[Red, very thick, ->] (1.5, 0.75) -- (1.5, 3.25);
 			\draw[Red, very thick, ->] (1, 0.75) -- (1, 3.25);
 			\draw[Red, very thick, ->] (0,0.75) -- (3,3.25); 
 			
 			\node  at (0,0.75) [circle, fill = Black, minimum size=0.2 cm, inner sep=0pt, draw=none] {};
 			\node  at (1,0.75) [circle, fill = Black, minimum size=0.2 cm, inner sep=0pt, draw=none] {};
 			\node  at (1.5,0.75) [circle, fill = Black, minimum size=0.2 cm, inner sep=0pt, draw=none] {};
 			\node  at (2,0.75) [circle, fill = Black, minimum size=0.2 cm, inner sep=0pt, draw=none] {};
 			\node  at (3,0.75) [circle, fill = Black, minimum size=0.2 cm, inner sep=0pt, draw=none] {};
 		\end{tikzpicture}	
 	\end{minipage}
 	\begin{minipage}{0.49\textwidth}
 		\centering
 		\begin{tikzpicture}[xscale=1.8]
 			\draw[Black!40,  ->, thick] (-0.25,0) -- (4.25, 0);
 			
 			\foreach \x in {0,1,...,4}
 			\draw[Black!40, thick] (\x cm,-2pt) node[below]{\x} -- (\x cm,2pt);
 			
 			\filldraw[LimeGreen] (3,3) -- (3,3.25) -- (4, 3.25) -- (4,3) ;
 			\filldraw[Red] (0,3) -- (0,3.25) -- (3, 3.25) -- (3,3) ;
 			\draw[Black, very thick] (0,3) -- (0,3.25) -- (4, 3.25) -- (4,3) -- (0,3) ;
 			
 			\draw[LimeGreen, very thick, ->] (1.75,1) -- (3.5,3); 
 			\draw[Red, very thick, ->] (1.25, 1) -- (1.25, 3);
 			\draw[Red, very thick, ->] (1.25, 1) -- (1.75, 3);
 			\draw[Red, very thick, ->] (0.75, 1) -- (0.75, 3);
 			\draw[Red, very thick, ->] (0.75, 1) -- (2.25, 3);
 			\draw[Red, very thick, ->] (0.25, 1) -- (0.25, 3);
 			\draw[Red, very thick, ->] (0.25, 1) -- (2.75, 3);

 			\filldraw[LimeGreen] (1.5,0.5) -- (1.5,1) -- (2, 1) -- (2,0.5) ;
 			\filldraw[Red] (0,0.5) -- (0,1) -- (1.5, 1) -- (1.5,0.5) ;
 			\draw[Black, very thick] (0,0.5) -- (0,1) -- (2, 1) -- (2,0.5) -- (0,0.5) ;
 		\end{tikzpicture}	
 	\end{minipage}
 	\caption{Illustration of the transport plan produced by Algorithm \ref{algo:A-star}, for $k=3$: (Left.) for discrete distributions with equally weighted masses; and (Right.) for $\mu = $ Uniform$(0,2)$ and $\nu = $ Uniform$(0,4)$. In green are the transport portions such that $x+y > 3$; in red are the ones such that $x+y \leq 3$.}
 	\label{fig:D-chen}
 \end{figure}

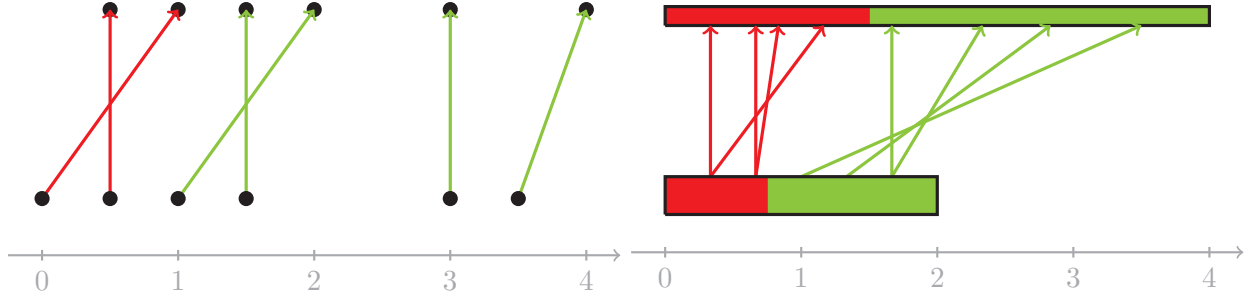
\begin{figure}[tb]
	\begin{minipage}{0.49\textwidth}
		\centering
		\begin{tikzpicture}[xscale=1.8]
			\draw[Black!40,  ->, thick] (-0.25,0) -- (4.25, 0);
			
			\foreach \x in {0,1,...,4}
			\draw[Black!40,  thick] (\x cm,-2pt) node[below]{\x} -- (\x cm,2pt);

			\node  at (0.5,3.25) [circle, fill = Black, minimum size=0.2 cm, inner sep=0pt, draw=none] {};
			\node  at (1,3.25) [circle, fill = Black, minimum size=0.2 cm, inner sep=0pt, draw=none] {};
			\node  at (1.5,3.25) [circle, fill = Black, minimum size=0.2 cm, inner sep=0pt, draw=none] {};
			\node  at (2,3.25) [circle, fill = Black, minimum size=0.2 cm, inner sep=0pt, draw=none] {};
			\node  at (3,3.25) [circle, fill = Black, minimum size=0.2 cm, inner sep=0pt, draw=none] {};
			\node  at (4,3.25) [circle, fill = Black, minimum size=0.2 cm, inner sep=0pt, draw=none] {};
			
			\draw[LimeGreen, very thick, ->] (3.5, 0.75) -- (4, 3.25);
			\draw[LimeGreen, very thick, ->] (3, 0.75) -- (3, 3.25);
			\draw[LimeGreen, very thick, ->] (1.5, 0.75) -- (1.5, 3.25);
			\draw[LimeGreen, very thick, ->] (1, 0.75) -- (2, 3.25);
			\draw[Red, very thick, ->] (0.5,0.75) -- (0.5,3.25); 
			\draw[Red, very thick, ->] (0,0.75) -- (1,3.25); 
			
			\node  at (0,0.75) [circle, fill = Black, minimum size=0.2 cm, inner sep=0pt, draw=none] {};
			\node  at (0.5,0.75) [circle, fill = Black, minimum size=0.2 cm, inner sep=0pt, draw=none] {};
			\node  at (1,0.75) [circle, fill = Black, minimum size=0.2 cm, inner sep=0pt, draw=none] {};
			\node  at (1.5,0.75) [circle, fill = Black, minimum size=0.2 cm, inner sep=0pt, draw=none] {};
			\node  at (3,0.75) [circle, fill = Black, minimum size=0.2 cm, inner sep=0pt, draw=none] {};
			\node  at (3.5,0.75) [circle, fill = Black, minimum size=0.2 cm, inner sep=0pt, draw=none] {};
		\end{tikzpicture}	
	\end{minipage}
	\begin{minipage}{0.49\textwidth}
		\centering
		\begin{tikzpicture}[xscale=1.8]
			\draw[Black!40,  ->, thick] (-0.25,0) -- (4.25, 0);
			
			\foreach \x in {0,1,...,4}
			\draw[Black!40, thick] (\x cm,-2pt) node[below]{\x} -- (\x cm,2pt);
			
			\filldraw[LimeGreen] (1.5,3) -- (1.5,3.25) -- (4, 3.25) -- (4,3) ;
			\filldraw[Red] (0,3) -- (0,3.25) -- (1.5, 3.25) -- (1.5,3) ;
			\draw[Black, very thick] (0,3) -- (0,3.25) -- (4, 3.25) -- (4,3) -- (0,3) ;
			
			\draw[LimeGreen, very thick, ->] (1.666, 1) -- (1.666, 3);
			\draw[LimeGreen, very thick, ->] (1.666, 1) -- (2.333, 3);
			\draw[LimeGreen, very thick, ->] (1.333, 1) -- (2.833, 3);
			\draw[LimeGreen, very thick, ->] (1, 1) -- (3.5, 3);
			\draw[Red, very thick, ->] (0.666, 1) -- (0.666, 3);
			\draw[Red, very thick, ->] (0.666, 1) -- (0.8333, 3);
			\draw[Red, very thick, ->] (0.333, 1) -- (1.166, 3);
			\draw[Red, very thick, ->] (0.333, 1) -- (0.333, 3);

			\filldraw[LimeGreen] (0.75,0.5) -- (0.75,1) -- (2, 1) -- (2,0.5) ;
			\filldraw[Red] (0,0.5) -- (0,1) -- (0.75, 1) -- (0.75,0.5) ;
			\draw[Black, very thick] (0,0.5) -- (0,1) -- (2, 1) -- (2,0.5) -- (0,0.5) ;
			
		\end{tikzpicture}	
	\end{minipage}
	\caption{Illustration of the transport plan produced by Algorithm \ref{algo:B-star}, for $k=3$: (left)~for discrete distributions with equally weighted masses; and (right)~for $\mu =$ Uniform$(0,2)$ and $\nu = $ Uniform$(0,4)$. In green are the transport portions such that $x+y\geq 3$; in red are the ones such that $x+y < 3$.}
	\label{fig:E-chen}
\end{figure}

\section{On the order of coupling in the algorithms}
\label{sect:coupling-order}

Algorithms~\ref{algo:A} and~\ref{algo:B} (Section~\ref{sect:EFFECTIVENESS}), as well as \ref{algo:A-eta} and \ref{algo:B-eta} (Section~\ref{sect:NONBINARY}), and~\ref{algo:A-star} and~\ref{algo:B-star} (Appendix~\ref{sect:CHEN}), all proceed by coupling atoms from $\mu$ to $\nu$. %This is an arbitrary choice for their design; for each, there is a mirrored version wherein atoms are coupled from $\mu$ to $\nu$. 
The algorithms iteratively process the atoms in $\mu$ from right to left. 
For each algorithm, there is an alternate version where coupling is from $\nu$ to $\mu$, iteratively processing the atoms in $\nu$ from left to right instead, when $\nu(\R)=\mu(\R)$. The preference order for selecting atoms in $\mu$ is also reversed compared to the ones selecting atoms in $\nu$ in the original algorithms. For Algorithms~\ref{algo:A-eta} and~\ref{algo:B-eta}, the more restrictive atoms in $\eta_1,\ldots, \eta_h$ are now the ones at the smaller locations.
We hereby summarize the heuristics for these alternative algorithms:
\begin{itemize}
    \item[]Algorithm~\ref{algo:A} (resp. \ref{algo:A-star}): Iteratively from left to right for locations of atoms in $\nu$, one couples each to the atom in $\mu$ at the smallest uncoupled location satisfying $y-x\leq k$ (resp. $y+x\leq k$), and otherwise at the completely smallest uncoupled location in $\mu$.  
    \item[]Algorithm~\ref{algo:B} (resp. \ref{algo:B-star}): Iteratively from left to right for locations of atoms in $\nu$, one couples each to the atom in $\mu$ at the largest uncoupled location such that $y-x > k$ (resp. $y+x\geq k$), and otherwise at the largest uncoupled location such that $y\geq x$.  
    \item[]Algorithm \ref{algo:A-eta}: Iteratively from left to right for locations of atoms in $\nu$, one verifies whether the atom in $\mu$ at the smallest uncoupled location such that $y-x\leq k$ may satisfy the partial-treatment constraint. To do so, starting with that atom in $\mu$, we commit atoms from $\eta_1$ to $\eta_h$ at smallest admissible location and check whether that atom in $\eta_h$ is smaller or equal to the atom in $\nu$.  If not, we otherwise couple the atom in $\nu$ to the smallest uncoupled location in $\mu$ and commit atoms at the smallest locations in each of $\eta_1,\ldots,\eta_h$ to it.
    \item[]Algorithm \ref{algo:B-eta}: Iteratively from left to right for locations of atoms in $\nu$, for each uncoupled atom in $\mu$ satisfying $y-x> k$, we verify whether they may satisfy the partial-treatment constraint (using the same criterion as in Algorithm~\ref{algo:A-eta} just above). Select the atom in $\mu$ at the largest location among those able to satisfy the constraints. If there were none, instead for uncoupled atom in $\mu$ satisfying $y\geq x$, we verify whether they may satisfy the partial-treatment constraint, and choose the largest among these.
\end{itemize}
It is straightforward to see that the couplings produced yield optimal probabilities of worthwhile effect, in the sense of Theorems~\ref{th:Q-inf} and~\ref{th:Q-inf-eta}: the alternate algorithms correspond to flipping the atoms' locations' signs and inverting the roles of $\mu$ and $\nu$ before applying the original algorithms. Visually, this transformation corresponds to turning, for example, Figure~\ref{fig:D} upside down and executing the original algorithm on these new distributions. One may easily verify that the executed commands match the heuristics above after reverting the flipping. 

Although also optimal, the couplings produced by the alternate algorithms are not necessarily the same as for the original algorithms. For instance, in Example~\ref{ex:atomic-inf}, the coupling multiset hence produced would be 
\begin{equation*}
    \Gamma = \{(0.5,3.5), (2,4),  (4,5.5), ( 4.5,6 ), (0,7)\},
\end{equation*}
which also yields $Q^{\inf}_3(\mu,\nu) = 1/5$.

\section{Approximation arguments for the sup problem}
\label{sect:approx-sup}

We examine the atomization and discretization approximations for the problem in \eqref{eq:submeasures-sup}. As we stated, applying exactly the same approach as for \eqref{eq:submeasures-inf} would prevent one from obtaining a guarantee of convergence. This is because we cannot constrain the variation of $Q_k^{\sup}$ for very small variations in measures; an analogous statement for $Q_k^{\sup}$ to Theorem~\ref{th:tv} would not hold. In particular, in contrast to $Q_k^{\inf}(\mu,\nu)$, removing mass in $\mu$ may lead to a larger value of $Q_k^{\sup}(\mu,\nu)$, as it may remove a bottleneck in satisfying the objective $y-x>k$.  We give an example to illustrate this. 

\begin{example}
\label{ex:domino}
    Consider $\mu$ and $\nu$ comprising same-size atoms at locations $\mathbf{x} = \mathbf{y} = \{1,2,3,4,5\}$. Suppose $k=0.5$. Obviously, the only coupling multiset satisfying the monotone-response constraint $y\geq x$ is $x_i = y_i$ for every $i\in\{1,\ldots,5\}$. Therefore, $Q_{0.5}^{\sup}(\mu,\nu) = 0$. 
    The mandatory coupling of $x_5$ with $y_5$ is really what prevents any other coupling. 
    Define $\mu^{\prime}$ by removing that atom from $\mu$, leaving only atoms at locations $\mathbf{x}^{\prime} = \{1,2,3,4\}$. Then, the coupling multiset defined by $y_i = x_i +1$ is now valid, as $y_5$ was left vacant from the absence of $x_5$. Now, every atom satisfies the objective $y - x > 0.5$, and thus $Q_{0.5}^{\sup}(\mu^{\prime},\nu) = 4/5 > Q_{0.5}^{\sup}(\mu,\nu)$. The two coupling multisets are illustrated in Figure~\ref{fig:domino}. 
\end{example}

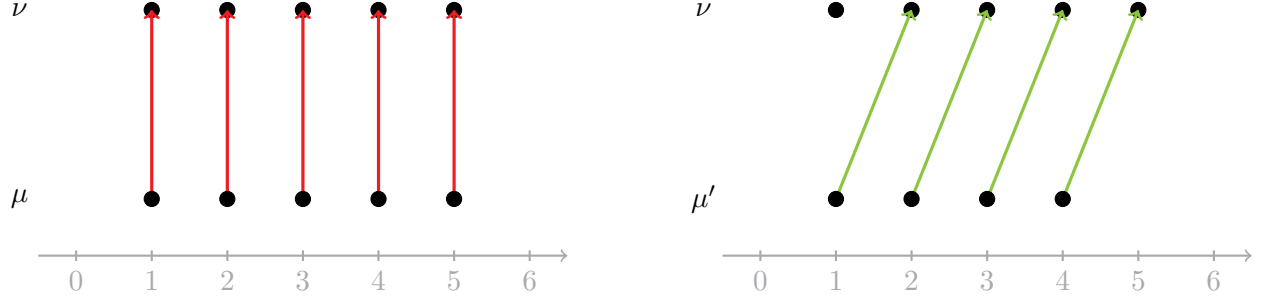
\begin{figure}[tb]
		\centering
		\begin{tikzpicture}
			\draw[Black!40,  ->, thick] (-0.5,0) -- (6.5, 0);
			  \node at (-0.75, 0.75) [rectangle] {$\mu$};
          \node at (-0.75, 3.25) [rectangle] {$\nu$};
			\foreach \x in {0,1,...,6}
			\draw[Black!40,  thick] (\x cm,-2pt) node[below]{\x} -- (\x cm,2pt);
			
			\filldraw (1,3.25) circle (0.1 cm);
			\filldraw (2,3.25) circle (0.1 cm);
			\filldraw (3,3.25) circle (0.1 cm);
			\filldraw (4,3.25) circle (0.1 cm);
			\filldraw (5,3.25) circle (0.1 cm);
			
			\draw[Red, very thick, ->] (4, 0.75) -- (4, 3.25);
			\draw[Red, very thick, ->] (3, 0.75) -- (3, 3.25);
			\draw[Red, very thick, ->] (2, 0.75) -- (2, 3.25);
            \draw[Red, very thick, ->] (1, 0.75) -- (1, 3.25);
            \draw[Red, very thick, ->] (5, 0.75) -- (5, 3.25);
            
			\filldraw (1,0.75) circle (0.1 cm);
			\filldraw (2,0.75) circle (0.1 cm);
			\filldraw (3,0.75) circle (0.1 cm);
			\filldraw (4,0.75) circle (0.1 cm);
			\filldraw (5,0.75) circle (0.1 cm);
		\end{tikzpicture}
        \hfill
        \begin{tikzpicture}
			\draw[Black!40,  ->, thick] (-0.5,0) -- (6.5, 0);
			  \node at (-0.75, 0.75) [rectangle] {$\mu^{\prime}$};
          \node at (-0.75, 3.25) [rectangle] {$\nu$};
			\foreach \x in {0,1,...,6}
			\draw[Black!40,  thick] (\x cm,-2pt) node[below]{\x} -- (\x cm,2pt);
			
			\filldraw (1,3.25) circle (0.1 cm);
			\filldraw (2,3.25) circle (0.1 cm);
			\filldraw (3,3.25) circle (0.1 cm);
			\filldraw (4,3.25) circle (0.1 cm);
			\filldraw (5,3.25) circle (0.1 cm);
			
			\draw[LimeGreen, very thick, ->] (4, 0.75) -- (5, 3.25);
			\draw[LimeGreen, very thick, ->] (3, 0.75) -- (4, 3.25);
			\draw[LimeGreen, very thick, ->] (2, 0.75) -- (3, 3.25);
            \draw[LimeGreen, very thick, ->] (1, 0.75) -- (2, 3.25);
            
			\filldraw (1,0.75) circle (0.1 cm);
			\filldraw (2,0.75) circle (0.1 cm);
			\filldraw (3,0.75) circle (0.1 cm);
			\filldraw (4,0.75) circle (0.1 cm);
		\end{tikzpicture}
	\caption{Illustration of the coupling multisets in Example~\ref{ex:domino}. (Left panel) The only admissible coupling is such that no part satisfies $y-x > 0.5$. (Right panel) With $\mu^{\prime}$ having one less atom to the right, a domino effect takes place and all can satisfy $y-x > 0.5$.}
	\label{fig:domino}
\end{figure}

That such a domino effect can occur for $Q_k^{\sup}(\mu,\nu)$ but not for $Q_k^{\inf}(\mu,\nu)$, may seem surprising, but it actually stems from the way we defined the semi-couplings in \eqref{eq:semicoupling}: we require all mass in $\mu$ to be coupled while some mass in $\nu$ may remain uncoupled, this was to permit cases where $\mu(\mathbb{R}) < \nu(\mathbb{R})$. We chose this semi-coupling direction so that everything agrees nicely for \eqref{eq:submeasures-inf}. However, the forced coupling of all mass in $\mu$ is what causes bothersome portion of mass jeopardizing the satisfaction of the objective $y-x > k$ for remaining masses to be coupled, as in Example~\ref{ex:domino}. We could instead define the semi-couplings of interest in the opposite way, letting
\begin{equation*}
    \mathcal{H}^*(\mu,\nu) = \{ \pi \in \Pi : \mathrm{supp}(\pi)\subseteq \mathbb{H}, P_1(\pi)\leq \mu, P_2(\pi) = \nu\},  
\end{equation*}
that is, we couple all the mass from $\nu$ but some mass in $\mu$ would remain uncoupled. Thus, we have $\mu(\mathbb{R})\geq \nu(\mathbb{R})$ for $\mathcal{H}^*(\mu,\nu)$. The general requirement for $\mathcal{H}^*(\mu,\nu)$ to be non-empty involves a notion of stochastic order symmetric to $\leq_{\rm st}$ defined above: define $\leq_{\rm st}^*$ such that $\mu\leq_{\rm st}^* \nu$ if $F_{\mu} \geq F_{\nu}$, where $F_{\mu}, F_{\nu}$ are the cumulative distribution functions of $\mu$ and $\nu$. Then $\mu\leq_{\rm st}^*\nu$ is necessary and sufficient for $\mathcal{H}^*(\mu,\nu)$ to be non-empty, and it implies $\mu(\mathbb{R})\geq \nu(\mathbb{R})$.     

With these changes, 
similar statements and proofs to Theorem~\ref{th:tv}, Theorems~\ref{th:discrete}--\ref{th:continuous} and Proposition~\ref{prop:atomization}--\ref{prop:convergence} can be developed for $Q_k^{\sup}(\mu,\nu)$ with fairly similar arguments. For that matter,  
atomization and discretization definitions in \eqref{eq:atomization-mu}--\eqref{eq:atomization-nu} and \eqref{eq:discretization-mu}--\eqref{eq:discretization-nu} should be inverted, overestimating mass in $\mu$ and underestimating mass in $\nu$. This would ensure $\widetilde{\mu}_{r} \leq_{\rm st}^* \widetilde{\nu}_r$ and $\overline{\mu}_{\varepsilon} \leq_{\rm st}^* \overline{\nu}_{\varepsilon}$  holds for all $r\in\mathbb{N}$ and $\varepsilon>0$.  
Also, for cases where $\mu(\mathbb{R})> \nu(\mathbb{R})$, the order of coupling discussed in Appendix~\ref{sect:coupling-order}, rather coupling atoms from $\nu$ to $\mu$ left to right, would be required as otherwise one cannot determine which atoms in $\mu$ ought not to be coupled.

% \textcolor{magenta}{\begin{remark}
%     The relation in \eqref{eq:inf-equality} does not hold when replacing ``inf'' by ``sup'' therein, and this is a main reason for an analogous relation to Theorem~\ref{th:tv} for the ``sup'' not holding.
% \end{remark}

Define $\gamma_{k,\mu,\nu}^{\rm B} = \lim_{s\to\infty} \gamma_{k,\overline{\mu}_{s}, \overline{\nu}_{s}}^{\rm B}$ when $\mu$ and $\nu$ are absolutely continuous; it is well-defined, as the pendant argument to Theorem~\ref{th:continuous} does not involve lower semi-continuity. 
Another distinction between \eqref{eq:submeasures-inf} and \eqref{eq:submeasures-sup} arising for the continuous case is the attainability of the solution. 
The first term of $c^*$ in \eqref{eq:nutz-sup} is not lower semicontinuous, and hence $\gamma_{k,\mu,\nu}^{\rm B}$ itself may yield a value inferior to $Q^{\sup}_k(\mu,\nu)$. We illustrate this 
in the following example; see also Example~2.4 of \cite{JR25}.  

\begin{example}
\label{ex:unattained}
    Let $\mu =\mathrm{Uniform}(0,2)$ and $\nu = \mathrm{Uniform}(0,4)$, and thus $\mu\leq_{\rm st}\nu$. We examine $\gamma_{k,\mu,\nu}^{\rm B}$ for $k=3$. We take a very large discretization level $s$ and atomize; we notice that Algorithm~\ref{algo:B} couples, for a given location $1<x\leq 2$, half of the atoms at $x=y$ and the other half at the smallest location with uncoupled atoms, $y=4-x$. This is because none of the atoms at locations $1<x\leq 2$ can possibly satisfy $y-x> 3$. Then, for atoms at locations between 0 and 1, at each location $x$, Algorithm~\ref{algo:B} transports half of them to $y=x+3+2^{-s}$, and the other half at $y=x$. As $s\to\infty$, probability mass belonging to this first half will be transported to exactly $y=x+3$ in $\gamma^{\rm B}_k(\mu,\nu)$.  For any discretization step, this portion of mass consistently meets $y-x>3$, but in the limiting transport plan $\gamma^{\rm B}_k(\mu,\nu)$, by being placed exactly on the boundary $y=x+3$, does not satisfy $y-x>3$; thus, the supremum is not attained. We coloured that portion of mass in purple in Figure~\ref{fig:E-Unif} depicting the transport plan $\gamma_{3,\mu,\nu}^{\rm B}$. We have $Q_{3}^{\sup}(\mu,\nu) = 0.25$ since $Q^{\sup}_{3}(\mu,\nu) = \lim_{s\to \infty}Q^{\sup}_3(\overline{\mu}_{s}, \overline{\nu}_{s})$, although $\gamma_{3,\mu,\nu}^{\rm B}(\{(x,y)\in\mathbb{R}^2: y>x+3\}) = 0$. 
\end{example}

\begin{figure}[tb]
    \centering
 \begin{minipage}{0.75\textwidth}
		\centering
		\begin{tikzpicture}[xscale=1.8]
        \node at (-1.1, 0.75) [rectangle] {$\mu = \mathrm{Uniform}(0,2)$};
          \node at (-1.1, 3.15) [rectangle] {$\nu = \mathrm{Uniform}(0,4)$};
			\draw[Black!40,  ->, thick] (-0.25,0) -- (4.25, 0);
			
			\foreach \x in {0,1,...,4}
			\draw[Black!40, thick] (\x cm,-2pt) node[below]{\x} -- (\x cm,2pt);
			
			\filldraw[Violet!80] (3,3) -- (3,3.25) -- (4, 3.25) -- (4,3) ;
			\filldraw[Red] (0,3) -- (0,3.25) -- (3, 3.25) -- (3,3) ;
			\draw[Black, very thick] (0,3) -- (0,3.25) -- (4, 3.25) -- (4,3) -- (0,3) ;
			
			\draw[Red, very thick, ->] (1.666, 1) -- (1.666, 3);
			\draw[Red, very thick, ->] (1.666, 1) -- (2.333, 3);
			\draw[Red, very thick, ->] (1.333, 1) -- (1.333, 3);
			\draw[Red, very thick, ->] (1.333, 1) -- (2.666, 3);
			\draw[Red, very thick, ->] (0.666, 1) -- (0.666, 3);
			\draw[Violet!80, very thick, ->] (0.666, 1) -- (3.666, 3);
			\draw[Red, very thick, ->] (0.333, 1) -- (0.333, 3);
			\draw[Violet!80, very thick, ->] (0.333,1) -- (3.333,3); 
			
			\filldraw[Violet!80] (0,0.5) -- (0,0.75) -- (1, 0.75) -- (1,0.5) ;
			\filldraw[Red] (1,0.5) -- (1,1) -- (2, 1) -- (2,0.5) ;
			\filldraw[Red] (0,0.75) -- (0,1) -- (1, 1) -- (1,0.75) ;
			\draw[Black, very thick] (0,0.5) -- (0,1) -- (2, 1) -- (2,0.5) -- (0,0.5) ;
			
		\end{tikzpicture}	
	\end{minipage}
    \caption{Transport plan $\gamma_{3,\mu,\nu}^{\rm B}$ from Example~\ref{ex:unattained}. }
    \label{fig:E-Unif}
\end{figure}
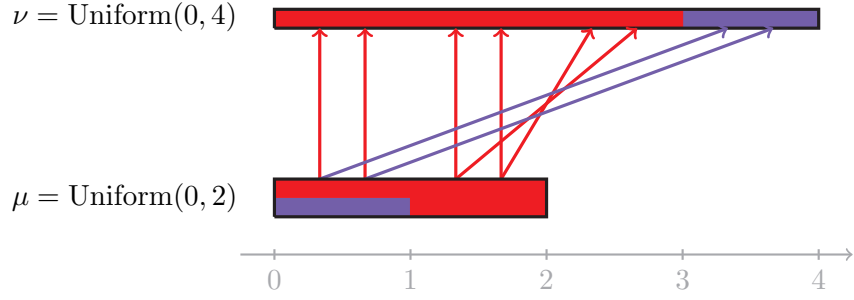

\section{Duality}
\label{sect:duality}
In the main text, we do not appeal to duality because we solve the primal problems directly; we briefly discuss duality for measures with equal total mass in this appendix. 
For general measures $\mu,\nu\in\mathcal{M}$, assuming $\mu\leq_{\rm st}\nu$ and $\mu(\mathbb{R}) = \nu(\mathbb{R})$, problems \eqref{eq:nutz-inf}--\eqref{eq:nutz-sup} are rewritten as
\begin{align}
Q^{\inf}_k(\mu,\nu) &= \inf\left\{\int c \, \mathrm{d}\pi : \pi \text{ has marginals }\mu \text{ and }\nu \right\}, \label{eq:OT-inf}
\\   Q^{\sup}_k(\mu,\nu) &=\mu(\R)- \inf\left\{\int c^* \, \mathrm{d}\pi : \pi \text{ has marginals }\mu \text{ and }\nu \right\}, \label{eq:OT-sup} 
\end{align}
where $c$ and $c^*$ are the same as in \eqref{eq:nutz-inf}--\eqref{eq:nutz-sup}.
These cost functions take values in the extended reals and both belong to the closure of lower semicontinuous functions on $\mathbb{R}^2$.  By \citet[Theorem 2.3.1(c), cf.~Definition 2.2.2]{RR06}, there is duality. This means that problems \eqref{eq:OT-inf}--\eqref{eq:OT-sup} are equivalent to 
\begin{align}
Q^{\inf}_{k}(\mu,\nu) &= \sup_{\phi, \psi}\left\{\int \phi \,\mathrm{d}\mu + \int \psi\,\mathrm{d}\nu : \phi(x) + \psi(y) \leq c(x,y),\, x,y\in\mathbb{R}\right\};\label{eq:dual-inf}\\
Q^{\sup}_k(\mu,\nu) &= \mu(\R) - \sup_{\phi, \psi}\left\{\int \phi \,\mathrm{d}\mu + \int \psi\,\mathrm{d}\nu  : \phi(x) + \psi(y) \leq c^*(x,y),\, x,y\in\mathbb{R}\right\},\label{eq:dual-sup}
\end{align}
where the suprema are taken over all integrable functions; see also \citet[Section~3]{K14}.\footnote{\cite{K14} also discusses some non-sharp bounds for \eqref{eq:OT-inf}--\eqref{eq:OT-sup}.} While many optimal transport problems require to leverage duality to be solvable, the codomain of the cost functions being $\{0,1,\infty\}$ makes optimization in the primal problems \eqref{eq:OT-inf}--\eqref{eq:OT-sup} actually simpler than in the dual problems above. In fact, it is this limited codomain that underlies the ``no-nuisance'' principle discussed in Section~\ref{sect:EFFECTIVENESS}, on which Algorithms~\ref{algo:A} and~\ref{algo:B} are predicated.

Suppose that $\mu$ and $\nu$ are probability measures both supported on $n$ values, $n\in\mathbb{N}$, say respectively $x_1,\ldots,x_n$ and $y_1,\ldots,y_n$, and have probability mass functions $p_{\mu}$ and $p_{\nu}$. Then, \eqref{eq:dual-inf} and \eqref{eq:dual-sup} become the problems
\begin{align*}
    \max_{\phi,\psi}\left\{\sum_{j=1}^n p_{\mu}(x_j) \phi(x_j) + \sum_{j=1}^n p_{\nu}(y_j) \psi(y_j) : \phi(x) + \psi(y) \leq \id_{\{y-x>k\}} + b\id_{\{y<x\}},\, x,y\in\mathbb{R}\right\};\\
    1 - \max_{\phi,\psi}\left\{\sum_{j=1}^n p_{\mu}(x_j) \phi(x_j) + \sum_{j=1}^n p_{\nu}(y_j)\psi(y_j)  : \phi(x) + \psi(y) \leq \id_{\{y-x\leq k\}} + b\id_{\{y<x\}},\, x,y\in\mathbb{R}\right\},
\end{align*}
which are linear programs with $2n$ variables and $n^2$ constraints, with $b$ being a finite number large enough to enforce the monotone-response constraint. 
This is the method presented in \cite{JLS23} to solve problems \eqref{eq:dual-inf} and \eqref{eq:dual-sup} for discrete probability measures of bounded support. (They also solve for continuous probability measures with bounded support by a discretization approximation.)
The approach in this present paper, solving the primal problem directly, is more computationally efficient, especially if $\mu$ and $\nu$ are atomic with the same size. This ability to solve the primal problem directly is mainly grounded in that the cost functions $c$ and $c^*$ in \eqref{eq:OT-inf}--\eqref{eq:OT-sup} take only three values, 0, 1, $\infty$, and the latter corresponds to enforcing an additional constraint. \cite{JR25} discuss an interpretation of Kantorovich problems with cost functions taking values $\{0,1\}$ in terms of Strassen's theorem, and further discuss duality for such problems in more general topological spaces.

\end{document}